\documentclass[authorversion,acmsmall,nonacm,screen]{acmart}

\setcopyright{none}

\usepackage[scaled=0.8]{beramono}
\usepackage{listings,array,varwidth}
\usepackage{mathpartir}
\usepackage{bcprules}
\usepackage{mathtools}
\usepackage{makecell}
\usepackage{diagbox}
\usepackage{float}
\usepackage{varioref}
\usepackage{xfrac}
\usepackage{enumitem}
\usepackage[skip=2.5pt]{caption}
\usepackage{slashbox}
\usepackage{cleveref}
\usepackage{tikz}
\usepackage{cancel}
\usepackage{bm}
\usepackage{mathtools}
\usepackage{lipsum}
\usepackage{wrapfig}
\usepackage{soul}
\usepackage{xspace}
\usepackage{mdframed}
\usepackage[normalem]{ulem}
\usepackage{cancel}
\usepackage{scalerel}
\usepackage{multirow}
\usepackage[export]{adjustbox}
\usepackage{xifthen}
\usepackage{tikz} %
\usepackage{xcolor,colortbl}
\usepackage{tikz-cd}

\usetikzlibrary{shadows}
\usetikzlibrary{shadows.blur}
\usetikzlibrary{decorations.pathreplacing} %
\tikzcdset{arrow style=tikz,
           diagrams={>=stealth}
         }    %
\tikzcdset{crossing over clearance=.75ex}
\tikzset{hardd/.style={teal,thick}} %
\tikzset{softd/.style={teal, dashed,thick}} %
\tikzset{datad/.style={black,densely dotted}} %
\tikzset{alld/.style={black,thick}} %
\tikzset{expnode/.style={
    asymmetrical rectangle,rounded corners,draw,fill=white,inner sep=2.5pt,
    font=\footnotesize,blur shadow={shadow xshift=0ex,shadow yshift=-.2ex}
  }} %
\tikzset{blknode/.style={
    asymmetrical rectangle,sharp corners,draw,violet,fill=violet!10,
    font=\footnotesize %
  }} %
\tikzset{dummynode/.style={fill=none,draw=none,asymmetrical rectangle,general shadow/.style=,blur shadow={shadow opacity=0}}} %
\tikzset{nmute/.style={opacity=.4,blur shadow={shadow opacity=.4}}} %
\tikzset{expn/.style={
    circle,fill=white,draw,minimum size=5pt,inner sep=0pt,outer sep=0pt,
    font=\footnotesize,blur shadow={shadow xshift=0ex,shadow yshift=-.2ex},
    label={[font=\footnotesize]left:#1}
  }} %
\tikzset{expn/.default={}}
\tikzset{mexpn/.style={opacity=.4,
    circle,fill=white,draw,minimum size=5pt,inner sep=0pt,outer sep=0pt,
    font=\footnotesize,blur shadow={shadow xshift=0ex,shadow yshift=-.2ex,shadow opacity=.4},
    label={[font=\footnotesize,opacity=.4]left:#1}
  }} %
\tikzset{mexpn/.default={}}
\tikzset{rmexpn/.style={opacity=.4,
    circle,fill=white,draw,minimum size=5pt,inner sep=0pt,outer sep=0pt,
    font=\footnotesize,blur shadow={shadow xshift=0ex,shadow yshift=-.2ex,shadow opacity=.4},
    label={[font=\footnotesize,opacity=.4]right:#1}
  }} %
\tikzset{rmexpn/.default={}}
\tikzset{mexpndummy/.style={opacity=.4,
    circle,fill=white,minimum size=5pt,inner sep=0pt,outer sep=0pt,
    label={[font=\footnotesize,opacity=.4]left:#1}
  }} %
\tikzset{mexpndummy/.default={}}

\usepackage{quiver}

\colorlet{lightred}{red!30} %

\usetikzlibrary{arrows}

\makeatletter %
\def\arcr{\@arraycr}
\makeatother

\newcommand{\showDOI}[1]{\unskip}

\newtheorem{innercustomgeneric}{\customgenericname}
\providecommand{\customgenericname}{}
\newcommand{\newcustomtheorem}[2]{%
  \newenvironment{#1}[1]
  {%
   \renewcommand\customgenericname{#2}%
   \renewcommand\theinnercustomgeneric{##1}%
   \innercustomgeneric
  }
  {\endinnercustomgeneric}
}
\newcustomtheorem{re-definition}{Definition}
\newcustomtheorem{re-lemma}{Lemma}

\newcommand{\figref}[1]{Fig.~\ref{#1}}
\newcommand{\secref}[1]{Sec.~\ref{#1}} %
\newcommand{\thmref}[1]{Theorem~\ref{#1}}
\newcommand{\lemref}[1]{Lemma~\ref{#1}}
\newcommand{\colref}[1]{Corollary~\ref{#1}}

\newcommand{\Specsharp}{%
	{\settoheight{\dimen0}{C}Spec\kern-.05em \resizebox{!}{\dimen0}{\raisebox{\depth}{\#}}}}
\newcommand{\Csharp}{%
	{\settoheight{\dimen0}{C}C\kern-.05em \resizebox{!}{\dimen0}{\raisebox{\depth}{\#}}}}

\newcommand{\fun}[1]{\operatorname{#1}}
\newcommand{\DOM}{\fun{dom}}
\newcommand{\CODOM}{\fun{cod}}

\definecolor{blue-violet}{rgb}{0.54, 0.17, 0.89}
\definecolor{ccomment}{HTML}{006400}
\definecolor{depmap}{HTML}{00007B}
\definecolor{dark-cyan}{HTML}{135579}
\definecolor{magenta}{HTML}{a8264f}

\lstdefinelanguage{Neutral}%
{morekeywords={abstract,%
  case,catch,char,class,%
  def,else,extends,final,finally,for,%
  if,import,implicit,%
  match,module,%
  new,null,%
  object,override,%
  package,private,protected,public,%
  for,public,return,super,%
  this,throw,trait,try,type,%
  val,var,%
  with,while,%
  yield,%
  let,end,%
	in,fun,alloc,inc%
  },%
  mathescape=true,%
  sensitive,%
  keywordstyle={\color{black}\bf\ttfamily},%
  commentstyle=\color{OliveGreen},%
  escapebegin=\color{OliveGreen},
  morecomment=[l]//,%
  morecomment=[s]{/*}{*/},%
  morecomment=[s][\color{darkgray}]{@}{\ },%
  morestring=[b]",%
  morestring=[b]',%
  showstringspaces=false%
}[keywords,comments,strings]%

\lstdefinelanguage{OOPSLA21}%
{morekeywords={abstract,%
  case,catch,char,class,%
  def,else,extends,final,finally,for,%
  if,import,implicit,%
  match,module,%
  new,null,%
  object,override,%
  package,private,protected,public,%
  for,public,return,super,%
  this,throw,trait,try,type,%
  val,var,%
  with,while,%
  yield,%
  let,end,%
	in,fun,alloc,inc%
  },%
  mathescape=true,%
  sensitive,%
  keywordstyle={\color{magenta}\bf\ttfamily},%
  commentstyle=\color{magenta},%
  escapebegin=\color{magenta},
  morecomment=[l]//,%
  morecomment=[s]{/*}{*/},%
  morecomment=[s][\color{magenta}]{@}{\ },%
  morestring=[b]",%
  morestring=[b]',%
  showstringspaces=false%
}[keywords,comments,strings]%

\lstdefinelanguage{PolyRT}%
{morekeywords={abstract,%
  case,catch,char,class,%
  def,else,extends,final,finally,for,%
  if,import,implicit,%
  match,module,%
  new,null,%
  object,override,%
  package,private,protected,public,%
  for,public,return,super,%
  this,throw,trait,try,type,%
  val,var,%
  with,while,%
  yield,%
  let,end,%
	in,fun,alloc,inc%
  },%
  mathescape=true,%
  sensitive,%
  keywordstyle={\color{dark-cyan}\bf\ttfamily},%
  commentstyle=\color{dark-cyan},%
  escapebegin=\color{dark-cyan},%
  morecomment=[l]//,%
  morecomment=[s]{/*}{*/},%
  morecomment=[s][\color{dark-cyan}]{@}{\ },%
  morestring=[b]",%
  morestring=[b]',%
  showstringspaces=false%
}[keywords,comments,strings]%

\lstdefinelanguage{Scala}%
{morekeywords={abstract,%
    case,catch,char,class,%
    def,else,extends,final,finally,for,%
    if,import,implicit,%
    match,module,%
    new,null,%
    object,override,%
    package,private,protected,public,%
    for,public,return,super,%
    this,throw,trait,try,type,%
    val,var,%
    with,while,%
    yield,%
    let, end, @track,%
    ref, Ref, move, swap,%
    println,%
    until%
  },%
  numbers=none, %
  mathescape=true,%
  sensitive,%
  commentstyle=\color{depmap},
  moredelim=**[is][\color{red}]{<}{>},%
  morecomment=[l]//,%
  morecomment=[l]/,%
  morecomment=[s][\color{ccomment}]{/*}{*/},%
  morestring=[b]",%
  morestring=[b]',%
  showstringspaces=false%
}[keywords,comments,strings]%

\newcommand{\langg}{\irlang}

\newcommand{\ts}[1][]{\ensuremath{\ifthenelse{\isempty{#1}}{\,\vdash\,}{\,\vdash_{#1}\,}}}
\newcommand{\EMM}{\ensuremath{\scalerel*{\mathsf{M}}{a}}}
\newcommand{\tsM}{\ts[\EMM]}
\newcommand{\GEE}{\ensuremath{\scalerel*{\mathsf{G}}{a}}}

\newcommand{\oldlang}{\ensuremath{\lambda^{*}}\xspace}
\newcommand{\directlang}{\ensuremath{\lambda^{*}_{\varepsilon}}\xspace}
\newcommand{\maybelang}{\directlang}

\DeclareRobustCommand{\mnflang}{\ensuremath{\lambda^{*}_{\EMM}}\xspace}
\DeclareRobustCommand{\irlang}{\ensuremath{\lambda^{*}_{\GEE}}\xspace}
\newcommand{\langa}{\irlang}
\newcommand{\Type}[1]{\ensuremath{\mathsf{#1}}}
\newcommand{\Var}{\Type{Var}}

\newcommand{\TRef}{\Type{Ref}}

\newcommand{\tref}{\text{\textbf{\textsf{ref}}}}
\newcommand{\tlet}{\text{\textbf{\textsf{let}}}}
\newcommand{\tlets}{\ensuremath{\text{\textbf{\textsf{let}}}_{\mathsf{s}}}}

\newcommand{\tin}{\text{\textbf{\textsf{in}}}}
\newcommand{\TUnit}{\Type{Unit}}
\newcommand{\tunit}{\mathsf{unit}}
\newcommand{\Loc}{\Type{Loc}}
\newcommand{\V}[1]{\ensuremath{\mathtt{#1}}}

\newcommand{\ty}[2][]{\ensuremath{\ifthenelse{\isempty{#1}}{#2}{#2^{\,#1}}}}

\newcommand{\flt}{\ensuremath{\varphi}}

\newcommand{\cx}[2][]{\ensuremath{\ifthenelse{\isempty{#1}}{#2}{#2^{\,#1}}}}
\newcommand{\csx}[3][]{\ensuremath{\ifthenelse{\isempty{#1}}{#3\mid #2}{{\color{gray!50}[}#3\mid#2{\color{gray!50}]}^{\,#1}}}}
\providecommand{\G}{G} %
\renewcommand{\G}[1][]{\cx[#1]{\Gamma}}
\newcommand{\GS}[1][]{\csx[#1]{\Gamma}{\Sigma}}
\newcommand{\cdsx}[4][]{\ensuremath{\ifthenelse{\isempty{#1}}{#3\mid #2\has #4}{{\color{gray!50}[}#3\mid#2{\color{gray!50}]}^{\,#1}\has #4 }}}
\newcommand{\GSD}[1][]{\cdsx[#1]{\Gamma}{\Sigma}{\DELTA}}

\newcommand{\qbot}{\ensuremath{\varnothing}}
\newcommand{\qfresh}{\ensuremath{\vardiamondsuit}}
\newcommand{\subq}{\ensuremath{\subseteq}}
\newcommand{\qlub}{\ensuremath{\cup}}
\newcommand{\qglb}{\ensuremath{\cap}}
\newsavebox{\SMALLSTAR}
\savebox{\SMALLSTAR}{\(\raisebox{.25ex}{\(\qfresh\)}\)}
\newcommand{\starred}[1]{}
\newsavebox{\OVRLP}
\savebox{\OVRLP}{$\raisebox{.37ex}[0pt][0pt]{$\mathrlap{\hspace{.415ex}\scaleobj{.5}{\vardiamondsuit}}$}\cap$}
\newcommand{\overlap}{\cap}

\newcommand{\qsat}[1]{\ensuremath{#1\mathord{*}}}
\newcommand{\WF}[1]{\ensuremath{#1\ \mathsf{ok}}}
\newcommand{\reaches}{\ensuremath{\mathrel{\leadsto}}}

\newcommand{\BOX}[1]{\fbox{$\strut #1$}}

\newcommand{\FV}{\ensuremath{\operatorname{fv}}}

\newcommand{\vgap}{\vspace{7pt}}

\newcommand{\rulename}[1]{(\textsc{#1})} %

\colorlet{mute}{teal}
\colorlet{eff}{magenta}
\newcommand{\mute}[1]{{\color{mute}#1}}

\definecolor{light-gray}{gray}{0.92}
\definecolor{dark-gray}{gray}{0.5}
\definecolor{light-yellow}{HTML}{faeeaa}
\definecolor{light-pink}{HTML}{faf2f7}

\newcommand{\HLBox}[2][teal!12]{\ensuremath{\mathchoice%
  {\setlength{\fboxsep}{.5ex}\colorbox{#1}{$\displaystyle#2$}}%
  {\setlength{\fboxsep}{.5ex}\colorbox{#1}{$\textstyle#2$}}%
  {\setlength{\fboxsep}{.5ex}\colorbox{#1}{$\scriptstyle#2$}}%
  {\setlength{\fboxsep}{.5ex}\colorbox{#1}{$\scriptscriptstyle#2$}}}}%

\newcommand{\HLCode}[2][teal!12]{\setlength{\fboxsep}{.5ex}\colorbox{#1}{#2}}

\colorlet{lightred}{red!30}

\newcommand{\FX}[1]{\ensuremath{{\color{eff}#1}}}
\newcommand{\EPS}[1][]{\ifthenelse{\isempty{#1}}{\FX{\bm{\varepsilon}}}{\FX{\bm{\varepsilon_{#1}}}}}
\newcommand{\EPSS}[1][]{\ifthenelse{\isempty{#1}}{\FX{\qsat{\bm{\varepsilon}}}}{\FX{\qsat{\bm{\varepsilon_{#1}}}}}}
\newcommand{\EPSPR}[1][]{\ifthenelse{\isempty{#1}}{\FX{\bm{\varepsilon'}}}{\FX{\bm{\varepsilon'_{#1}}}}}
\newcommand{\EPSSPR}[1][]{\ifthenelse{\isempty{#1}}{\FX{\qsat{\bm{\varepsilon'}}}}{\FX{\qsat{\bm{\varepsilon'_{#1}}}}}}
\newcommand{\PURE}{\FX{\boldsymbol{\varnothing}}}
\newcommand{\EFFSEQ}{\ensuremath{\mathbin{\FX{\boldsymbol{\rhd}}}}}

\newcommand{\DELTA}{\ensuremath{\mute{\Delta}}}
\newcommand{\DELTAP}{\ensuremath{\mute{\Delta'}}}
\newcommand{\NODEP}{\ensuremath{\mute{\varnothing}}}
\newcommand{\DEP}[1][]{\ensuremath{\ifthenelse{\isempty{#1}}{\mute{\delta}}{\mute{\delta_{#1}}}}}

\newcommand{\HDEP}[1][]{\ensuremath{\ifthenelse{\isempty{#1}}{\mute{\mathsf{h}}}{\mute{\mathsf{h}_{#1}}}}}
\newcommand{\SDEP}[1][]{\ensuremath{\ifthenelse{\isempty{#1}}{\mute{\mathsf{s}}}{\mute{\mathsf{s}_{#1}}}}}
\newcommand{\HHD}{\ensuremath{\mute{\mathsf{H}}}}
\newcommand{\SSD}{\ensuremath{\mute{\mathsf{S}}}}
\newcommand{\RD}{\ensuremath{\mathsf{r}}}
\newcommand{\WR}{\ensuremath{\mathsf{w}}}

\newcommand{\yields}[2][mute]{\ensuremath{{\color{#1}\ \leadsto #2}}}

\newcommand{\MAPSTO}{\ensuremath{\hstretch{.5}{\mapsto}}}
\newcommand{\pointsto}[1]{\ensuremath{\MAPSTO#1}}

\newcommand{\keyw}[1]{\textsf{\textbf{#1}}\xspace}

\def\Let#1#2#3{\keyw{let}\ {#1} = {#2}\ \keyw{in}\ {#3}}

\newcommand{\Typ}[1]{\textsf{#1}\xspace}

\newcommand{\has}{\ensuremath{\mathbin{\bullet}}}

\newcommand{\erased}[1]{\ensuremath{\ulcorner #1\urcorner}}

\newcommand{\trackset}[1]{^{\texttt{\{#1\}}}}
\newcommand{\tracksetl}[1]{^{\tt{\{#1\}}}}

\newcommand{\hole}[1]{\ensuremath{[\,#1\,]}}
\newcommand{\CX}[3][black]{\ensuremath{{\color{#1}#2\ifthenelse{\isempty{#3}}{}{\hole{{\color{black}#3}}}}}}

\newcommand{\seq}[1]{\ensuremath{\overline{#1}}}

\newcommand{\bfparagraph}[1]{\paragraph{\textbf{#1}}}

\newcommand{\redsv}{\ensuremath{\mathrel{\longrightarrow_\mathbf{sv}}}}
\newcommand{\redg}{\ensuremath{\mathrel{\longrightarrow_{\GEE}}}}
\newcommand{\redv}{\ensuremath{\mathrel{\longrightarrow_\mathbf{v}}}}

\newcommand\restr[2]{{%
  \left.\kern-\nulldelimiterspace %
  #1 %
  \vphantom{\big|} %
  \right|_{#2} %
}}

\lstdefinelanguage{DOT}%
{morekeywords={val,new},%
  sensitive,%
  morecomment=[l]//,%
  morecomment=[s]{/*}{*/},%
  morestring=[b]",%
  morestring=[b]',%
  showstringspaces=false%
}[keywords,comments,strings]%

\newlength{\trulemargin}
\newlength{\trulewidth}
\newlength{\srulewidth}
\newenvironment{trules}{$\vspace{0.5em}\ba{p{\trulemargin}@{~}p{\trulewidth}@{~}p{\trulemargin}}}{\ea$}
\newenvironment{srules}{$\vspace{0.5em}\ba{p{\trulemargin}@{~}p{\srulewidth}}}{\ea$}

\newcommand{\ba}{\begin{array}}
\newcommand{\ea}{\end{array}}

\newcommand{\ei}{\end{array}}
\newcommand{\bcases}{\left\{\begin{array}{ll}}
\newcommand{\ecases}{\end{array}\right.}

\newcommand{\eg}{{\em e.g.}\xspace}
\newcommand{\ie}{{\em i.e.}\xspace}

\newcommand{\dom}{\mbox{\sl dom}}

\newcommand{\judgement}[2]{{\textsf{\textbf{#1}}} \hfill #2}

\newcommand{\synbracket}[1]{[\![{#1}]\!]}

\newcommand{\DEF}{\stackrel{{\rm def}}{=}}

\newcommand{\equiva}{\ensuremath{\mathrel{\approx_{\text{ctx}}}}}
\newcommand{\equivlog}{\ensuremath{\mathrel{\approx_{\text{log}}}}}
\newcommand{\carrow}{\Rrightarrow}

\newcommand{\fltp}{\ensuremath{\varphi'}}

\newcommand{\mredv}[1]{\ensuremath{\mathrel{\longrightarrow^{{#1}}_\mathbf{v}}}}

\newcommand{\extends}{\ensuremath{\ensuremath{;}}}

\newcommand{\restrict}[2]{({#1}\!\!\downarrow\!\!{#2})}
\newcommand{\restricts}[2]{{#1}\!\downarrow\!{#2}}

\newcommand{\GP}[1][]{\cx[#1]{\Gamma'}}

\newcommand{\DGMV}[1]{\mathcal{V}\synbracket{{#1}}^{\gamma}}
\newcommand{\DGMVE}[2]{\mathcal{V}\synbracket{{#1}}^{{#2}}}

\newcommand{\DGMt}[2]{\mathcal{E}\synbracket{#2}^{\gamma}_{{#1}}}

\newcommand{\DEPS}[3]{{#1}\hookrightarrow^{#2}{#3}}

\newcommand{\csxs}[2][]{\ensuremath{\ifthenelse{\isempty{#1}}{\mid #2}{{\color{gray!50}[}#2{\color{gray!50}]}^{\,#1}}}}

\newcommand{\State}[2]{\ensuremath{{#1}, \, {#2}}}

\newcommand*{\Scale}[2][4]{\scalebox{#1}{$#2$}}%

\newcommand{\UTC}[3]{\ensuremath{(\State{#1}{#2}, {#3})}}
\newcommand{\UG}[1]{G\synbracket{{#1}}}

\newcommand{\WFRS}[3]{({#1}, {#2}):{#3}}

\newcommand{\W}{\text{W}}

\newcommand{\DVTC}[3]{\ensuremath{({#1},{#2}, {#3})}}

\newcommand{\dvalocss}[1]{\text{locs}(\ensuremath{{#1}})}
\newcommand{\dvarslocs}[1]{\text{locs}(\ensuremath{{#1}})}

\newcommand{\dvalq}[3]{\ensuremath{#2 \reaches^{\Scale[0.7]{{#1}}}  {#3}}}

\let\emptyset\varnothing

\newlength{\myskip}
\makeatletter%
\begin{document}

\title{Graph IRs for Impure Higher-Order Languages (Technical Report)}

\author{Oliver Bra\v{c}evac}
\affiliation{
  \institution{Purdue University}            %
  \state{IN}
  \country{USA}                    %
}
\email{oliver@galois.com}          %

\author{Guannan Wei}
\affiliation{
  \institution{Purdue University}            %
  \city{West Lafayette}
  \state{IN}
  \country{USA}                    %
}
\email{guannanwei@purdue.edu}          %

\author{Songlin Jia}
\affiliation{
  \institution{Purdue University}            %
  \city{West Lafayette}
  \state{IN}
  \country{USA}                    %
}
\email{jia137@purdue.edu}          %

\author{Supun Abeysinghe}
\affiliation{
  \institution{Purdue University}            %
  \country{USA}                    %
}
\email{tabeysin@purdue.edu}          %

\author{Yuxuan Jiang}
\affiliation{
  \institution{Purdue University}            %
  \country{USA}                    %
}
\email{jiang700@purdue.edu}          %

\author{Yuyan Bao}
\affiliation{
  \institution{Augusta University}            %
  \country{USA}                    %
}
\email{yubao@augusta.edu}          %

\author{Tiark Rompf}
\affiliation{
  \institution{Purdue University}            %
  \country{USA}                    %
}
\email{tiark@purdue.edu}          %
\authorsaddresses{}

\lstMakeShortInline[keywordstyle=,%
  flexiblecolumns=false,%
  language=Scala,
  basewidth={0.56em, 0.52em},%
  mathescape=false,%
  basicstyle=\ttfamily]@

\begin{abstract}

This is a companion report for the OOPSLA 2023 paper of the same title, presenting a
detailed end-to-end account of the \langg{} graph IR, at a level of detail beyond a regular
conference paper. Our first concern is adequacy and soundness of \langg{}, which we derive from a
direct-style imperative functional language (a variant of Bao et al.'s
\oldlang-calculus with reachability types and a simple effect system) by a series of type-preserving
translations into a calculus in monadic normalform (MNF). Static reachability types and effects
entirely inform \langg{}'s dependency synthesis. We argue for its adequacy by proving its functional
properties along with \emph{dependency safety} via progress and preservation lemmas with respect to
a notion of call-by-value (CBV) reduction that checks the observed order of effects.

Our second concern is establishing the correctness of \langg{}'s equational rules that drive
compiler optimizations (\eg, DCE, \(\lambda\)-hoisting, etc.), by proving contextual equivalence
using logical relations. A key insight is that the functional properties of dependency synthesis
permit a logical relation on \langg{} in MNF in terms of previously developed logical relations for
the direct-style \oldlang-calculus.

Finally, we also include a longer version of the conference paper's section
on code generation and code motion for \langg{} as implemented in Scala~LMS. \end{abstract}

\maketitle

\tableofcontents
\listoffigures

\begin{figure}[t]
  \adjustbox{scale=0.8,center}{%
    \begin{tikzcd}[sep=small]
      &&&&&& {\begin{smallmatrix} \text{Preservation}\\\text{of Separation}\end{smallmatrix}}\\
      & {\begin{smallmatrix}\text{Direct Style \maybelang}  \\ \text{\Cref{sec:directstyle}}\end{smallmatrix}} && {\begin{smallmatrix}\text{Monadic \mnflang} \\ \text{\Cref{sec:monadic}}\end{smallmatrix}} && {\begin{smallmatrix}\text{Graph IR \irlang} \\ \text{\Cref{sec:monad-norm-with}}\end{smallmatrix}} \\
      \\
      {\begin{smallmatrix} \text{CBV Reduction}\end{smallmatrix}} && {\begin{smallmatrix} \text{Store-Allocated}\\\text{CBV Reduction}\\\text{\Cref{sec:directstyle_letv}}\end{smallmatrix}} && {\begin{smallmatrix} \text{Dependency-checking}\\\textsf{CBV Reduction}\end{smallmatrix}} && {\begin{smallmatrix} \text{Dependency}\\\text{Safety}\end{smallmatrix}}
      \arrow["{\text{Thm.~\ref{thm:progress} \& \ref{thm:soundness}}}"', very thick, color={rgb,255:red,214;green,92;blue,92}, from=2-2, to=4-1]
      \arrow["{\text{Lem.~\ref{lem:type_preservation:translation_backwards}}}"', dotted, curve={height=-12pt}, from=2-4, to=2-2]
      \arrow[color={rgb,255:red,214;green,92;blue,92}, very thick, dashed, from=2-4, to=4-3]
      \arrow["{\text{Lem.~\ref{lem:synth-totality}}}", curve={height=-12pt}, from=2-4, to=2-6]
      \arrow[dotted, double, from=4-5, to=4-3]
      \arrow["{\text{Lem.~\ref{lem:type_preservation:translation}}}", curve={height=-12pt}, from=2-2, to=2-4]
      \arrow["{\text{Lem.~\ref{lem:synth-totality}}}"', curve={height=-12pt}, from=2-6, to=2-4, dotted]
      \arrow["{\text{Thm.~\ref{thm:soundness_letv}}}"', very thick, color={rgb,255:red,214;green,92;blue,92}, from=2-2, to=4-3]
      \arrow[shift left=2, curve={height=-36pt}, dashed, from=2-2, to=2-6]
      \arrow["{\text{Thm.~\ref{thm:graphir:progress} \& \ref{thm:graphir:preservation}}}", very thick, color={rgb,255:red,214;green,92;blue,92}, from=2-6, to=4-5]
      \arrow["{\text{Cor.~\ref{coro:graphir:dep-safety}}}", from=2-6, to=4-7, dashed]
      \arrow["{\text{Cor.~\ref{thm:graphir:preservation_of_separation}}}", from=2-6, to=1-7, dashed]
      \arrow[from=4-3,to=4-1,double,tail reversed]
    \end{tikzcd}
  }
  \caption[Overview of the calculi and their metatheory in this report.]{Overview of the calculi and their metatheory in this report. Thick red arrows indicate type soundness proofs w.r.t.\ a reduction semantics. Black arrows indicate type preservation proofs between calculi.
    Dashed arrows indicate corollaries, and dotted arrows indicate an embedding or erasure. Double arrows indicate an equivalence.}
  \label{fig:overview}
\end{figure}
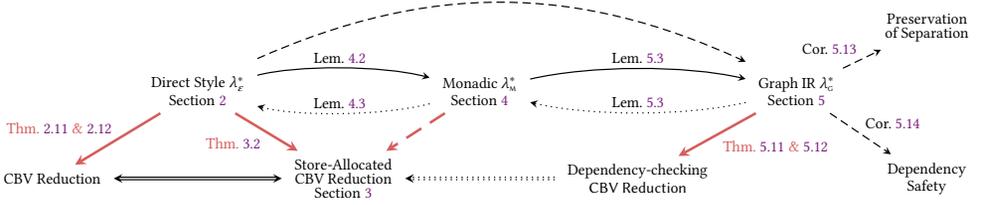
\vspace{-20pt}
\section{Introduction}
This document complements the paper ``Graph IRs for Impure Higher-Order Languages''~\cite{oopsla23} and
presents a detailed end-to-end account of the \langg{} graph IR,  step by step refinining a series of calculi
and the formal development of their metatheory, in a way that goes beyond the space available in a
conference paper. The process begins from the direct-style \directlang-calculus (a variant of
\citet{DBLP:journals/pacmpl/BaoWBJHR21}'s \oldlang-calculus with a simple effect system), proceeds
with \mnflang, a version in monadic normal form (MNF), and ends in the \irlang-calculus, the typed
graph IR in monadic normal form (MNF)\footnote{Monadic normal form~\cite{DBLP:conf/popl/HatcliffD94} is a generalization of ANF~\cite{DBLP:conf/pldi/FlanaganSDF93} and related let-normal forms, where let bindings permit nesting.} with effect dependencies. Each stage constitutes a straightforward refinement which
is provably type/effect/qualifier preserving, and is provably type-safe with respect to a notion of
call-by-value (CBV) reduction, established by standard progress and preservation lemmas.

The key end-to-end safety guarantees are visualized in \Cref{fig:overview}. The result is a sequence of
corollaries establishing key properties, namely that the direct-style language can be translated in
a type-preserving and dependency-synthesizing manner into \langg{}, and that \langg{} is sound
with respect to a dependency-checking call-by-value operational semantics. It follows that synthesized dependencies
correctly reflect the execution order of effects, and \citet{DBLP:journals/pacmpl/BaoWBJHR21}'s
preservation of separation holds for the graph IR, a memory property guaranteeing that reductions
never cause disjoint graph IR computations to become aliased.

Leveraging the type-and-effects safety in the \maybelang-calculus, the equational rules shown in the paper (Section 5.1) are
proved sound by contextual equivalence via logical relations,
building upon a framework developed in parallel with
this report \cite{bao2023logrel}.

We present the corresponding development in detail as follows:

\begin{itemize}
  \vspace{2ex}
  \item \textbf{The direct-style \maybelang-calculus (\Cref{sec:directstyle}):} we introduce the
        syntax and typing rules of our base calculus with reachability types and effects, motivating the
        basic design principles and features of the system. The formalization and proofs very closely
        follow the publicly available Coq mechanizations of the original \oldlang-calculus\footnote{\url{http://github.com/tiarkrompf/reachability}}.
        \maybelang as presented in this report lacks some features of the original \oldlang-calculus
        (e.g., no recursion, no escaping closures,  flat mutable references) which are non-essential to
        understand the core ideas. %

  \item \textbf{The \maybelang-calculus with store-allocated values (\Cref{sec:directstyle_letv}):}
        as a stepping stone towards monadic normal form, we refine the call-by-value operational semantics
        of \maybelang to place both mutable references and immutable introduction forms in the store, and prove the direct-style
        type system sound w.r.t.\ this refined semantics. Notably, substitutions become simpler, because
        they are just renamings of variables to store locations.
  \item \textbf{The monadic normal form \mnflang (\Cref{sec:monadic}):}
        we define a provably type/effect/qualifier-preserving translation of the direct-style \maybelang-calculus into \mnflang which is in monadic normal form (MNF).
        This normal form generalizes the previously used A-normal form (ANF) by permitting nested terms.
        Unlike ANF, reductions preserve the monadic normalform at all times, even under $\beta$-reduction.
        We further establish that \mnflang is a proper sublanguage of \maybelang, by (1) proving that MNF terms can always be assigned the
        same type, effect, and reachability qualifiers in both systems, and (2) proving that reduction with store-allocated values (\Cref{sec:directstyle_letv})
        preserves MNF. Type soundness of \mnflang follows from type soundness of \maybelang as a corollary.
  \item \textbf{The graph IR \irlang with (hard) dependencies (\Cref{sec:monad-norm-with}):}
        we enrich the MNF-calculus \mnflang with effect dependencies in the \irlang-calculus. Dependencies are entirely determined by reachability qualifiers and effects.
        We prove type soundness and preservation of separation with respect to a stricter operational semantics, which establishes
        \emph{dependency safety}: evaluation respects the order of effect dependencies for well-typed graph IR terms, \ie,
        an effectful graph node is executed only if all its dependencies have already been resolved.
        We also prove end-to-end type/effect/qualifier preservation and effect synthesis from the direct-style \maybelang-calculus into the \irlang graph IR.
  \item \textbf{The graph IR \irlang in MNF with hard and soft dependencies (\Cref{sec:extension-with-soft}):}
        we refine the effects of \irlang from mere uses to a read and write distinction, which synthesize into hard and soft dependencies.
  \item \textbf{The equational theory of \maybelang (\Cref{sec:direct-lr}):} We develop logical
        relations over reachability types and effects, which enables reasoning about contextual
        equivalence of \maybelang{} terms. %
  \item \textbf{Optimization rules and equational theory of \langg (\Cref{sec:optimizations}):}
        We prove soundness of the optimization rules in the main paper for the \langg graph IR with hard dependencies
        in terms of contextual equivalence. We leverage the results from \Cref{sec:monadic,sec:extension-with-soft,sec:direct-lr},
        to derive the logical relations argument by appealing to the direct-style system through a "round-trip" translation
        erasing and re-synthesizing hard dependencies. We leave the logical relations argument for the system including
        soft dependencies as future work.

  \item \textbf{Code motion algorithms for \irlang (\Cref{sec:codemotion}):}
        we present the code motion algorithms that transform graphs into trees and emit code.
        A vanilla code motion algorithm is presented first, which is then extended with
        frequency estimation and compact code generation.
\end{itemize}

\clearpage %
\section{The Direct-Style \maybelang{}-Calculus}
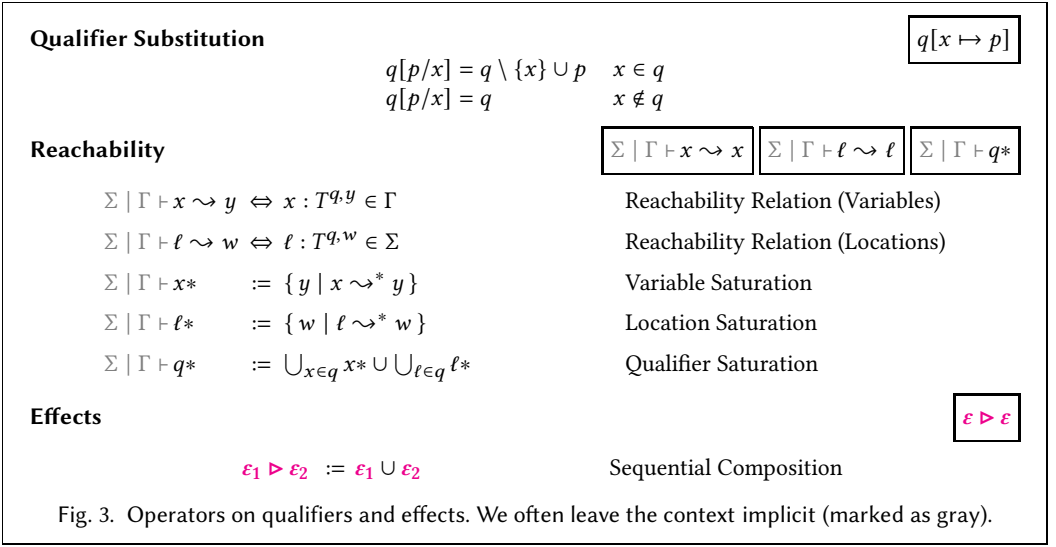
\label{sec:directstyle}
\begin{figure}\small
\begin{mdframed}
\judgement{Syntax}{\BOX{\maybelang}}\small\vspace{-8pt}
\[\begin{array}{l@{\quad}l@{\quad}l@{\quad}l}
    x,y,z                &\in & \Var                                                                  & \text{Variables} \\
    \ell,w                 &\in & \Loc                                                                  & \text{Locations}   \\

    v                    &::= & c\mid \lambda x.t\mid \ell                                             & \text{Values}\\
    t                    &::= & v\mid\ell\mid x \mid t~t \mid \tref_\ell~t \mid\ !~t \mid t \coloneqq t\mid \tlet~{x=t}~\tin~t  & \text{Terms}\\

    p,q,r,\EPS,\flt      &\in & \mathcal{P}_{\mathsf{fin}}(\Var\uplus\Loc)                            & \text{Qualifiers/Effects/Observations} \\
    S,T,U,V              &::= & B \mid (x: \ty[q]{T}) \to^{\EPS} \ty[q]{T} \mid \TRef~T               & \text{Types} \\

    \Gamma               &::= & \varnothing\mid \Gamma, x : \ty[q]{T}                                 & \text{Typing Environments} \\
    \Sigma               &::=&  \varnothing\mid \Sigma, \ell : \ty[q]{T}                              & \text{Store Typing}\\
\end{array}\]\\
\typicallabel{t-abs}
\judgement{Term Typing}{\BOX{\strut\GS[\flt] \ts t : \ty[q]{T}\ \EPS}} %
\begin{minipage}[t]{.47\linewidth}\vspace{0pt}
  \infrule[t-cst]{\ \\ \ \\
    c \in B
  }{
    \GS[\flt] \ts c : \ty[\qbot]{B}\ \PURE
  }
  \vgap
    \infrule[t-var]{
      x : \ty[q]{T} \in \G\quad\quad x \subq \flt
    }{
      \GS[\flt] \ts x : \ty[x]{T}\ \PURE
    }
    \vgap
    \infrule[t-loc]{
    \ell : \ty[q]{T} \in \Sigma\qquad \ell \subq \flt
    }{
      \GS[\flt]\ts \ell : \ty[\ell]{T}\ \PURE
    }
\vgap
    \infrule[t-abs]{
      \csx[q,x]{\G\ ,\ x: \ty[p]{T}}{\Sigma} \ts t : \ty[r]{U}\ \EPS\\ q\subq \flt
    }{
      \GS[\flt] \vdash \lambda x.t : \ty[q]{\left(x: \ty[p]{T} \to^{\EPS} \ty[r]{U}\right)}\; \PURE
    }
\vgap
\infrule[t-app]{
      \GS[\flt]\vdash t_1 : \ty[q]{\left(x{\,:\,}\ty[\qsat{p}\,{\overlap}\, \qsat{q}]{T} \to^{\FX{\EPS[3]}} \ty[r]{U}\right)}\ \EPS[1]\\
      \GS[\flt]\ts t_2 : \ty[p]{T}\ \EPS[2]\quad\theta = [p/x]\\
      x\notin\FV(U)\quad\FX{\EPS[3]}\subq q,x \quad r\subq\varphi,x
    }{
      \GS[\flt]\ts t_1~t_2 : (\ty[r]{U}\ \EPS[1]\EFFSEQ\EPS[2]\EFFSEQ\FX{\EPS[3]})\theta
    }
\end{minipage}%
\begin{minipage}[t]{.03\linewidth}
\hspace{1pt}%
\end{minipage}%
\begin{minipage}[t]{.5\linewidth}\vspace{0pt}\typicallabel{t-let}
  \infrule[t-let]{
    \GS[\flt]\ts t_1 : \ty[p]{S} \ \EPS[1]\\ %
      \csx[\flt,x]{\G\, ,\, x: \ty[\qsat{p}\cap\qsat{\flt}]{S}}{\Sigma} \ts t_2 : \ty[q]{T}\ \EPS[2]\\
      \theta = [p/x] \quad x \notin\FV{(T)}
}{
  \GS[\flt]\ts \tlet~{x = t_1}~\tin~t_2 : (\ty[q]{T}\ \EPS[1]\EFFSEQ\FX{\EPS[2]})\theta
}
\vgap
  \infrule[t-ref]{
    \GS[\flt]\ts t_1 : \ty[q]{\Typ{Alloc}}\ \EPS[1]\\ \GS[\flt]\ts t_2 : \ty[\qbot]{B}\ \EPS[2]
    }{
      \GS[\flt]\ts \tref_{t_1}~t_2 : \ty[\qbot]{(\TRef~\ty{B})}\ \EPS[1]\EFFSEQ\EPS[2]\EFFSEQ\FX{\bm{q}}
    }
\vgap
    \infrule[t-\(!\)]{
      \GS[\flt]\ts t : \ty[q]{(\TRef~\ty{B})}\ \EPS
    }{
      \GS[\flt]\ts !t : \ty[\qbot]{B}\ \EPS\EFFSEQ\FX{\bm{q}}
    }
\vgap
    \infrule[t-\(:=\)]{
      \GS[\flt]\ts t_1 : \ty[q]{(\TRef~\ty{B})}\ \EPS[1] \\
      \GS[\flt]\ts t_2 : \ty[\qbot]{B}\ \EPS[2]
    }{
      \GS[\flt]\ts t_1 \coloneqq t_2 : \ty[\qbot]{\TUnit}\ \EPS[1]\EFFSEQ\EPS[2]\EFFSEQ\FX{\bm{q}}
    }
\vgap
\infrule[t-sub]{
      \GS[\flt]\ts t : \ty[p]{S}\ \EPS[1] \quad \GS\ts\ty[p]{S}\ \EPS[1] <: \ty[q]{T}\ \EPS[2]\\ q,\EPS[2]\subq\flt
    }{
      \GS[\flt]\ts t : \ty[q]{T}\ \EPS[2]
    }
\end{minipage}

\judgement{Subtyping}{\BOX{\strut\GS \ts q <: q}\ \BOX{\strut \GS\ts\ty{T} <: \ty{T}}\ \BOX{\strut \GS\ts\ty[q]{T}\ \EPS <: \ty[q]{T}\ \EPS}}\\[1ex]
\begin{minipage}[t]{.47\linewidth}\small\vspace{0pt}
  \typicallabel{s-base}
  \infrule[q-sub]{\ \\ p\subq q\subq \dom(\G)\cup\dom(\Sigma)}{\GS\ts p <: q}
  \vgap
  \infrule[s-base]{\ }{
    \GS\ts\ty{B} <: \ty{B}
  }
\vgap
  \infrule[s-ref]{
  }{
    \GS\ts\ty{\TRef~\ty{B}} <: \ty{\TRef~\ty{B}}
  }
\end{minipage}%
\begin{minipage}[t]{.03\linewidth}
\hspace{1pt}
\end{minipage}%
\begin{minipage}[t]{.5\linewidth}\small\vspace{0pt}
  \typicallabel{sqe-sub}
  \infrule[s-fun]{
    \GS\ts\ty[q]{U}\ \PURE <: \ty[o]{S}\ \PURE \\
    \G\, ,\, x : \ty[p]{U}\mid\Sigma\ts \ty[q]{T}\ \EPS[1] <: \ty[r]{V}\ \EPS[2]
  }{
    \GS\ts\ty{(x: \ty[o]{S}) \to^{\EPS[1]} \ty[q]{T}} <: \ty{(x: \ty[p]{U}) \to^{\EPS[2]} \ty[r]{V}}
  }
\vgap
  \infrule[sqe-sub]{
\GS\ts\ty{S} <: \ty{T}\\ \GS\ts p <: q\quad\GS\ts \EPS[1] <: \EPS[2]
  }{
    \GS\ts\ty[p]{S}\ \EPS[1] <: \ty[q]{T}\ \EPS[2]
  }
\end{minipage}

 \caption{The direct-style \maybelang-calculus.}\label{fig:maybe:syntax}
\end{mdframed}
\end{figure}

\begin{figure}
\begin{mdframed}\small
\judgement{Qualifier Substitution}{\BOX{q[x\mapsto p]}}\vspace{-8pt}
\[\begin{array}{l@{\;}c@{\;}ll}
    q[p/x] & = & q\setminus\{x\}\qlub p& x\in q    \\
    q[p/x] & = & q                     & x\notin q
  \end{array}\]
  \judgement{Reachability}{\BOX{{\color{gray}\GS\vdash}\,x \reaches x}\ \BOX{{\color{gray}\GS\vdash}\,\ell \reaches \ell}\ \BOX{{\color{gray}\GS\vdash}\, \qsat{q}}}
  \[\begin{array}{l@{\ \,}c@{\ \,}l@{\qquad\qquad\qquad\ \ }l}
  {\color{gray}\GS\vdash}\, x \reaches y & \Leftrightarrow  & x : T^{q,y}\in\Gamma  & \text{Reachability Relation (Variables)} \\[1.1ex]
  {\color{gray}\GS\vdash}\, \ell\reaches w & \Leftrightarrow  & \ell : T^{q,w} \in\Sigma & \text{Reachability Relation (Locations)} \\[1.1ex]
  {\color{gray}\GS\vdash}\, \qsat{x} & := & \left\{\, y \mid x \reaches^* y\, \right\} & \text{Variable Saturation} \\[1.1ex]
  {\color{gray}\GS\vdash}\, \qsat{\ell} & := & \left\{\, w \mid \ell \reaches^* w\, \right\} & \text{Location Saturation} \\[1.1ex]
  {\color{gray}\GS\vdash}\, \qsat{q} & := &  \bigcup_{x\in q}\qsat{x} \cup \bigcup_{\ell\in q}\qsat{\ell} & \text{Qualifier Saturation}
  \end{array}\]
\judgement{Effects}{\BOX{\EPS\EFFSEQ\EPS}}
\[\begin{array}{l@{\ \,}c@{\ \,}l@{\qquad\qquad\qquad\ \ }l}
  \phantom{\G\vdash\,}\EPS[1]\EFFSEQ\EPS[2] \ \  & := & \EPS[1] \cup \EPS[2] & \quad\ \, \text{Sequential Composition}\\
   \end{array}\]
\caption[Operators on qualifiers and effects.]{Operators on qualifiers and effects. %
  We often leave the context implicit (marked as gray).}\label{fig:saturation_overlap}
\end{mdframed}
\end{figure}

\subsection{Overview}

The \maybelang{}-calculus is a variant of \citeauthor{DBLP:journals/pacmpl/BaoWBJHR21}'s
\oldlang{}-calculus. The original system features an effect system based on
\citet{DBLP:journals/toplas/Gordon21}'s effect quantale framework. For simplicity, we only
consider a stripped-down effect system corresponding to a trivial effect quantale just tracking
whether an effect is induced on reachable variables, effectively making effects just another
qualifier (i.e., a set of variables and locations) in the typing judgment. This effect system is
sufficient for calculating granular effect dependencies. This version also lacks a \(\bot\)
qualifier for untracked values, and recursive \(\lambda\)-abstractions.\footnote{Modulo the
(straightforward) addition of effects, the system presented here closely follows the mechanized
``overlap lazy'' variant found at
\url{https://github.com/tiarkrompf/reachability/tree/main/base/lambda_star_overlap_lazy}.}
To keep the discussion focused and on point, we omitted those features which do not add much
to the discussion of the core ideas apart from additional proof cases.

\subsection{Examples}\label{sec:directstyle:overview}

Here we use a few examples from \citet{DBLP:journals/pacmpl/BaoWBJHR21} to
motivate the direct-style \maybelang{}-calculus as the base system.
The basic idea is to track which other values are reachable from a given expression's result.
For example, the following expression allocates a reference cell of integers and binds
it to @x@:
\begin{lstlisting}
  val x = new Ref(42)       // : Ref[Int]$\tracksetl{x}$ $\FX{\{w\}}$
\end{lstlisting}
Therefore, by reflexivity, the type of @x@ has a qualifier (\ie a set of
variables) that contains itself.
The allocation also induces an effect over an allocator $w$ for
allocating memory resources.
We can further bind @x@ to a new variable @y@, which induces no effect (indicated by the empty set):
\begin{lstlisting}
  val y = x                 // : Ref[Int]$\tracksetl{y}\phantom{x}$ $\PURE$      ← lazy assignment (this work)
                            // : Ref[Int]$\trackset{x,y}$ $\PURE$ ← eager assignment (Bao et al.)
\end{lstlisting}
As in \citet{wei2023polymorphic},
the system considered here is ``lazy'' with respect to qualifier assignment, in the sense that it
assigns minimal qualifiers. For example, @y@ reaches itself, just like @x@ does, but
in this particular context, we can also deduce that @y@ indirectly reaches @x@, since it is an alias.
This provides a certain degree of lightweight, i.e., quantifier-free, reachability polymorphism \cite{wei2023polymorphic}.
We use the notation \lstinline|{y}$^*$ = {x,y}| for the transitively closed reachability set
in a given context (\Cref{fig:saturation_overlap}). In contrast,
\citeauthor{DBLP:journals/pacmpl/BaoWBJHR21} use an ``eager''
qualifier assignment, \eg, they would always assign the transitively closed qualifier.

We refer readers to Section 2 of \citet{DBLP:journals/pacmpl/BaoWBJHR21} for more illustrative
examples of reachability types.
Section 2 of the main paper~\cite{oopsla23} also provides examples
after adapting reachability types to the graph IR.

\subsection{Syntax}\label{sec:directstyle:syntax}

\Cref{fig:maybe:syntax} shows the syntax of \maybelang which is based on the simply-typed
\(\lambda\)-calculus
with mutable references and subtyping.  We denote general term variables by the meta variables
\(x, y, z\),
and reserve \(\ell, w\) for store locations.

Terms consist of constants of base types, variables, functions \(\lambda x.t\), function
applications, reference allocations, dereferences, assignments, and $\tlet$-expressions.

Reachability qualifiers \(p,q,r\) are finite sets of variables and store locations. For readability,
we often drop the set notation for qualifiers and write them down as comma-separated lists of atoms.

We distinguish ordinary types \(T\) from qualified types \(\ty[q]{T}\), where the latter
annotates a qualifier $q$ to an ordinary type $T$. The types consist of base types \(B\) (\eg,
\Type{Int}, \TUnit), dependent function types \((x: \ty[q]{T}) \to^{\EPS} \ty[p]{S}\), where both argument
and return type are qualified. The codomain \(\ty[p]{S}\) may depend on the argument \(x\) in its
qualifier and type. Function types carry an annotation $\EPS$ for its latent effect, which is
a set of variables and locations, akin to qualifiers.

For simplicity, mutable reference types \(\TRef~B\) can only store values of base types.
We could also permit forms of nested references, by adding a flow-sensitive
effect system~\cite{DBLP:journals/pacmpl/BaoWBJHR21}.

An \emph{observation} \(\flt\) is a finite set of variables which is part of the term typing
judgment (\Cref{sec:directstyle:statics}). It specifies which variables and locations in the typing
context \(\Gamma\) and store typing \(\Sigma\) are observable.  The former assigns qualified typing
assumptions to variables.

\subsection{Statics}\label{sec:directstyle:statics}
The term typing judgment $ \strut\GS[\flt] \ts t : \ty[q]{T}\ \EPS $
in \Cref{fig:maybe:syntax} states that term \(t\)
has qualified type \(\ty[q]{T}\) and may induce effect \(\EPS\),
and may only access the typing assumptions of \(\Gamma\)
observable by \(\flt\).
One may think of \(t\) as a computation that incurs effect $\EPS$ and
yields a result value of type \(T\) aliasing no more than \(q\), if it terminates.

Different from \citet{DBLP:journals/pacmpl/BaoWBJHR21}, we internalize
the filter $\flt$ as part the typing relation.
Alternatively, we could formulate the typing judgment without internalizing \(\flt\),
and instead have an explicit context filter operation
\(\G[\flt] := \{x : \ty[q]{T}\in\G \mid q,x \subq \flt \}\)
for restricting the context in subterms, just like \citet{DBLP:journals/pacmpl/BaoWBJHR21} which
loosely takes inspiration from substructural type systems. Internalizing \(\flt\)
(1) makes observability an explicit notion, which facilitates reasoning about
separation and overlap, and (2) greatly simplifies the Coq mechanization.
Context filtering is only needed for term typing, but not for subtyping, so as
to keep the formalization simple.

\subsubsection{Functions and Lightweight Polymorphism}\label{sec:fun-lightweight-poly}
Function typing \rulename{t-abs} implements the observable separation
guarantee, \ie, the body \(t\) can only observe what the function type's
qualifier \(q\) specifies, plus the argument \(x\),
and is otherwise oblivious to anything else in the environment.
We model this by setting the observation to \(q,x,f\) when typing the body.
Thus, its observation \(q\) at least includes the free variables of \(t\).
To ensure well-scopedness, \(q\) must be a subset of the observation \(\flt\) on the outside.
In essence, a function type \emph{implicitly} quantifies over anything that is not
observed by \(q\), achieving a lightweight form of qualifier polymorphism.

\subsubsection{Dependent Application, Separation and Overlap}\label{sec:depend-appl}
Function applications \rulename{t-app} are qualifier-dependent in that the
result qualifier can depend on the argument.

Function applications also establish an \emph{observable separation} between the argument reachable
set $p$ and the function reachable set $q$, as denoted as $\qsat{p} \cap \qsat{q}$. The intersection
between $\qsat{p}$ and $\qsat{q}$ specifies the permitted overlap. We are careful to intersect the
transitive reachability closure (a.k.a.\ saturated version, \Cref{fig:saturation_overlap}) of the
two qualifiers. This is necessary in the lazy reachability assignment, because we might miss
common, indirect overlap between the sets otherwise. If the intersection declared in the function
type is empty, then it means complete separation between the argument and the entities
observed by the function from the environment.

\subsubsection{Qualifier Substitution}\label{sec:qual-subst-growth}
The base substitution operation \(q[p/x]\) of qualifiers for variables is defined in
\Cref{fig:saturation_overlap}, and we use it along with its homomorphic
extension to types in dependent function application.  Substitution replaces
the variable with the given qualifier, if present in the target.

\subsubsection{Effects}
Our effect system is a simple flow-insensitive instantiation of
\citet{DBLP:journals/toplas/Gordon21}'s effect quantale system.
An effect $\EPS$ denotes the set of variables/locations that might be used
during the computation.
For a compound term, the final effect is computed by composing the effects of
sub-terms with the intrinsic effect of this term.
For example, the effect of assignments has two parts: (1) $\EPS[1]$, $\EPS[2]$
the effects of sub-terms, and (2) $\FX{q}$ the variables/locations being
modified. The final effect is obtained by composing these effects.

Although the typing rules presented in \Cref{fig:maybe:syntax} pretend to use the
sequential effect composition operator $\EFFSEQ$, its definition $\cup$ computes an
upper bound of two effects and is \emph{not} flow-sensitive
(\Cref{fig:saturation_overlap}), \ie the composed effect is not sensitive to
the order of composition. This simple instantiation is sufficient for deriving
dependencies (cf.~\Cref{sec:graphir}).

\subsubsection{Mutable References}\label{sec:mutable-refs}

Slightly different from \citet{DBLP:journals/pacmpl/BaoWBJHR21}, the allocation
$\tref_{t_1}~t_2$ \rulename{t-ref} additionally takes a term $t_1$ that has an
\textsf{Alloc} primitive type. An allocation induces an effect over aliases of
$t_1$, which is recorded as the composed term effect.

The typing rule of reference allocation \rulename{t-ref}, read \rulename{t-!},
and write \rulename{t-:=} work with reference types whose inner referent values
are base values. This is sufficient for understanding the core ideas of the graph IR.
It is nevertheless possible to extend the base system with nested references (\eg
allowing storing functions) by a flow-sensitive effect system, as shown by
\citet{DBLP:journals/pacmpl/BaoWBJHR21}.
The subtyping of references is invariant in the referent type.

\subsubsection{Subtyping}\label{sec:subtyping}
We distinguish subtyping between qualifiers \(q\),
ordinary types \(T\),
and qualified types \(\ty[q]{T}\),
where the latter two are mutually dependent. Subtyping is assumed to be well-scoped under the typing
context \(\G\) and store \(\Sigma\),
\ie, types and qualifiers mention only variables/locations bound in \(\G\) and \(\Sigma\),
and so do its typing assumptions.  Qualified subtyping \rulename{sqe-sub} just forwards to the other
two judgments for scaling the type, qualifier, and effect respectively.

\paragraph{Qualifier Subtyping}
Qualifier subtyping includes the subset relation \rulename{q-sub}, which resort to the subset relation
since qualifiers are sets. Since effects are just qualifiers, we use the same subtyping relation
for subeffecting.

\paragraph{Ordinary Subtyping}
Subtyping rules for base types \rulename{s-base}, reference types \rulename{s-ref}, and function
types \rulename{s-fun} are standard modulo qualifiers. Reflexivity and transitivity are both
admissible for subtyping on ordinary and qualified types.
Function
types are contravariant in the domain, and covariant in the codomain and effect, as usual.
Due to dependency in the codomain, we are careful to extend the context with
the smaller argument type.

\subsection{Dynamics}
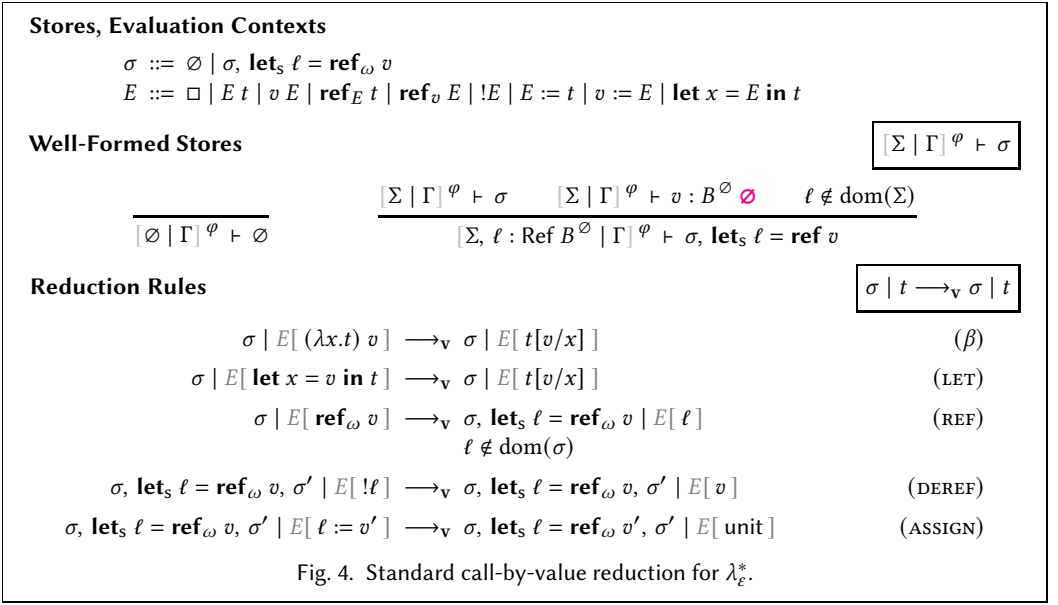
\label{sec:directstyle:dynamics}
\begin{figure}\small
\begin{mdframed}
\judgement{Stores, Evaluation Contexts}{}
\[\begin{array}{l@{\ \ }c@{\ \ }l@{\qquad\qquad\ }l@{\ \ }c@{\ \ }l}
    {\sigma} & ::= & \varnothing \mid\sigma,\,\tlets~{\ell=\tref_{\omega}~v}            & & & \\
    {E}      & ::= & \square \mid E\ t \mid v\ E \mid \tref_E~t\mid \tref_v~E \mid\ !{E} \mid {E} := {t} \mid {v} := {E}\mid \tlet~{x = E}~\tin~t & & & \\
  \end{array}\]
\judgement{Well-Formed Stores}{\BOX{\GS[\flt]\ts \sigma}} %
\begin{mathpar}
  \inferrule{\ }{\csx[\flt]{\G}{\varnothing}\ts\varnothing}\and
    \inferrule{\csx[\flt]{\G}{\Sigma}\ts\sigma\qquad \GS[\flt]\ts v : \ty[\qbot]{B}\ \PURE\qquad\ell \notin\DOM(\Sigma)
     }{
      \csx[\flt]{\G}{\Sigma,\, \ell : \ty[\qbot]{\TRef~B}}\ts\sigma,\,\tlets~\ell = \tref~v
     }
  \end{mathpar}
\judgement{Reduction Rules}{\BOX{\sigma \mid t \redv \sigma\mid t}}
\[\begin{array}{r@{\ \ }c@{\ \ }ll@{\qquad\qquad}r}
  \sigma\mid\CX[gray]{E}{(\lambda x.t)\ v}                                                  & \redv & \sigma\mid\CX[gray]{E}{t[v/x]}                                            & & \rulename{$\beta$} \\[1ex]
  \sigma\mid\CX[gray]{E}{\tlet~{x = v}~\tin~t}                                              & \redv & \sigma\mid\CX[gray]{E}{t[v/x]}                                            & & \rulename{let} \\[1ex]
  \sigma\mid\CX[gray]{E}{\tref_\omega~v}                                                    & \redv & \sigma,\,\tlets~{\ell = \tref_\omega~v}\mid\CX[gray]{E}{\ell}             & & \rulename{ref} \\
                                                                                            &       & \ell \not\in \DOM(\sigma)                                                 & & \\[1ex]
  \sigma,\,\tlets~{\ell = \tref_\omega~v},\,\sigma'\mid\CX[gray]{E}{!\ell}                  & \redv & \sigma,\,\tlets~{\ell = \tref_\omega~v},\,\sigma'\mid\CX[gray]{E}{v}       & & \rulename{deref} \\[1ex]
  \sigma,\,\tlets~{\ell = \tref_\omega~v},\,\sigma'\mid\CX[gray]{E}{\ell := v'}             & \redv & \sigma,\,\tlets~{\ell = \tref_\omega~v'},\,\sigma'\mid\CX[gray]{E}{\tunit} & & \rulename{assign}
  \end{array}\]
\caption{Standard call-by-value reduction for \maybelang.}\label{fig:directstyle:semantics}
\end{mdframed}
\end{figure}
 
The single-step, call-by-value (CBV) for \maybelang (\Cref{fig:directstyle:semantics}) is standard.
To bridge the gap to monadic normal forms later on, we model stores as sequences of mutable let
bindings. One may view computations as syntactic sequences of let bindings (think already evaluated
graph nodes) followed by a redex. Another difference to standard treatments is that reference
allocation takes the allocation-capability variable \(\omega\) as an explicit extra argument.
This design makes the treatment of effectful operations uniform, in the sense that an operation always
induces an effect on some operand.

The allocation capability \(\omega\) is a base constant for allocation of base type \(\Typ{Alloc}\),
we consider an initial store with \(\tlets~{w = \omega}\), and \(w : \ty[\qbot]{\Typ{Alloc}}\).

\subsection{Metatheory}\label{sec:directstyle:metatheory}

The \maybelang-calculus exhibits syntactic type soundness which we prove by standard progress and
preservation properties (\Cref{thm:progress,thm:soundness}). Type soundness implies the preservation
of separation corollary (\Cref{coro:preservation_separation}) as set forth by
\citet{DBLP:journals/pacmpl/BaoWBJHR21} for their \(\lambda^*\)-calculus.
It is a memory property certifying that the results of well-typed \maybelang terms with disjoint
qualifiers indeed never alias.

The metatheory of the \maybelang-calculus closely follows the ``lazy overlap'' variant of $\lambda^*$-calculus,
which has been mechanized.\footnote{The mechanization can be found at
\url{https://github.com/tiarkrompf/reachability/tree/main/base/lambda_star_overlap_lazy}.}
The major difference lies in the addition of a simple effect system, which does not change the metatheory significantly,
other than carrying an extra qualifier for effects in judgments.
Below, we discuss key lemmas required for the type soundness proof.

\subsubsection{Observability Properties} \label{sec:maybe:obs}
Reasoning about substitutions and their interaction with overlap/separation in preservation lemmas
requires that the qualifiers assigned by term typing
are observable. The following lemmas are proved by induction over the respective typing derivations:
\begin{lemma}[Observability Invariant]\label{lem:has_type_filter}
Term typing always assigns observable qualifiers and effects, \ie,
if $\csx[\flt]{\G}{\Sigma} \ts t : \ty[q]{T}\ \EPS$, then $q,\EPS\subq \flt$.
\end{lemma}

\noindent Well-typed values cannot observe anything about the context beyond their assigned qualifier:
\begin{lemma}[Tight Observability for Values]\label{lem:values_tight}
If \(\ \csx[\flt]{\G}{\Sigma} \ts v : \ty[q]{T}\ \EPS\), then \(\csx[q]{\G}{\Sigma} \ts v : \ty[q]{T}\ \PURE\).
\end{lemma}
\noindent It is easy to see that any observation for a function \(\lambda x.t\)  will at least track
the free variables of the body \(t\).

\subsubsection{Weakening and Narrowing Lemmas}\label{sec:maybe:weaken_narrow}
The \maybelang calculus has standard weakening and narrowing lemmas.
\begin{lemma}[Subtyping Weakening]\label{lem:subtyping_weakening}\ \vspace{-5pt}
    \infrule{%
    \csx{\G}{\Sigma}\ts p <: q \qquad  \G' \supseteq \G\qquad \Sigma' \supseteq \Sigma
    }{
        \csx{\G'}{\Sigma'}\ts p <: q
    }
    \infrule{%
    \csx{\G}{\Sigma}\ts S <: T \qquad  \G' \supseteq \G\qquad \Sigma' \supseteq \Sigma
    }{
        \csx{\G'}{\Sigma'}\ts S <: T
    }
    \infrule{%
    \csx{\G}{\Sigma}\ts \ty[p]{S}\ \EPS[1] <: \ty[q]{T}\ \EPS[2] \qquad  \G' \supseteq \G\qquad \Sigma' \supseteq \Sigma
    }{
        \csx{\G'}{\Sigma'}\ts \ty[p]{S}\ \EPS[1] <: \ty[q]{T}\ \EPS[2]
    }
\end{lemma}
\begin{proof}
    Weakening on qualifier subtyping trivially follows from its definition. The others are proved by mutual induction over the respective derivations.
\end{proof}
\begin{lemma}[Weakening]\label{lem:weakening}\ \vspace{-8pt}
    \infrule{%
    \csx[\flt]{\G}{\Sigma}\ts t : \ty[q]{T}\ \EPS \qquad \G' \supseteq \G\qquad \Sigma' \supseteq \Sigma\qquad \flt' \supseteq \flt
    }{
        \csx[\flt']{\G'}{\Sigma'}\ts t : \ty[q]{T}\ \EPS
    }
\end{lemma}
\begin{proof}
    By induction over the term typing derivation, using \Cref{lem:subtyping_weakening} where appropriate.
\end{proof}
\begin{lemma}[Subtyping Narrowing]\label{lem:subtyping_narrowing}\ \vspace{-8pt}
\infrule{%
\csx{\G, x : \ty[p]{V},\G'}{\Sigma}\ts q <: r \qquad \GS\ts \ty[o]{U}\ \EPS[1] <: \ty[p]{V}\ \EPS[2]
}{
    \csx{\G, x : \ty[o]{U},\G'}{\Sigma}\ts q <: r
}
\infrule{%
\csx{\G, x : \ty[p]{V},\G'}{\Sigma}\ts S <: T \qquad \GS\ts \ty[o]{U}\ \EPS[1] <: \ty[p]{V}\ \EPS[2]
}{
    \csx{\G, x : \ty[o]{U},\G'}{\Sigma}\ts S <: T
}
\infrule{%
\csx{\G, x : \ty[p]{V},\G'}{\Sigma}\ts \ty[q]{S}\ \EPS[3] <: \ty[r]{T}\ \EPS[4] \qquad \GS\ts \ty[o]{U}\ \EPS[1] <: \ty[p]{V}\ \EPS[2]
}{
    \csx{\G, x : \ty[o]{U},\G'}{\Sigma}\ts \ty[q]{S}\ \EPS[3] <: \ty[r]{T}\ \EPS[4]
}
\end{lemma}
\begin{proof}
    By mutual induction over the respective derivations.
\end{proof}
\begin{lemma}[Narrowing]\label{lem:narrowing}\ \vspace{-8pt}
    \infrule{%
    \csx[\flt]{\G, x : \ty[{p}]{V},\G'}{\Sigma}\ts t : \ty[q]{T}\ \EPS[1] \qquad \GS\ts \ty[o]{U}\ \EPS[2] <: \ty[p]{V}\ \FX{\EPS[3]}
    }{
        \csx[\flt]{\G, x : \ty[o]{U},\G'}{\Sigma}\ts t : \ty[q]{T}\ \EPS[1]
    }
\end{lemma}
\begin{proof}
    By induction over the term-typing derivation, using \Cref{lem:subtyping_narrowing} where appropriate.
\end{proof}

\subsubsection{Substitution Lemmas}

We consider type soundness for closed terms and apply ``top-level'' substitutions, \ie, substituting
closed values with qualifiers that do not contain term variables, but only store locations.  The
proof of the substitution lemma  critically relies on the distributivity of substitution and the
qualifier intersection operator for checking overlap, which is required to proceed in the
\rulename{t-app} case:

\begin{lemma}[Top-Level Substitutions Distribute with Overlap]\label{lem:subst_commutes_overlap}\ \vspace{-8pt}
\infrule{ x: \ty[q]{T}\in\G\quad \theta = [p/x] \quad p,q\subq\dom(\Sigma) \quad p \qglb \flt \subq q\quad {r},{r'}\subq\flt\quad r = \qsat{r}\quad r' = \qsat{r'}}{({r} \overlap {r'})\theta = {r}\theta\overlap {r'}\theta}
\end{lemma}
\noindent Qualifier substitution does not generally distribute with set intersection, due to the problematic
case when the substituted variable \(x\) occurs in only one of the saturated sets \({r}\)
and \({r'}\). Distributivity holds if (1) we ensure that what is observed about the qualifier \(p\) we substitute for \(x\)
is bounded by what the context observes about \(x\), \ie, \(p \qglb \flt\subq q\)
for \(x : \ty[q]{T}\in\G\), and (2) \(p,q\) are top-level as above.
Furthermore, we require that the intersected qualifiers \(r\) and \(r'\) are reachability saturated, which is given in the context
of \rulename{t-app}.

\begin{lemma}[Top-Level Substitution for Qualifier/Effect Subtyping]\label{lem:subst_quals}\ \vspace{-8pt}
\infrule{%
  \csx{\G,x:\ty[q]{T}}{\Sigma} \ts p <: r \qquad
  q,q'\subseteq \DOM(\Sigma) \qquad
  \theta = [q'/x] \qquad
  \WF{\csx{\G, x : \ty[q]{T}}{\Sigma}}
}{
  \csx{\G\theta}{\Sigma} \ts p\theta <: r\theta
}
\end{lemma}
\begin{proof}
  By the fact that substitution is monotonic w.r.t.\ subset inclusion
  \(\subseteq\) and qualifier/effect subtyping being that relation by
  definition.
\end{proof}

\begin{lemma}[Top-Level Substitution for Subtyping]\label{lem:subst_subtyping}\ \vspace{-5pt}
    \infrule{%
      \csx{\G,x:\ty[q]{S}}{\Sigma} \ts T <: U\qquad q,p\subseteq \DOM(\Sigma)\qquad\theta = [p/x]
    }{
      \csx{\G\theta}{\Sigma} \ts T\theta <: U\theta
    }
    \infrule{%
      \csx{\G,x:\ty[q]{S}}{\Sigma} \ts \ty[p]{T}\ \EPS[1] <: \ty[q]{U}\ \EPS[2]\qquad q,p\subseteq \DOM(\Sigma)\qquad\theta = [p/x]
    }{
      \csx{\G\theta}{\Sigma} \ts \ty[p]{T}\theta\ \FX{\EPS[1]\bm{\theta}} <: \ty[q]{U}\theta\ \FX{\EPS[2]\bm{\theta}}
    }
\end{lemma}
\begin{proof}
By mutual induction over the respective subtyping derivations, using \Cref{lem:subst_quals} where appropriate.
\end{proof}
In the type preservation proof, $\beta$-reduction substitutes a function parameter for some value,
which requires a carefully formulated substitution lemma:
\begin{lemma}[Top-Level Term Substitution]\label{lem:subst_term}\ \vspace{-8pt}
\infrule{
   \csx[\flt]{\G,x:\ty[p\overlap r]{S}}{\Sigma} \ts t : \ty[q]{T}\ \EPS\qquad \csx[p]{\varnothing}{\Sigma}\ts v : \ty[p]{S}\ \PURE\qquad \theta = [p/x] \\[1ex]
    p\subq\dom(\Sigma)\qquad p \qglb \flt \subq p\overlap r
    }{
        \csx[\flt\theta]{\Gamma\theta}{\Sigma} \ts t[v/x] : (\ty[q]{T}\ \FX{\EPS})\theta
     }
\end{lemma}
\begin{proof}
By induction over the derivation \(\csx[\flt]{\G,x:\ty[p\overlap r]{S}}{\Sigma} \ts t : \ty[q]{T}\
\EPS\). Most cases are straightforward, exploiting that qualifier substitution is monotonous w.r.t.\
\(\subq\) and that the substitute \(p\) for \(x\) consists of store locations only.
The case \rulename{t-app} critically
requires \Cref{lem:subst_commutes_overlap} for \((p \overlap q)\theta = p\theta \overlap q\theta\)
in the induction hypothesis. The case \rulename{t-sub} requires the substitution lemma for subtyping
(\Cref{lem:subst_subtyping}).
\end{proof}
\noindent Just as in \Cref{lem:subst_commutes_overlap} above, the substitution lemma imposes the
observability condition \(p \overlap \flt \subq p\overlap r\), \ie, \(t\) observes nothing more
about \(v\)'s reachability set than its assumption about \(x\), and it is oblivious of \(p\setminus r\).
That is to say, substitution ``grows'' the parameter in \rulename{t-app} with
overlap between \(p\) and the function qualifier \(r\), growing the result
by  \(p\setminus r\), realizing implicit polymorphism over qualifiers.

\subsubsection{Main Soundness Result}\label{sec:maybesoundness}

\begin{theorem}[Progress]\label{thm:progress}
 If \(\ \csx[\DOM(\Sigma)]{\varnothing}{\Sigma}  \ts t : \ty[q]{T}\ \EPS\), then either \(t\) is a value, or
 for any store \(\sigma\) where \(\csx[\DOM(\Sigma)]{\varnothing}{\Sigma} \ts \sigma\), there exists
 a term \(t'\) and store \(\sigma'\) such that \(\sigma \mid t  \redv \sigma'\mid t'\).
\end{theorem}
\begin{proof}
By induction over the derivation \(\csx[\DOM(\Sigma)]{\varnothing}{\Sigma}  \ts t : \ty[q]{T}\ \EPS\).
\end{proof}
\noindent Similar to \cite{DBLP:journals/pacmpl/BaoWBJHR21}, reduction preserves types up to qualifier growth
by fresh allocations:
\begin{theorem}[Preservation]\label{thm:soundness}\ \vspace{-8pt}
  \infrule{\csx[\DOM(\Sigma)]{\varnothing}{\Sigma}  \ts t : \ty[q]{T}\ \EPS\qquad \csx[\DOM(\Sigma)]{\varnothing}{\Sigma} \ts \sigma\qquad\sigma\mid t \redv \sigma' \mid t'
  }{
    \exists \Sigma' \supseteq \Sigma.\;\exists p \subq\DOM(\Sigma'\setminus\Sigma).\quad\csx[\DOM(\Sigma')]{\varnothing}{\Sigma'} \ts t' : \ty[q,p]{T}\ \FX{\EPS,p}\qquad
    \csx[\DOM(\Sigma')]{\varnothing}{\Sigma'} \ts \sigma'
  }
\end{theorem}
\begin{proof}
  By induction over the derivation \(\csx[\DOM(\Sigma)]{\varnothing}{\Sigma}  \ts t : \ty[q]{T}\ \EPS\).
  Most of the cases are straightforward.
  We discuss the beta reduction case of \rulename{t-app} where the substitution
  lemma (\Cref{lem:subst_term}) needs to be applied.
  To make the proof simpler, we assume explicit congruence reduction rules here.

  In this case, we have $t = (\lambda x. t_0) ~ v$ and their typings by induction hypotheses:
  $ \csx[\DOM(\Sigma)]{\varnothing}{\Sigma} \ts \lambda x. t_0 :
      \ty[q]{\left(x{\,:\,}\ty[\qsat{p}\,{\overlap}\, \qsat{q}]{T} \to^{\FX{\EPS[3]}} \ty[r]{U}\right)}\ \EPS[1] $ and
  $ \csx[\DOM(\Sigma)]{\varnothing}{\Sigma} \ts v : \ty[p]{T}\ \EPS[2] $.
  We need to show
  $$ \csx[\DOM(\Sigma)]{\varnothing}{\Sigma} \ts t_0[v/x] : (\ty[r]{U}\ \FX{\EPS, p})[p/x].$$
  Inverting the tying of the lambda value, we have the body term $t_0$ typing
  $$ \csx[q' \cup \{x\}]{x : \ty[p']{T'}}{\Sigma} \ts t_0 : \ty[r']{U'}\ \FX{\EPS'}, $$
  and
  $$ \ty[q \cap p]{T} <: \ty[p']{T'} ,
  q' <: q , \text{ and }
  \csx{x : \ty[\qsat{p} \cap \qsat{q}]{T}}{\Sigma} \ts \ty[r']{U'} \FX{\EPS'} <: \ty[r]{U} \FX{\EPS, p}.$$
  By narrowing the context and weakening the filter, we obtain a body term typing that is amenable
  to apply the substitution lemma (\Cref{lem:subst_term}):
  $$ \csx[q \cup \{x\}]{x : \ty[\qsat{p} \cap \qsat{q}]{T}}{\Sigma} \ts t_0 : \ty[r']{U'}\ \FX{\EPS'}. $$
  Then after applying \Cref{lem:subst_term}, we use \rulename{t-sub} to up-cast the result type and effect,
  which proves the goal.

\end{proof}

\begin{corollary}[Preservation of Separation]\label{coro:preservation_separation}\
Interleaved executions preserve types and disjointness:\vspace{-10pt}
  \infrule{%
{\begin{array}{l@{\qquad}l@{\qquad}l@{\qquad}l}
\csx[\DOM(\Sigma)]{\varnothing}{\Sigma} \ts t_1 : \ty[q_1]{T_1}\ \EPS[1] & \sigma\phantom{'}\mid t_1  \redv t_1' \mid \sigma' & \csx[\DOM(\Sigma)]{\varnothing}{\Sigma} \ts \sigma &  \\[1ex]
\csx[\DOM(\Sigma)]{\varnothing}{\Sigma} \ts t_2: \ty[q_2]{T_2}\ \EPS[2]  & \sigma'\mid t_2 \redv t_2' \mid \sigma''          & q_1 \overlap q_2 \subq \qbot       &
\end{array}}
}{{\begin{array}{ll@{\qquad}l@{\qquad}l}
\exists p_1\;p_2\;\EPSPR[1]\;\EPSPR[2]\;\Sigma'\;\Sigma''. & \csx[\DOM(\Sigma')\phantom{'}]{\varnothing}{\Sigma'\phantom{'}} \ts t_1' : \ty[p_1]{T_1}\ \EPSPR[1] & \Sigma'' \supseteq \Sigma' \supseteq \Sigma \\[1ex]
                                     & \csx[\DOM(\Sigma'')]{\varnothing}{\Sigma''} \ts t_2' : \ty[p_2]{T_2}\ \EPSPR[2]                     & p_1 \overlap p_2 \subq \qbot
\end{array}}}
\end{corollary}
\begin{proof}
  By sequential application of preservation (\Cref{thm:soundness}) and the fact that a reduction step
  increases the assigned qualifier by at most a fresh new location, thus preserving disjointness.
\end{proof}
\section{The Direct-Style \maybelang-Calculus with Store-Allocated Values}\label{sec:directstyle_letv}

\begin{figure}\small
\begin{mdframed}
\judgement{Introductions, Store Terms, Evaluation Contexts}{}
\[\begin{array}{l@{\ \ }c@{\ \ }l@{\qquad\qquad\ }l@{\ \ }c@{\ \ }l}
    \iota    & ::= & \lambda x.t \mid c\mid \tref_{\HLBox[gray!20]{\ell}}~\ell & & & \\
    {\sigma} & ::= & \varnothing \mid \sigma,\,\tlets~{\ell = \HLBox[gray!20]{\iota}}                        & & & \\
    {E}      & ::= & \square \mid E\ t \mid \HLBox[gray!20]{\ell}\ E \mid \tref_{E}~t \mid \tref_{\HLBox[gray!20]{\ell}}~E \mid\ !{E} \mid {E} := {t} \mid {\HLBox[gray!20]{\ell}} := {E}\mid \tlet~{x = E}~\tin~t  & & & \\
  \end{array}\]
\judgement{\HLBox[gray!20]{\text{\textbf{\textsf{Well-Formed Store Entries and Stores}}}}}{\BOX{\csx[\flt]{\G}{\Sigma} \ts \ell : \iota \in \sigma}\ \BOX{\GS[\flt]\ts \sigma}} %
\begin{mathpar}
     \inferrule{
       \Sigma(\ell) = \ty[\qbot]{\TRef~B}\qquad\csx[\flt]{\G}{\Sigma}\ts \ell' : \ty[\qbot]{B}\ \PURE\qquad\Sigma(w) = \ty[\qbot]{\Typ{Alloc}}\qquad\sigma(w) = \omega
     }{
      \csx[\flt]{\G}{\Sigma} \ts\ell: \tref_w~\ell' \in \sigma
     }
    \and
    \inferrule{
       \Sigma(\ell) = \ty[q]{T}\qquad\csx[\flt]{\G}{\Sigma}\ts \iota : \ty[q]{T}\ \PURE\qquad \forall \ell,w.\ \iota \neq \tref_w~\ell
     }{
      \csx[\flt]{\G}{\Sigma} \ts \ell : \iota \in \sigma
     }
    \and
    \inferrule{
      \vert \Sigma \vert = \vert \sigma \vert\qquad \left(\csx[\flt]{\G}{\Sigma} \ts \ell:\iota \in \sigma\right)_{\tlets~\ell=\iota\in\sigma}
    }{
       \csx[\flt]{\G}{\Sigma}\ts\sigma
    }
\end{mathpar}
\judgement{Reduction Rules}{\BOX{\sigma \mid t \redsv \sigma\mid t}}
\[\begin{array}{r@{\ \ }c@{\ \ }ll@{\ \ }r}
  \sigma,\,\tlets~{\ell_1 = \lambda x.t},\,\sigma'\mid\CX[gray]{E}{\ell_1\ \ell_2}           & \redsv & \sigma,\,\tlets~{\ell_1 = \lambda x.t},\,\sigma'\mid\CX[gray]{E}{t[\ell_2/x]} & & \rulename{$\beta$} \\[1ex]
  \sigma\mid\CX[gray]{E}{\tlet~{x = \ell}~\tin~t}                                            & \redsv & \sigma\mid\CX[gray]{E}{t[\ell/x]}                                           & & \rulename{let} \\[1ex]
  \sigma\mid\CX[gray]{E}{\iota}                                                              & \redsv & \sigma,\tlets~{\ell = \iota}\mid\CX[gray]{E}{\ell}                          & & \rulename{intro} \\
                                                                                             &        & \ell \not\in \DOM(\sigma)                                                   & &\\[1ex]
  \sigma,\,\tlets~{\ell = \tref_w~\ell'},\,\sigma'\mid\CX[gray]{E}{!\ell}                    & \redsv & \sigma,\,\tlets~{\ell = \tref_w~\ell'},\,\sigma'\mid\CX[gray]{E}{\ell'}       & & \rulename{deref} \\[1ex]
  \sigma,\,\tlets~{\ell = \tref_w~\ell'},\,\sigma'\mid\CX[gray]{E}{\ell := \ell''}           & \redsv & \sigma,\,\tlets~{\ell = \tref_w~\ell''},\,\sigma'\mid\CX[gray]{E}{\tunit}     & & \rulename{assign}
\end{array}\]
\caption{Call-by-value reduction for \maybelang with store-allocated values.}\label{fig:directstyle_letv:semantics}
\end{mdframed}
\end{figure}
 
As a first step towards transitioning into monadic normal form, we refine the previous system's
operational semantics into one that has all values in the store, \ie,
substitution becomes variable renaming, because all intermediate results are named and bound in the
store. We keep the same type system as before and show its soundness with respect to the refined
operational semantics with store-allocated values.

\subsection{Syntax}\label{sec:directstyle_letv_syntax}

We introduce a slight change to the syntax of \maybelang (\Cref{fig:maybe:syntax}) that does not
affect the typing rules, namely changing what constitutes a value and re-categorizing former values
and reference allocations as ``introductions'' \(\iota\) for store-bound entities:
\begin{align*}
  v  & ::= \HLBox[gray!20]{\ell} & \text{Values}\\
  \iota & ::= c \mid \lambda x.t\mid \tref_{\HLBox[gray!20]{\ell}}~\ell & \text{Introductions}\\
  t & ::= \cdots \mid v \mid \iota & \text{Terms}\\
  {\sigma} & ::=  \varnothing \mid \sigma,\,\tlets~{\ell = \HLBox[gray!20]{\iota}} & \text{Stores}
\end{align*}
Both mutable references and immutable constants are part of the store now, and we can discern by
types and context relations whether a location \(\ell\) may be mutated at runtime or not.

Since all constants are store-bound, we also expect that the first operand of \(\tref\) is a location
binding the allocation capability/constant \(\omega\).

\subsection{Dynamics}\label{sec:directstyle_letv_dynamics}

\Cref{fig:directstyle_letv:semantics} shows the operational semantics for \maybelang with store-allocated values.
All elimination forms, now operate on store-bound introductions. For instance, the function application rule \rulename{\(\beta\)}
replaces the call with the body of the function stored at \(\ell_1\), and passes a location \(\ell_2\) pointing to
the argument of the call. Substitution on terms simply becomes a renaming of a variable to a store location.
The new rule \rulename{intro} replaces the previous rule \rulename{ref}, generalizing it to commit any introduction
into the store at a fresh location.
\subsection{Metatheory}\label{sec:directstyle_letv_metatheory}

Since the type system has not changed, we can reuse most of the results developed in \Cref{sec:directstyle:metatheory}.

\begin{lemma}[Top-Level Term Substitution]\label{lem:subst_term_letv}\ \vspace{-8pt}
    \infrule{%
      \csx[\flt]{\G,x:\ty[p\overlap r]{S}}{\Sigma} \ts t : \ty[q]{T}\ \EPS\qquad \csx[p]{\varnothing}{\Sigma}\ts \HLBox[gray!20]{\ell} : \ty[p]{S}\ \PURE\qquad \theta = [p/x] \\[1ex]
        p\subq\dom(\Sigma)\qquad p \qglb \flt \subq p\overlap r
        }{
            \csx[\flt\theta]{\G\theta}{\Sigma} \ts t[\HLBox[gray!20]{\ell}/x] : (\ty[q]{T}\ \FX{\EPS})\theta
        }
\end{lemma}
\begin{proof}
This is a special case of \Cref{lem:subst_term}.
\end{proof}

\begin{theorem}[Preservation]\label{thm:soundness_letv}\ \vspace{-8pt}
   \infrule{\csx[\DOM(\Sigma)]{\varnothing}{\Sigma}  \ts t : \ty[q]{T}\ \EPS\qquad \csx[\DOM(\Sigma)]{\varnothing}{\Sigma} \ts \sigma\qquad\sigma\mid t \redsv \sigma' \mid t'
   }{
     \exists \Sigma' \supseteq \Sigma.\;\exists p \subq\DOM(\Sigma'\setminus\Sigma).\quad\csx[\DOM(\Sigma')]{\varnothing}{\Sigma'} \ts t' : \ty[q,p]{T}\ \FX{\EPS,p}\qquad
     \csx[\DOM(\Sigma')]{\varnothing}{\Sigma'} \ts \sigma'
   }
 \end{theorem}
 \begin{proof}
   By induction over the derivation \(\csx[\DOM(\Sigma)]{\varnothing}{\Sigma}  \ts t : \ty[q]{T}\ \EPS\).
   The proof is similar to the proof for \Cref{thm:soundness}, with the difference that typing evidence
   for operands needs to be extracted from the well-formed store \(\sigma\).
 \end{proof}

Finally, the preservation of separation \Cref{coro:preservation_separation} continues to hold in this system, with exactly the same proof. %
\section{Monadic Normal Form}\label{sec:monadic}

\begin{figure}\small
\begin{mdframed}
\judgement{Monadic Normal Form}{\BOX{\mnflang}}\vspace{-10pt}
\[\begin{array}{l@{\qquad}l@{\qquad}l@{\qquad}l}
    \V{x},\V{y},\V{z}    &::= & x \mid \ell                                                           & \text{Names}\\
    \iota                &::= & c \mid \lambda x.g  \mid \tref_{\ell}~\ell                            & \text{Introductions}\\
    v                    &::= & \ell                                                                  & \text{Values}\\
    n                    &::= & \iota \mid \V{x}~\V{x} \mid \tref_{\V{x}}~\V{x}\mid {!~\V{x}} \mid \V{x}\coloneqq \V{x}     & \text{Graph Nodes} \\
    g                    &::= & \V{x} \mid \tlet~x = b~\tin~g                                         & \text{Graph Terms}\\
    b                    &::= & n \mid g                                                              & \text{Bindings}\\
\end{array}\]\\
\typicallabel{g-let}%
\judgement{MNF Typing}{\BOX{\GS[\varphi]\tsM n : \ty[q]{T}\ \EPS}\ \BOX{\GS[\varphi]\tsM g : \ty[q]{T}\ \EPS}\ \BOX{\GS[\varphi]\tsM b : \ty[q]{T}\ \EPS}}
\begin{minipage}[t]{.5\linewidth}\vspace{0pt}
\infrule[n-cst]{\ \\
      c \in B
    }{
      \GS[\flt] \tsM c : \ty[\qbot]{B}\ \PURE
    }
\vgap
\infrule[n-abs]{
  \csx[q,x]{\G\,,\, x: \ty[p]{T}}{\Sigma} \tsM g : \ty[r]{U}\ \EPS  \\ q\subq \flt
}{
  \GS[\flt] \tsM \lambda x.g : \ty[q]{\left(x: \ty[p]{T} \to^{\EPS} \ty[r]{U}\right)}\; \PURE
}
\vgap
\infrule[n-app]{
    \V{x} : \ty[q]{\left(z: \ty[\qsat{p}\overlap \qsat{q}]{T} \to^{\EPS} \ty[r]{U}\right)}\in\GS[\flt]\\
    \V{y} : \ty[p]{T}\in\GS[\flt]\quad\theta = [p/z]\\
    z \notin\FV(U)\quad\EPS\subq q,z\quad r\subq\flt,z
  }{
    \GS[\flt]\tsM \V{x}~\V{y} : (\ty[r]{U}\ \EPS)\theta
  }
\vgap
\infrule[n-ref]{
      \V{x} : \ty[\qbot]{B}\in \GS[\flt] \\ \V{y} : \ty[q]{\Typ{Alloc}}\in\GS[\flt]
    }{
      \GS[\flt]\tsM \tref_\V{y}~\V{x} : \ty[\qbot]{(\TRef~\ty{B})}\ \FX{\V{y}}
    }
\vgap
    \infrule[n-\(!\)]{
      \V{x} : \ty[q]{(\TRef~\ty{B})}\in\GS[\flt]
    }{
      \GS[\flt]\tsM {!~\V{x}} : \ty[\qbot]{B}\ \FX{\V{x}}
    }
\end{minipage}%
\begin{minipage}[t]{.5\linewidth}\vspace{0pt}
\infrule[n-\(:=\)]{
      \V{x} : \ty[q]{(\TRef~\ty{B})} \in\GS[\flt]\\
      \V{y} : \ty[\qbot]{B}\in\GS[\flt]
    }{
      \GS[\flt]\tsM \V{x} \coloneqq \V{y} : \ty[\qbot]{\TUnit}\ \FX{\V{x}}
    }
\vgap
    \infrule[g-ret]{\ \\
      \V{x} : \ty[q]{T} \in \GS[\flt]
    }{
      \GS[\flt] \tsM \V{x} : \ty[\V{x}]{T}\ \PURE
    }
\vgap
    \infrule[g-let]{\ \\
      \GS[\flt]\tsM b : \ty[p]{S} \ \EPS[1]\\
      \csx[\flt,x]{\G\, ,\, x: \ty[\qsat{p}\cap\qsat{\flt}]{S}}{\Sigma} \tsM g : \ty[q]{T}\ \EPS[2] \\
      \theta = [p/x]\quad x \notin\FV{(T)}
    }
    {
      \GS[\varphi] \vdash_{\EMM} \Let{x}{b}{g} : (\ty[q]{T}\ \EPS[1] \EFFSEQ  \EPS[2])\theta
    }
\vgap
\infrule[b-sub]{\ \\
      \GS[\flt]\tsM b : \ty[p]{S}\ \EPS[1] \\  \GS\ts\ty[p]{S}\ \EPS[1] <: \ty[q]{T}\ \EPS[2]\\
      q,\EPS[2]\subq\flt
    }{
      \GS[\flt]\tsM b : \ty[q]{T}\ \EPS[2]
    }
\end{minipage}
\judgement{Name Lookup}{\BOX{\V{x}: \ty[q]{T}\in\GS[\flt]}}\\
\begin{minipage}[t]{.5\linewidth}\vspace{0pt}
\infrule[l-var]{x : \ty[q]{T}\in\G\qquad x\subq\flt}{x : \ty[q]{T}\in \GS[\flt] }
\end{minipage}%
\begin{minipage}[t]{.5\linewidth}\vspace{0pt}
\infrule[l-loc]{\ell : \ty[q]{T}\in\Sigma\qquad \ell\subq\flt}{\ell : \ty[q]{T}\in \GS[\flt] }
\end{minipage}%

\caption[The syntax and typing rules of the monadic normal form \mnflang.]{The syntax and typing rules of the monadic normalform \mnflang. Cf.~\Cref{fig:maybe:syntax} for the subtyping rules.}\label{fig:mnf:syntax}\label{fig:checking-vanilla-mnf}
\end{mdframed}
\end{figure}

\begin{figure}\small
    \begin{mdframed}
    \judgement{}{\BOX{t \leadsto g}}\vspace{-30pt}
    \infrule{}{c \leadsto \tlet~x = c~\tin~x}
    \infrule{}{\V{x} \leadsto \V{x}}
    \infrule{t\leadsto g}{\lambda x.t \leadsto \Let{y}{\lambda x.g}{y}}
    \infrule{t\leadsto g}{!t\leadsto\Let{x_1}{g}{\Let{x_2}{!x_1}{x_2}}}
    \infrule{t_1\leadsto g_1\qquad t_2\leadsto g_2}{\tlet~{x = t_1}~\tin~t_2\leadsto \tlet~{x = g_1}~\tin~g_2}
    \infrule{t_1\leadsto g_1\qquad t_2\leadsto g_2}{t_1~t_2\leadsto\Let{x_1}{g_1}{\Let{x_2}{g_2}{\Let{x_3}{x_1~x_2}{x_3}}}}
    \infrule{t_1\leadsto g_1\qquad t_2\leadsto g_2}{\tref_{t_1}~t_2\leadsto \Let{x_1}{g_1}{\Let{x_2}{g_2}{\Let{x_3}{\tref_{x_1}~x_2}{x_3}}}}
    \infrule{t_1\leadsto g_1\qquad t_2\leadsto g_2}{t_1 := t_2\leadsto\Let{x_1}{g_1}{\Let{x_2}{g_2}{\Let{x_3}{(x_1 := x_2)}{x_3}}}}
    \caption[Translation from \maybelang into \mnflang.]{Translation from \maybelang into \mnflang. Variable names introduced on the right-hand side are always fresh.}\label{fig:direct:mnf:translation}
    \end{mdframed}
\end{figure} %

The penultimate step towards deriving the graph IR is restricting the \directlang language to
\emph{monadic normal form} (MNF), called \mnflang (\Cref{fig:mnf:syntax}). We establish soundness of
\mnflang by (1) showing that \directlang's reduction relation with store-allocated values (\redsv,
\Cref{fig:directstyle_letv:semantics}) preserves MNF, (2) specifying provably type-preserving
translations between both languages, so that (3) we can resort to the previous section's soundness
result for \directlang.

\subsection{Syntax}\label{sec:mnf-syntax}

We make use of the syntactic category of names in places where both variables and locations
are permitted, written in typewriter font.

MNF (\Cref{fig:mnf:syntax}) is characterized by having all
intermediate results and subterms of expressions let-bound to variable names. Unlike A-normal form
(ANF), which has strictly flat sequences of let bindings with primitive operations, MNF permits
binding nested computations. A (directed, acyclic) graph can be read from graph terms \(g\), by
regarding let bindings as introducing a name for either (1) a primitive graph node labelled with a
primitive operation drawn from \(n\), or (2) naming a nested subgraph \(g\). Variable
occurrences in bound nodes correspond to edges pointing to the let binding in scope.

In this work, we choose an even stricter form of MNF than usual, \ie, sequences of let bindings
in graph terms \(g\) always end with a name. We found that this more regular form is
easier to work with when specifying optimization rules.

\subsection{Reduction Preserves MNF}\label{sec:mnf:reduction_preserves_mnf}

The reduction relation for \directlang with store-allocated values (\Cref{fig:directstyle_letv:semantics}) preserves MNF, and
can thus be restricted to obtain the call-by-value reduction relation for \mnflang:
\begin{lemma}[Reduction Preserves MNF]\label{lem:mnf:reduction_preserves_mnf}
    Let \(g\) be a graph term in \mnflang, and \(\sigma\) a store that only binds \mnflang introductions, such that \(\sigma\mid g \redsv \sigma'\mid t\)
    for some \(\sigma'\) and term \(t\). Then it holds that
    \begin{enumerate}
        \item \(\sigma'\) is a store binding only \mnflang introductions.
        \item \(t\) is a graph term of \mnflang.
    \end{enumerate}
\end{lemma}
\begin{proof}
    Since the reduction step begins with a graph term \(g\), it can be only decomposed into a redex and evaluation context
    according to \(E ::= \square \mid \Let{x}{E}{g}\), and no other cases apply. Redexes in the hole can only be bindings \(b\), \ie,
    either a node \(n\) or a nested graph term \(g'\).  It is easy to see that each possible reduction rule focuses on such a binding,
    and each rule plugs the hole with another binding \(b'\) on the right-hand side, thus preserving MNF.
    Furthermore, \(\sigma'\) is either equal to \(\sigma\), or a modification of the latter where each binding is in MNF.
\end{proof}
\subsection{Translation from Direct Style to MNF}\label{sec:mnf:translation}

This section considers the syntax-directed translation of \directlang into \mnflang
(\Cref{fig:direct:mnf:translation}).

\begin{lemma}[Type Preservation of the MNF Translation]\label{lem:type_preservation:translation}
If\ \ \(\GS[\flt]\ts t : \ty[q]{T}\ \EPS\) and \(t \leadsto g\), then \(\GS[\flt]\tsM g : \ty[q]{T}\ \EPS\).
\end{lemma}
\begin{proof}
Straightforward by induction over the typing derivation \(\GS[\flt]\ts t : \ty[q]{T}\ \EPS\).
We exemplify the proof for applications \rulename{t-app}. In this case \[t_1~t_2 \leadsto \tlet~{x_1 = g_1}~\tin~\tlet~{x_2 = g_2}~\tin~\tlet~{x_3 = x_1~x_2}~\tin~x_3 \]
where \(t_1 \leadsto g_1\) and \(t_2\leadsto g_2\).
\begin{enumerate}
  \item We have \(\GS[\flt]\vdash t_1 : \ty[q]{\left(x{\,:\,}\ty[\qsat{p}\,{\overlap}\, \qsat{q}]{T} \to^{\FX{\EPS[3]}} \ty[r]{U}\right)}\ \EPS[1]\).
  \item We have \(\GS[\flt]\ts t_2 : \ty[p]{T}\ \EPS[2]\).
  \item We have \(x \notin \FV(U)\), \(\EPS[3]\subq q,x\), \(r\subq \flt,x\), and \(\theta = [p/x]\).
  \item By IH: \(\GS[\flt]\tsM g_1 : \ty[q]{\left(x{\,:\,}\ty[\qsat{p}\,{\overlap}\, \qsat{q}]{T} \to^{\FX{\EPS[3]}} \ty[r]{U}\right)}\ \EPS[1]\).
  \item By IH:  \(\GS[\flt]\tsM g_2 : \ty[p]{T}\ \EPS[2]\).
  \item By weakening: \(\csx[\flt,x_1]{\G, x_1 : \ty[q]{\left(x{\,:\,}\ty[\qsat{p}\,{\overlap}\, \qsat{q}]{T} \to^{\FX{\EPS[3]}} \ty[r]{U}\right)}}{\Sigma}\tsM g_2 : \ty[p]{T}\ \EPS[2]\).
  \item Let \(\G' := \G, x_1 : \ty[q]{\left(x{\,:\,}\ty[p\,{\overlap}\, q]{T} \to^{\FX{\EPS[3]}} \ty[r]{U}\right)}, x_2 : \ty[p]{T}\).
  \item By rule \rulename{n-app} and (3): \(\csx[\flt,x_1,x_2]{\G'}{\Sigma}\tsM x_1~x_2 : \ty[r\theta]{U}\ \EPS[3]\theta\).
  \item By \rulename{g-let} and \rulename{g-ret}: \(\csx[\flt,x_1,x_2]{\G'}{\Sigma}\tsM \tlet~{x_3 = x_1~x_2}~\tin~x_3 : \ty[r\theta]{U}\ \EPS[3]\theta\).
  \item With (6) and \rulename{g-let}: \[\csx[\flt,x_1]{\G, x_1 : \ty[q]{\left(x{\,:\,}\ty[\qsat{p}\,{\overlap}\, \qsat{q}]{T} \to^{\FX{\EPS[3]}} \ty[r]{U}\right)}}{\Sigma}\tsM \tlet~{x_2 = g_2}~\tin~\tlet~{x_3 = x_1~x_2}~\tin~x_3 : (\ty[r\theta]{U}\ \EPS[2]\EFFSEQ\EPS[3]\theta)[p/x_2]\]
  \item With (4) and \rulename{g-let}:  \[\csx[\flt]{\G}{\Sigma}\tsM \tlet~{x_1 = g_1}~\tin~\tlet~{x_2 = g_2}~\tin~\tlet~{x_3 = x_1~x_2}~\tin~x_3 : (\ty[r\theta{[p/x_2]}]{U}\ \EPS[1]\EFFSEQ(\EPS[2]\EFFSEQ\EPS[3]\theta)[p/x_2])[q/x_1]\]
  \item Since \(x_1\) and \(x_2\) were picked fresh, and \(x\) is not free in \EPS[1] and \EPS[2] by (1) and (2), we have \[(\ty[r\theta{[p/x_2]}]{U}\ \EPS[1]\EFFSEQ(\EPS[2]\EFFSEQ\EPS[3]\theta)[p/x_2])[q/x_1] = \ty[r\theta]{U}\ \EPS[1]\EFFSEQ\EPS[2]\EFFSEQ\EPS[3]\theta = (\ty[r]{U}\ \EPS[1]\EFFSEQ\EPS[2]\EFFSEQ\EPS[3])\theta.\]
        That is, (11) proves the goal.
\end{enumerate}

\end{proof}

\subsection{Soundness}\label{sec:mnf-metatheory}

Instead of proving progress and preservation directly, we assert that terms in monadic normal form
can always be typed in the same manner in both the direct style and MNF type systems.
The intention is that we have the same type system, but restricted in the terms.

\begin{lemma}[Type-preserving Embedding of MNF Terms]\label{lem:type_preservation:translation_backwards}\hfill
\begin{enumerate}
        \item \(\GS[\flt]\tsM n : \ty[q]{T}\ \EPS\) iff \(\GS[\flt]\ts n : \ty[q]{T}\ \EPS\).
        \item \(\GS[\flt]\tsM g : \ty[q]{T}\ \EPS\) iff \(\GS[\flt]\ts g : \ty[q]{T}\ \EPS\).
\end{enumerate}
\end{lemma}
\begin{proof}
Each direction is proved by mutual induction over the respective typing derivation.
\end{proof}
Together with \Cref{lem:mnf:reduction_preserves_mnf}, it follows
that the type soundness and preservation of separation results of the direct style system
(\Cref{sec:directstyle_letv_metatheory}) carry over to the MNF version.
\begin{corollary}[MNF Progress]\label{thm:mnf:progress}
    If \(\ \csx[\DOM(\Sigma)]{\varnothing}{\Sigma}  \tsM g : \ty[q]{T}\ \EPS\), then either \(g\) is a value, or
    for any store \(\sigma\) where \(\csx[\DOM(\Sigma)]{\varnothing}{\Sigma} \ts \sigma\), there exists
    a graph term \(g'\) and store \(\sigma'\) such that \(\sigma \mid g  \redsv \sigma'\mid g'\).
\end{corollary}
\begin{corollary}[MNF Preservation]\label{thm:mnf:soundness}\ \vspace{-8pt}
  \infrule{\csx[\DOM(\Sigma)]{\varnothing}{\Sigma}  \tsM g : \ty[q]{T}\ \EPS\qquad  \csx[\DOM(\Sigma)]{\varnothing}{\Sigma} \ts \sigma\qquad\sigma\mid g \redsv \sigma' \mid g'
  }{
    \exists \Sigma' \supseteq \Sigma.\;\exists p \subq\DOM(\Sigma'\setminus\Sigma).\quad\csx[\DOM(\Sigma')]{\varnothing}{\Sigma'} \tsM g' : \ty[q,p]{T}\ \FX{\EPS,p}\qquad
    \csx[\DOM(\Sigma')]{\varnothing}{\Sigma'} \ts \sigma'
  }
\end{corollary}
\begin{corollary}[MNF Preservation of Separation]\label{coro:mnf:preservation_separation}
    Interleaved executions preserve types and disjointness:\vspace{-8pt}
      \infrule{%
    {\begin{array}{l@{\qquad}l@{\qquad}l@{\qquad}l}
    \csx[\DOM(\Sigma)]{\varnothing}{\Sigma} \tsM g_1 : \ty[q_1]{T_1}\ \EPS[1] & \sigma\phantom{'}\mid g_1\redsv  \sigma'\phantom{'}\mid  g_1'  & \csx[\DOM(\Sigma)]{\varnothing}{\Sigma} \ts \sigma &  \\[1ex]
    \csx[\DOM(\Sigma)]{\varnothing}{\Sigma} \tsM g_2: \ty[q_2]{T_2}\ \EPS[2]  & \sigma'\mid g_2  \redsv \sigma''\mid g_2'           & q_1 \overlap q_2 \subq \qbot       &
    \end{array}}
    }{{\begin{array}{ll@{\qquad}l@{\qquad}l}
    \exists p_1\;p_2\;\EPSPR[1]\;\EPSPR[2]\;\Sigma'\;\Sigma''. & \csx[\DOM(\Sigma')\phantom{'}]{\varnothing}{\Sigma'\phantom{'}} \tsM g_1' : \ty[p_1]{T_1}\ \EPSPR[1] & \Sigma'' \supseteq \Sigma' \supseteq \Sigma \\[1ex]
                                         & \csx[\DOM(\Sigma'')]{\varnothing}{\Sigma''} \tsM g_2' : \ty[p_2]{T_2}\ \EPSPR[2]                     & p_1 \overlap p_2 \subq \qbot
    \end{array}}}
    \end{corollary} %
\section{Monadic Normal Form with Hard Dependencies}\label{sec:monad-norm-with}\label{sec:graphir}

\begin{figure}
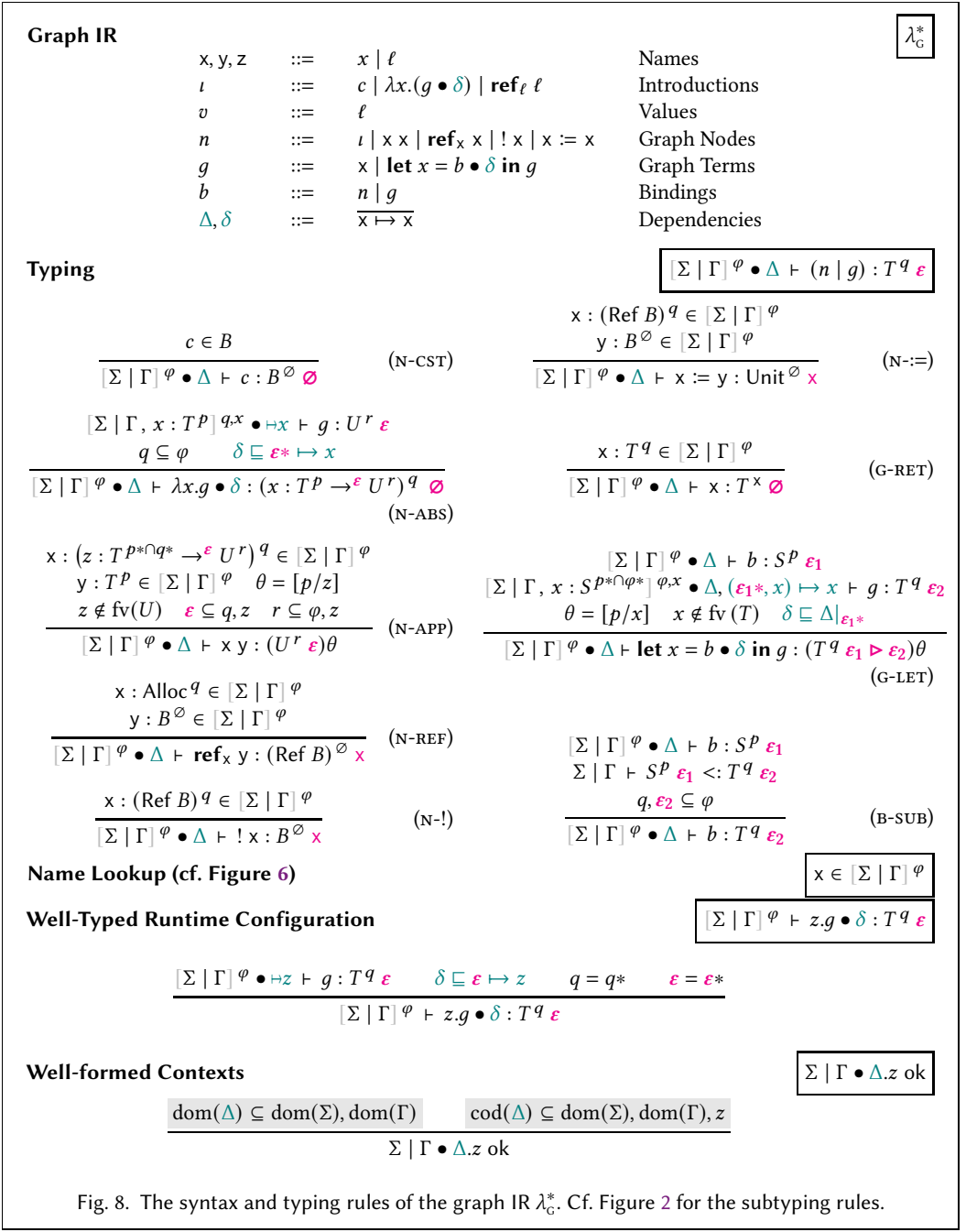
\small
\begin{mdframed}
\judgement{Graph IR}{\BOX{\irlang}}\vspace{-10pt}
\[\begin{array}{l@{\qquad}l@{\qquad}l@{\qquad}l}
    \V{x},\V{y},\V{z} &::= & x \mid \ell                                                                             & \text{Names}\\
    \iota             &::= & c \mid \lambda x.(g\has\DEP)  \mid \tref_{\ell}~\ell                                    & \text{Introductions}\\
    v                 &::= & \ell                                                                                    & \text{Values}\\
    n                 &::= & \iota \mid \V{x}~\V{x} \mid \tref_{\V{x}}~\V{x}\mid {!~\V{x}} \mid \V{x}\coloneqq \V{x} & \text{Graph Nodes} \\
    g                 &::= & \V{x} \mid \tlet~x = b\has\DEP~\tin~g                                                   & \text{Graph Terms}\\
    b                 &::= & n \mid g                                                                                & \text{Bindings}\\
    \DELTA,\DEP       &::= & \seq{\V{x}\mapsto \V{x}}                                                                & \text{Dependencies}\\
\end{array}\]\\
\typicallabel{g-let}%
\judgement{Typing}{\BOX{\GSD[\varphi]\ts (n\mid g) : \ty[q]{T}\ \EPS}}
\begin{minipage}[t]{.47\linewidth}\vspace{0pt}
\infrule[n-cst]{\ \\
      c \in B
    }{
      \GSD[\flt]\ts c : \ty[\qbot]{B}\ \PURE
    }
\vgap
\infrule[n-abs]{
  \csx[q,x]{\G\,,\, x: \ty[p]{T}}{\Sigma}\has\mute{\pointsto{x}} \ts g : \ty[r]{U}\ \EPS  \\ q\subq \flt\qquad
  \mute{\DEP\sqsubseteq\FX{\qsat{\EPS}}\mapsto x}
}{
  \GSD[\flt]\ts \lambda x.g\has\DEP : \ty[q]{\left(x: \ty[p]{T} \to^{\EPS} \ty[r]{U}\right)}\; \PURE
}
\vgap
\infrule[n-app]{
    \V{x} : \ty[q]{\left(z: \ty[\qsat{p}\overlap \qsat{q}]{T} \to^{\EPS} \ty[r]{U}\right)}\in\GS[\flt]\\
    \V{y} : \ty[p]{T}\in\GS[\flt]\quad\theta = [p/z]\\
    z \notin\FV(U)\quad\EPS\subq q,z\quad r\subq\flt,z\\
  }{
    \GSD[\flt]\ts \V{x}~\V{y} : (\ty[r]{U}\ \EPS)\theta
  }
\vgap
\infrule[n-ref]{
      \V{x} : \ty[q]{\Typ{Alloc}}\in\GS[\flt]\\\V{y} : \ty[\qbot]{B}\in \GS[\flt]
    }{
      \GSD[\flt]\ts \tref_\V{x}~\V{y} : \ty[\qbot]{(\TRef~\ty{B})}\ \FX{\V{x}}
    }
\vgap
    \infrule[n-\(!\)]{
      \V{x} : \ty[q]{(\TRef~\ty{B})}\in\GS[\flt]
    }{
      \GSD[\flt]\ts {!~\V{x}} : \ty[\qbot]{B}\ \FX{\V{x}}
    }
\end{minipage}%
\begin{minipage}[t]{.03\linewidth}
\hspace{1pt}%
\end{minipage}%
\begin{minipage}[t]{.5\linewidth}\vspace{0pt}
\infrule[n-\(:=\)]{
      \V{x} : \ty[q]{(\TRef~\ty{B})} \in\GS[\flt]\\
      \V{y} : \ty[\qbot]{B}\in\GS[\flt]
    }{
      \GSD[\flt]\ts \V{x} \coloneqq \V{y} : \ty[\qbot]{\TUnit}\ \FX{\V{x}}
    }
\vgap
    \infrule[g-ret]{\ \\
      \V{x} : \ty[q]{T} \in \GS[\flt]
    }{
      \GSD[\flt] \ts \V{x} : \ty[\V{x}]{T}\ \PURE
    }
\vgap
    \infrule[g-let]{\ \\
      \GSD[\flt]\ts b : \ty[p]{S} \ \EPS[1]\\
      \cdsx[\flt,x]{\G\, ,\, x: \ty[\qsat{p}\cap\qsat{\flt}]{S}}{\Sigma}{\DELTA,\mute{(\FX{\qsat{\EPS[1]}},x)\mapsto x}} \ts g : \ty[q]{T}\ \EPS[2] \\
      \theta = [p/x]\quad x \notin\FV{(T)}\quad\mute{\DEP\sqsubseteq\DELTA\vert_{\FX{\qsat{\EPS[1]}}}}
    }
    {
      \GSD[\varphi] \vdash \Let{x}{b\has\DEP}{g} : (\ty[q]{T}\ \EPS[1] \EFFSEQ  \EPS[2])\theta
    }
\vgap
\infrule[b-sub]{\ \\
      \GSD[\flt]\ts b : \ty[p]{S}\ \EPS[1] \\  \GS\ts\ty[p]{S}\ \EPS[1] <: \ty[q]{T}\ \EPS[2]\\
      q,\EPS[2]\subq\flt
    }{
      \GSD[\flt]\ts b : \ty[q]{T}\ \EPS[2]
    }
\end{minipage}
\judgement{Name Lookup (cf. \Cref{fig:checking-vanilla-mnf})}{\BOX{\V{x}\in\GS[\flt]}}\\
\judgement{Well-Typed Runtime Configuration}{\BOX{\GS[\flt]\ts z.g\has\DEP : \ty[q]{T}\ \EPS}}
\infrule{
  \cdsx[\flt]{\G}{\Sigma}{\mute{\pointsto{z}}}\ts g : \ty[q]{T}\ \EPS \qquad \mute{\DEP\sqsubseteq\EPS\mapsto z}\qquad q=\qsat{q}\qquad\EPS=\qsat{\EPS}
  }{
  \GS[\flt]\ts z.g\has\DEP : \ty[q]{T}\ \EPS
}
\judgement{Well-formed Contexts}{\BOX{\WF{\GSD.z}}}\vspace{-8pt}
\infrule{
  \HLBox[gray!20]{\DOM(\DELTA)\subq\DOM(\Sigma),\DOM(\G)}\qquad \HLBox[gray!20]{\CODOM(\DELTA)\subq\DOM(\Sigma),\DOM(\G),z}
}{
\WF{\GSD.z}
}
\caption[The syntax and typing rules of the graph IR \irlang.]{The syntax and typing rules of the graph IR \irlang. Cf.~\Cref{fig:maybe:syntax} for the subtyping rules.}\label{fig:graphir:syntax}\label{fig:checking-graphir-mnf}\label{fig:graphir:checking}
\end{mdframed}
\end{figure}

\begin{figure}\small
\begin{mdframed}
\typicallabel{\(\leadsto\)-let}%
\judgement{Dependency Synthesis}{\BOX{\GSD[\varphi]\ts (n\mid g) : \ty[q]{T}\ \EPS\yields{(\bm{n}\mid\bm{g})\has\DEP}}}
\begin{minipage}[t]{.47\linewidth}\vspace{0pt}
\infrule[\(\leadsto\)-cst]{
      c \in B
    }{
      \GSD[\flt]\ts c : \ty[\qbot]{B}\ \PURE\yields{c\has\NODEP}
    }
\vgap
\infrule[\(\leadsto\)-abs]{
  \csx[q,x]{\G\,,\, x: \ty[p]{T}}{\Sigma}\has\mute{\pointsto{x}} \ts g : \ty[r]{U}\ \EPS \yields{ \bm{g}\has\DEP}  \\ q\subq \flt
}{
  \GSD[\flt]\ts \lambda x.g : \ty[q]{\left(x: \ty[p]{T} \to^{\EPS} \ty[r]{U}\right)}\; \PURE\\\yields{(\lambda x.\bm{g}\has\DEP)\has\NODEP}
}
\vgap
\infrule[\(\leadsto\)-app]{
    \V{x} : \ty[q]{\left(z: \ty[\qsat{p}\overlap \qsat{q}]{T} \to^{\EPS} \ty[r]{U}\right)}\in\GS[\flt]\\
    \V{y} : \ty[p]{T}\in\GS[\flt]\quad\theta = [p/z]\\
    z \notin\FV(U)\quad\EPS\subq q,z\quad r\subq\flt,z
  }{
    \GSD[\flt]\ts \V{x}~\V{y} : (\ty[r]{U}\ \EPS)\theta\\\yields{\V{x}~\V{y}\has\DELTA\vert_{\FX{\qsat{(\EPS\FX{\theta})}}}}
  }
\vgap
\infrule[\(\leadsto\)-ref]{
  \V{x} : \ty[q]{\Typ{Alloc}}\in\GS[\flt]\\\V{y} : \ty[\qbot]{B}\in \GS[\flt]
    }{
      \GSD[\flt]\ts \tref_\V{x}~\V{y} : \ty[\qbot]{(\TRef~\ty{B})}\ \FX{\V{x}}\\\yields{\tref_\V{x}~\V{y}\has\DELTA\vert_{\FX{\qsat{\V{x}}}}}
    }
\vgap
    \infrule[\(\leadsto\)-\(!\)]{
      \V{x} : \ty[q]{(\TRef~\ty{B})}\in\GS[\flt]
    }{
      \GSD[\flt]\ts {!~\V{x}} : \ty[\qbot]{B}\ \FX{\V{x}}\\\yields{{!~\V{x}}\has\DELTA\vert_{\FX{\qsat{\V{x}}}}}
    }
\end{minipage}%
\begin{minipage}[t]{.03\linewidth}
\hspace{1pt}%
\end{minipage}%
\begin{minipage}[t]{.5\linewidth}\vspace{0pt}
\infrule[\(\leadsto\)-\(:=\)]{
      \V{x} : \ty[q]{(\TRef~\ty{B})} \in\GS[\flt]\\
      \V{y} : \ty[\qbot]{B}\in\GS[\flt]
    }{
      \GSD[\flt]\ts \V{x} \coloneqq \V{y} : \ty[\qbot]{\TUnit}\ \FX{\V{x}}\\\yields{\V{x} \coloneqq \V{y}\has\DELTA\vert_{\FX{\qsat{\V{x}}}}}
    }
\vgap
    \infrule[\(\leadsto\)-ret]{\ \\
      \V{x} : \ty[q]{T} \in \GS[\flt]
    }{
      \GSD[\flt] \ts \V{x} : \ty[\V{x}]{T}\ \PURE\yields{\V{x}\has\NODEP}
    }
\vgap
    \infrule[\(\leadsto\)-let]{
      \GSD[\flt]\ts b : \ty[p]{S} \ \EPS[1]\yields{\bm{b}\has\DEP[1]}\\
      \cdsx[\flt,x]{\G\, ,\, x: \ty[\qsat{p}\cap\qsat{\flt}]{S}}{\Sigma}{\DELTA,\mute{(\FX{\qsat{\EPS[1]}},x)\mapsto x}} \ts g : \ty[q]{T}\ \EPS[2] \\
      \yields{\bm{g}\has\DEP[2]}\\
      \theta = [p/x]\quad x \notin\FV{(T)}
    }
    {
      \GSD[\varphi] \vdash \Let{x}{b}{g} : (\ty[q]{T}\ \EPS[1] \EFFSEQ  \EPS[2])\theta\\\yields{\left(\Let{x}{\bm{b}\has\DEP_1}{\bm{g}}\right)\has\DEP[1],\DEP_2[x\leadsto\DELTA\vert_{\qsat{p}}]}
    }
\vgap
\infrule[\(\leadsto\)-sub]{\ \\
      \GSD[\flt]\ts b : \ty[p]{S}\ \EPS[1]\yields{\bm{b}\has\DEP[1]} \\  \GS\ts\ty[p]{S}\ \EPS[1] <: \ty[q]{T}\ \EPS[2]\\
      q,\EPS[2]\subq\flt
    }{
      \GSD[\flt]\ts b : \ty[q]{T}\ \EPS[2]\yields{\bm{b}\has\DELTA\vert_{\FX{\qsat{\EPS[2]}}}}
    }
\end{minipage}
\judgement{Name Lookup (cf. \Cref{fig:checking-vanilla-mnf})}{\BOX{\V{x}\in\GS[\flt]}}
\caption{Type-and-effect-directed dependency synthesis from \mnflang into \irlang.}\label{fig:graphir:synthesis}
\end{mdframed}
\end{figure}

\begin{figure}\small
\begin{mdframed}
\judgement{Store Terms, Graph Term Contexts, Binding Contexts}{}
\[\begin{array}{l@{\ \ }c@{\ \ }l@{\qquad\qquad\ }l@{\ \ }c@{\ \ }l}
    {\sigma} & ::= & \varnothing \mid \sigma,\,\tlets~{\ell = \iota}     & & & \\
    {G}      & ::= & \square\has\DEP \mid (\tlet~{x = G}~\tin~g)\has\DEP & & & \\
    {B}      & ::= & (\tlet~{x = \square}~\tin~g)\has\DEP \mid (\tlet~{x = B}~\tin~g)\has\DEP         & & & \\
  \end{array}\]
\judgement{Reduction Rules}{\BOX{\sigma \mid  z.g\has\DEP \redg \sigma\mid z.g\has\DEP}}
\[\begin{array}{r@{\ \ }c@{\ \ }ll@{\,}r}
\sigma\mid z.\CX[gray]{B}{\ell_1\ \ell_2\has\DEP[1]}                                                    & \redg & \sigma\mid z.\CX[gray]{B}{(g\has\DEP[2])\mute{[x\leadsto \DEP[1] ]}[\ell_2/x][\ell_2/x]_{\mathsf{t}}} & & \rulename{$\beta$} \\
                                                                                                        &       & \sigma = \sigma_1,\,\tlets~{\ell_1 = \lambda x.g\has\DEP[2]},\,\sigma_2                               & & \\
                                                                                                        &       & \DOM(\DEP[1])\subq\DOM(\sigma),\ \CODOM(\DEP[1]) \subq \{z\}                                                                           & & \\[1ex]
\sigma\mid z.\CX[gray]{G}{\tlet~{x = \ell\has\DEP}~\tin~g}                                              & \redg & \sigma\mid z.\CX[gray]{G}{g\mute{[x\leadsto\DEP]}[\ell/x][\ell/x]_{\mathsf{t}}}                       & & \rulename{let} \\
                                                                                                        &       & \DOM(\DEP)\subq\DOM(\sigma),\ \CODOM(\DEP) \subq \{z\}                                                                              & & \\[1ex]
\sigma\mid z.\CX[black]{G}{\tlet~{x = \iota\has\DEP}~\tin~g}                                            & \redg & \sigma,\tlets~{\ell = \iota}\mid z.\CX[black]{G'}{g\mute{[x\leadsto\DEP]}[\ell/x][\ell/x]_{\mathsf{t}}}                    & & \rulename{intro} \\
                                                                                                        &       & \ell \not\in \DOM(\sigma), \ \CODOM(\DEP) \subq \{z\},\ G' = G\langle \iota:z.\ell\rangle             & & \\[1ex]
\sigma,\,\tlets~{\ell = \tref_w~\ell'},\,\sigma'\mid z.\CX[gray]{B}{!\ell\has\DEP}                      & \redg & \sigma,\,\tlets~{\ell = \tref_w~\ell'},\,\sigma'\mid z.\CX[gray]{B}{\ell'\has\DEP}                    & & \rulename{deref} \\
                                                                                                        &       & \DOM(\DEP)\subq\DOM(\sigma),\ \CODOM(\DEP) \subq \{z\}                                                                              & & \\[1ex]
\sigma,\,\tlets~{\ell = \tref_w~\ell'},\,\sigma'\mid z.\CX[gray]{B}{\ell := \ell''\has\DEP}             & \redg & \sigma,\,\tlets~{\ell = \tref_w~\ell''},\,\sigma'\mid z.\CX[gray]{B}{\tunit\has\DEP}                  & & \rulename{assign}\\
                                                                                                        &       & \DOM(\DEP)\subq\DOM(\sigma),\ \CODOM(\DEP) \subq \{z\}                                                                              & &
\end{array}\]
\judgement{Contextual Effect Propagation}{\BOX{G\langle \iota: z.\ell\rangle}}
\[\begin{array}{l@{\ \ }c@{\ \ }l}
 G\langle  c : z.\ell \rangle & = & G\\
 G\langle  \lambda x.g\has\DEP : z.\ell \rangle & = & G\\
 (\square\has\DEP)\langle \tref_w~\ell_1 : z. \ell_2 \rangle & = & \square\has\mute{\DEP,\ell_2 \mapsto z}\\
 ((\tlet~{x = G}~\tin~g)\has\DEP)\langle \tref_w~\ell_1 : z. \ell_2 \rangle & = & (\tlet~{x = G\langle \tref_w~\ell_1 : z. \ell_2 \rangle}~\tin~g)\has\mute{\DEP,\ell_2 \mapsto z}\\
\end{array}\]
\judgement{Well-Formed Store Entries and Stores}{\BOX{\GSD[\flt] \ts \ell : \iota \in \sigma}\ \BOX{\GSD[\flt]\ts \sigma}} %
\begin{mathpar}
     \inferrule{
       \Sigma(\ell) = \ty[\qbot]{\TRef~B}\qquad\GSD[\flt]\ts \ell' : \ty[\qbot]{B}\ \PURE\qquad\Sigma(w) = \ty[\qbot]{\Typ{Alloc}}\qquad\sigma(w) = \omega
     }{
      \GSD[\flt] \ts\ell: \tref_w~\ell' \in \sigma
     }
    \and
    \inferrule{
       \Sigma(\ell) = \ty[q]{T}\qquad\GSD[\flt]\ts \iota : \ty[q]{T}\ \PURE\qquad \forall \ell,w.\ \iota \neq \tref_w~\ell
     }{
      \GSD[\flt] \ts \ell : \iota \in \sigma
     }
    \and
    \inferrule{
      \vert \Sigma \vert = \vert \sigma \vert\qquad \left(\GSD[\flt] \ts \ell:\iota \in \sigma\right)_{\tlets~\ell=\iota\in\sigma}
    }{
       \GSD[\flt]\ts\sigma
    }
\end{mathpar}
\caption{Call-by-value reduction for \irlang with runtime dependency checking.}\label{fig:graphir:cbv}
\end{mdframed}
\end{figure}
 
This section presents our graph IR \irlang which enriches the monadic normal form (MNF) of the previous section
with effect dependencies.

\subsection{Dependencies}\label{sec:hard-dependencies-intro}
Data dependencies are expressed by ordinary variable occurrences in terms. Tracking effect dependencies requires
extra term annotations.
Intuitively, a (hard) effect dependency \(x \mapsto y\) indicates that an effect
on the node \(x\) (\eg, a reference, or global capability) is induced, and that node \(y\) is the previous
node in the graph at which an effect on \(x\) occurred. That is to say, \(x \mapsto y\) \emph{does not}
indicate there is an edge between \(x\) and \(y\), but rather that there is an edge \emph{from the current node}
to which \(x \mapsto y\) is attached pointing to \(y\) and its (effect) label is \(x\).
Oftentimes we will just say ``dependency'' and omit the ``hard'' and ``effect'' qualifiers.

We bundle these dependencies into finite maps \(\DEP\) from variables to variables, and annotate them
to atomic nodes or nested graphs at let bindings.
Dependencies come with a standard "update" operator which is associative:
\[(\DEP[1],\DEP[2])(x) := \begin{cases} \DEP[2](x) & x\in\DOM(\DEP[2])\\ \DEP[1](x) & \text{otherwise.}\end{cases}\]
and a restriction operator of the domain to a given set (abusing set notation):
\[\DEP\vert_{\alpha} := \{x \mapsto y\in\DEP\mid x\in\alpha \},\]
and we also define two removal operators, i.e.,
\[\DEP - \alpha := \{x \mapsto y \in\DEP \mid y \notin \alpha \},\]
which removes all mappings pointing into \(\alpha\), and
\[\DEP \setminus \alpha := \DEP\vert_{\DOM(\DEP)\setminus \alpha},\]
which removes all entries pointing from \(\alpha\).
In symmetry with substitutions on terms and qualifiers, there is a notion of substitution (or rewiring, rerouting), over dependencies:
\[ \DEP[1][x \leadsto \DEP[2]] := \DEP[1] - \{x\}, \{y\mapsto z \in \DEP[2] \mid y\mapsto x \in \DEP[1]\},\]
which is rerouting the dependency target \(x\) via \(\DEP[2]\). For instance,
\[\mute{(a \mapsto b, c \mapsto x, y \mapsto x)[x \leadsto (c\mapsto a, y\mapsto b)]} = \mute{(a \mapsto b, c \mapsto a, y \mapsto b)}. \]
That is, rewiring substitutes some of the targets in dependency mappings, which happens when the variable \(x\) is replaced, e.g.,
when \(x\) is the formal parameter of a function at a call site. The function body may have dependencies pointing to \(x\), and these
must be rerouted using the call-site's dependencies. Dually, \(x\) might also be in the domain of a dependency, and the monadic
semantics will rename it to some store location holding the value of the call site argument. For this case, we also
lift qualifier substitutions to act on the domain of dependencies, i.e.,
\[\DEP{[q/x]} := \DEP\setminus x, \{ y\mapsto \DEP(x) \mid y \in q \}\]
so that all variables and locations in \(q\) will point to \(x\)'s target, if it has an entry.

\paragraph{\textbf{Effect-dependency calculation}}
Static effects \(\EPS\) and contextual information on last uses determine effect dependencies \(\DEP\) at graph nodes.
We outline the basic process in the following.

A few auxiliary definitions are in order:
\begin{align*}
   \alpha \mapsto x & := \{ y \mapsto x \mid y\in \alpha\}  \\
   \Sigma,\G \mapsto x &:= \DOM(\Sigma)\cup\DOM(\G)\mapsto x\\
   {\color{gray}\GS\vdash\,}\pointsto{x} & := \Sigma,\G \mapsto x,\ \Sigma,\G\ \text{are implicit}\\
   {\color{gray}\GS\vdash\,}\DEP\uparrow^z & := \DEP, \{ x \mapsto z \mid x \in (\DOM(\Sigma)\cup\DOM(\G))\setminus\DOM(\DEP)\},\ \Sigma,\G\ \text{are implicit}
\end{align*}
Firstly, we overload the finite maps notation to specify a map with domain \(\alpha\) pointing to
the single variable \(x\), and analogously the map pointing all the variables and locations bound in
\(\G\) and \(\Sigma\) to \(x\), for which we often use the shorthand notation \(\pointsto{x}\)
omitting those contexts when they are unambiguous. The \(\uparrow^z\) operator is used to ensure
that a dependency is defined for at least the variables and location in context, adding entries
pointing to a common block start variable \(z\).

If a nested graph
named \(x\) has effect \(\EPS\), then we record that the reachable variables affected by \(\FX{\qsat{\EPS}}\)\footnote{Since the systems in this report \emph{lazily} assign qualifiers and effects, we need to consider the transitive reachability closure 
(\Cref{fig:saturation_overlap}) to get ahold of all relevant dependencies.}
were last used at \(x\) (plus \(x\) was last used at itself, explained below). So if
\(\mute{\Delta}\) records the last use of each variable in scope (essentially a structural coeffect
which is also just a finite map from variables to variables), then we update it to \(\mute{\Delta,
(\FX{\qsat{\EPS}},x\mapsto x)} \). Synthesizing the dependency for node \(x\) itself is just a
matter of projecting the last-use coeffect \(\DELTA\).

Concretely, consider \(\tlet~{x = g_1}~\tin~g_2\) with last uses \(\mute{\Delta}\), and let \(\EPS\)
be the effect of \(g_1\). We first calculate the annotated version of \(g_1\):
\[g_1 \leadsto g_1' \has \mute{\Delta\vert_{\FX{{\qsat{\EPS}}}}}\] which attaches the current last uses with
respect to \(\EPS\). When proceeding with the continuation \(g_2\), \(x\) is added into the context,
and we need to define its last-use coeffect, which is simply \(x\) itself.
We will discuss the precise calculation rules of our type system in the next section.

\subsection{Syntax and Statics}\label{sec:hard-statics}

\Cref{fig:checking-graphir-mnf} shows the type checking relation,
and it is easy too see that the \irlang-calculus is identical to the \mnflang-calculus (\Cref{fig:checking-vanilla-mnf}),
when erasing all the \mute{teal} parts pertaining to dependencies. At the term level, we attach
dependencies \(\DEP\) to let bindings (representing the dependencies for all the effects of the bound graph term) and
the body of \(\lambda\)-abstractions (representing the dependencies of the abstraction's latent effects)

\subsubsection{Type Checking}
The typing judgment now carries an additional dependency map \(\DELTA\) attached to the context
(basically a form of coeffect), which is used to track the \emph{last uses} of variables and
locations in the context/store. We stipulate that at all times the domain of \(\DELTA\) ranges over
the domain of \(\G\) and \(\Sigma\). Last uses are threaded as an input through typing derivations,
and the only rules at which they are accessed are those for terms with dependency annotations, \ie,
\rulename{n-abs} and \rulename{g-let}. Those annotated dependencies should always conform to the
effect of the term in question, and as a rule of thumb, we regard the effect \(\EPS\)'s transitive
reachability closure (a set of variables and locations) as a slice of the currently known last uses
\(\DELTA\), \ie, restricting the domain of \(\DELTA\) to the effect in question yields all the
relevant dependencies at the current node.

In the rule \rulename{n-abs} for \(\lambda\)-abstractions, we check the body with all the last uses
pointing to the formal parameter \(x\). This is because we generally do not know the call site and
actual argument in advance, and the most natural choice is abstracting the last use of the free variables in
the body by \(x\) in symmetry to term abstraction. As we have motivated in the main paper and will shortly see
in the operational semantics (\Cref{sec:hard-dynamics}), the ``latent'' dependency \(\DEP\) annotated to the
function's body has to be rewired at the call site.
Following the ``rule of thumb'' above, we check that the annotated
dependency \(\DEP\) is conforming to the body's dependency in relation to its last uses, \ie, \(\DEP\) is
a sub-map of \(\mute{\FX{\qsat{\EPS}}\mapsto x}\).

Similarly, rule \rulename{g-let} checks that the annotated dependency for the bound node/nested
graph is conforming to its effect, \ie, it is a sub-map of the last uses for \(\FX{\qsat{\EPS[1]}}\)
in \(\DELTA\). When typing the let body \(g\), we update the last uses for the variables affected by
\(\FX{\qsat{\EPS[1]}}\) and let them point to the binding \(x\), precisely because those variables
have been last used here. Furthermore, since \(x\) is newly introduced in the body \(g\), we also
have to specify a last use for it, which is \(x\) itself.

\subsubsection{Dependency Synthesis}
While the typing relation provides a means to check dependencies, it does not provide a method to
compute them. For this purpose we define a type/qualifier/effect-directed synthesis relation
(\Cref{fig:graphir:synthesis}) for \mnflang terms which lack any dependency annotations and produces
dependency-annotated \irlang terms along with the dependency map for effects on free variables as
output, given an initial map \(\DELTA\) of last uses and typing context as input.
Synthesizing the dependencies follows the ``rule of thumb'' from above, \ie, synthesized and inserted
dependencies are always the currently known last uses of the variables in the term's effect
(cf.\ \Cref{sec:graphir:synthmeta}).

\subsection{Dynamics}\label{sec:hard-dynamics} The call-by-value operational semantics for \irlang
(\Cref{fig:graphir:cbv}) is a refinement of the operational semantics for \mnflang with
store-allocated values (\Cref{fig:directstyle_letv:semantics}). The changes are twofold: (1)
dependencies are part of the term syntax now, and have to be accounted for by reductions, and (2)
the semantics additionally checks for each reduction step whether all the effect dependencies have
been already evaluated and committed to the store. By type soundness
(cf.~\Cref{sec:hard-metatheory}), well-typed terms do not exhibit dependency violations, \ie,
dependencies correctly reflect the observed runtime execution order of effect operations.

Reduction occurs over runtime configurations \(\sigma \mid z.g \has\DEP\), which compared to the
\mnflang-calculus attaches a distinguished \emph{start variable} \(z\) to the graph term \(g\) along
with its dependency \(\DEP\). We stipulate that \(z\) is always chosen so that it is not a free
variable of \(g\). It is a mechanism to check whether a dependency has already been evaluated. That
is, at the top level, we set the initial use of all the free variables/locations to \(z\), and that
will be reflected in the synthesized dependencies of a term. The invariant is that the next
operation will have all its dependencies purely on store variables and these will point to \(z\),
which is outside of the program. The meaning is that all dependencies of the current node are in the
store, and each of the reduction rules checks this property. In the following, we discuss the
changes made to the reduction rules compared to \Cref{fig:directstyle_letv:semantics}.

In the function application rule \rulename{\(\beta\)}, we now have to account for the latent
dependencies of the function and the dependencies at the call site. Thus, in symmetry with dependent
function application, we rewire the function body and its dependency \(\DEP[2]\) with \(\DEP[1]\)
for \(x\). Since dependencies are type-level information annotated in the term syntax, we also have
to perform the qualifier substitution part of the static dependent function application, \ie, we
substitute the argument's location \(\ell_2\) for \(x\) in the domain of dependencies, if present.
To distinguish qualifier substitution in dependencies from term substitution/renaming of variables,
we attach the subscript \(\mathsf{t}\) to the latter. Hence, substitution becomes simultaneous
rewiring, qualifier substitution, and renaming on terms in the graph IR.
Accordingly, we do the same in \rulename{let} and \rulename{intro}.

Rule \rulename{intro} has changed compared to \Cref{fig:directstyle_letv:semantics} in that it
simultaneously performs the \rulename{let} elimination along with the introduction of the fresh
location. The reason is that from the perspective of static typing,
introducing the new location will make it appear in reachability qualifiers and effects, and
consequently in dependencies. So the reduction step has to patch up the dependencies in the
body of the let expression, because \(x\) now aliases \(\ell\), which it did not beforehand.
Furthermore, if \(\iota\) is a mutable reference, using it makes it also appear in the codomain of
effects dependencies, and this change needs propagating into all dependencies along the spine of
the evaluation context up to the top level, because \(\ell\) is globally visible.

Finally, modulo checking effect dependencies, rules \rulename{deref} and \rulename{assign} have not
changed.

\subsection{Metatheory}\label{sec:hard-metatheory}

\subsubsection{Properties of Dependency Synthesis}\label{sec:graphir:synthmeta}

Dependency synthesis induces a function over MNF typing derivations which given an input map of last
uses always produce an annotated graph IR term with the same type, qualifier, and effect. As a
corollary, we obtain a type/effect/qualifier-preserving and dependency-synthesizing translation from
the direct style \directlang system into the \irlang graph IR.

\begin{lemma}[Synthesis Invariant]\label{lem:deps-invariant}
  Dependencies are completely determined by the context and effect, as follows:
  \begin{enumerate}
    \item If\ \ \(\GSD[\varphi]\vdash n : \ty[q]{T}\ \EPS\yields{\bm{n} \has \DEP},\) then \(\DEP = \DELTA\vert_{\FX{\qsat{\EPS}}}\).
    \item If\ \ \(\GSD[\varphi]\vdash g : \ty[q]{T}\ \EPS\yields{\bm{g}\has \DEP},\) then \(\DEP = \DELTA\vert_{\FX{\qsat{\EPS}}}\).
  \end{enumerate}
\end{lemma}
\begin{proof}
By mutual induction over the respective derivation. Most cases are straightforward. The case
for rule \inflabel{$\leadsto$-let} requires careful reasoning about dependency maps:
\begin{itemize}
  \item Case \rulename{\(\leadsto\)-let}: We need to show \(\mute{\DEP[1],\DEP[2][x \leadsto \DELTA\vert_p]} = \DELTA\vert_{\FX{\qsat{(\EPS[1]\EFFSEQ\EPS[2]\theta)}}}\) for \(\theta = [p/x]\).
        \begin{enumerate}
          \item By IH: \(\DEP[1]=\DELTA\vert_{\FX{\qsat{\EPS[1]}}}\).
          \item By IH: \(\DEP[2]=(\DELTA,\mute{(\FX{\qsat{\EPS[1]}},x)\mapsto x})\vert_{\FX{\qsat{\EPS[2]}}}\).
          \item By scoping: \(x\notin\DOM(\DELTA)\), and \(x\notin\CODOM(\DELTA)\).
          \item Case distinction:
                \begin{enumerate}
                  \item \(x\notin\EPS[2]\): Thus,
                  \[\begin{array}{rcll}
                    \mute{\DEP[1],\DEP[2][x \leadsto \DELTA\vert_p]} & = & \DEP[1],\DEP[2] & \text{by}\ x\notin\EPS[2]  \\
                                     & = & \DELTA\vert_{\FX{\qsat{\EPS[1]}}},(\DELTA,\mute{(\FX{\qsat{\EPS[1]}},x)\mapsto x})\vert_{\FX{\qsat{\EPS[2]}}} & \text{by (1), (2)}\\
                                     & = &\DELTA\vert_{\FX{\qsat{\EPS[1]}}},\DELTA\vert_{\FX{\qsat{\EPS[2]}}\setminus\FX{\qsat{\EPS[1]}}} & \text{by}\ x\notin\EPS[2]\ \text{and def. of}\  \_,\_ \\
                                     & = &\DELTA\vert_{\FX{\qsat{\EPS[1]}},\FX{\qsat{\EPS[2]}}} & \text{by definition of}\ \_\vert_{\_} \\
                                     & = & \DELTA\vert_{\FX{\qsat{\EPS[1]}}\EFFSEQ\FX{\qsat{\EPS[2]}}} & \text{by definition of}\ \_\EFFSEQ\_\\
                                     & = & \DELTA\vert_{\FX{\qsat{(\EPS[1]\EFFSEQ\EPS[2])}}} & \text{by properties of}\ \qsat{\_}\\
                                     & = & \DELTA\vert_{\FX{(\qsat{\EPS[1]\EFFSEQ\EPS[2]\theta)}}}\quad \checkmark & \text{by}\ x\notin\EPS[2]\\
                  \end{array}\]
                  \item \(x\in\EPS[2]:\)
                         \begin{enumerate}
                          \item We have that
                          \[\begin{array}{rcll}
                            \DEP[2] & = & (\DELTA,\mute{(\FX{\qsat{\EPS[1]}},x)\mapsto x})\vert_{\FX{\qsat{\EPS[2]}}} & \text{by (2)} \\
                                    & = & (\DELTA\setminus \FX{\qsat{\EPS[1]}}, \mute{(\FX{\qsat{\EPS[1]}},x)\mapsto x}) \vert_{\FX{\qsat{\EPS[2]}}} &  \text{by def. of}\  \_,\_\\\
                                    & = & (\DELTA\setminus \FX{\qsat{\EPS[1]}})\vert_{\FX{\qsat{\EPS[2]}}}, \mute{(\FX{\qsat{\EPS[1]}}\mapsto x)\vert_{\FX{\qsat{\EPS[2]}}}, (x \mapsto x)\vert_{\FX{\qsat{\EPS[2]}}}}  & \text{by prop. of}\ \_\vert_{\_}\\
                                    & = & \DELTA\vert_{\FX{\qsat{\EPS[2]}}\setminus \FX{\qsat{\EPS[1]}}}, \mute{(\FX{\qsat{\EPS[1]}}\mapsto x)\vert_{\FX{\qsat{\EPS[2]}}}, (x \mapsto x)\vert_{\FX{\qsat{\EPS[2]}}}}  & \text{by prop. of}\ \_\vert_{\_}\\
                                    & = & \DELTA\vert_{\FX{\qsat{\EPS[2]}}\setminus \FX{\qsat{\EPS[1]}}}, \mute{(\FX{\qsat{\EPS[1]}}\overlap\FX{\qsat{\EPS[2]}}\mapsto x), x \mapsto x}   & \text{by }x\in\EPS[2],\ \text{and prop. of}\ \_\vert_{\_}\\
                                    & = & \DELTA\vert_{\FX{\qsat{\EPS[2]}}\setminus \FX{\qsat{\EPS[1]}},x}, \mute{(\FX{\qsat{\EPS[1]}}\overlap\FX{\qsat{\EPS[2]}}\mapsto x), x \mapsto x}   & \text{by }x\notin\CODOM{\DELTA}\\
                          \end{array}\]
                          \item From \(x\in\EPS[2]\) and properties of reachability saturation, we have that \(\qsat{p}\subq\FX{\qsat{\EPS[2]}}\).
                          \item From that, and  \(x\in\EPS[2]\) , and properties of saturation, we have either \(\qsat{p}\subq \FX{\qsat{\EPS[1]}}\overlap\FX{\qsat{\EPS[2]}}\) or
                                \(\qsat{p} \overlap \FX{\qsat{\EPS[1]}}\overlap\FX{\qsat{\EPS[2]}} = \qbot\). Case distinction:
                          \item Case \(\qsat{p}\subq \FX{\qsat{\EPS[1]}}\overlap\FX{\qsat{\EPS[2]}}\): Thus, \(\EPSS[1]\overlap\EPSS[2] = \qsat{p},q\) where \(\qsat{p}\overlap q=\qbot\)
                                  \begin{enumerate}
                                  \item That with (4).(b).(i) yields \(\DEP[2] = \DELTA\vert_{\EPSS[2]\setminus \EPSS[1],x}, \mute{(\qsat{p},q\mapsto x), x \mapsto x}\)
                                  \item Thus,
                                  \[\begin{array}{rcll}
                                    \DEP[2]\mute{[x \leadsto \DELTA\vert_p]} & = & (\DELTA\vert_{\EPSS[2]\setminus \EPSS[1],x}, \mute{(\qsat{p},q\mapsto x), x \mapsto x})\mute{[x \leadsto \DELTA\vert_\qsat{p}]} & \text{by (A)} \\
                                                                             & = & \DELTA\vert_{\EPSS[2]\setminus \EPSS[1],x}\mute{[x \leadsto \DELTA\vert_\qsat{p}]}, \mute{(\qsat{p},q\mapsto x)\mute{[x \leadsto \DELTA\vert_\qsat{p}]}, x \mapsto x\mute{[x \leadsto \DELTA\vert_\qsat{p}]}} & \text{by prop. of }\leadsto \\
                                                                             & = & \DELTA\vert_{\EPSS[2]\setminus \EPSS[1],x}\mute{[x \leadsto \DELTA\vert_\qsat{p}]}, \mute{(\qsat{p},q\mapsto x)\mute{[x \leadsto \DELTA\vert_\qsat{p}]}} & \text{by }x\notin\DOM(\mute{\DELTA\vert_\qsat{p}}) \\
                                                                             & = & \DELTA\vert_{\EPSS[2]\setminus \EPSS[1],x}\mute{[x \leadsto \DELTA\vert_\qsat{p}]}, \mute{(\qsat{p}\mapsto x)\mute{[x \leadsto \DELTA\vert_\qsat{p}]}} & \text{by (iv)} \\
                                                                             & = & \DELTA\vert_{\EPSS[2]\setminus \EPSS[1],x}\mute{[x \leadsto \DELTA\vert_\qsat{p}]}, \mute{\DELTA\vert_\qsat{p}} & \text{by def. of }\leadsto \\
                                                                             & = & \DELTA\vert_{\EPSS[2]\setminus \EPSS[1],x}, \mute{\DELTA\vert_\qsat{p}} & \text{by }x\notin\CODOM(\DELTA) \\
                                  \end{array}\]
                                  \item From that, and (1), we conclude
                                   \[\begin{array}{rcll}
                                    \mute{\DEP[1],\DEP[2][x \leadsto \DELTA\vert_\qsat{p}]}  & = & \DELTA\vert_{\EPSS[1]},\DELTA\vert_{\EPSS[2]\setminus \EPSS[1],x}, \mute{\DELTA\vert_\qsat{p}} & \text{}\\
                                                                                      & = & \DELTA\vert_{\EPSS[1],\EPSS[2]\setminus x,\mute{\qsat{p}}} & \text{by definition of}\ \_\vert_{\_}\text{ and prop. of sets} \\
                                                                                      & = & \DELTA\vert_{\EPSS[1]\EFFSEQ\EPSS[2][\qsat{p}/x]}  & \text{by def. of substitution and effect composition} \\
                                                                                      & = & \DELTA\vert_{\FX{(\qsat{\EPS[1]\EFFSEQ\EPS[2]\theta)}}}\quad \checkmark & \text{by prop. of saturation} \\
                                     \end{array}\]
                                \end{enumerate}
                          \item Case \(\qsat{p} \overlap \EPSS[1]\overlap\EPSS[2] = \qbot\): Thus with (ii), we have \(\qsat{p}\subq\EPSS[2]\setminus\EPSS[1]\).
                                \begin{enumerate}
                                  \item Thus,
                                  \[\begin{array}{rcll}
                                    \DEP[2]\mute{[x \leadsto \DELTA\vert_\qsat{p}]} & = & (\DELTA\vert_{\EPSS[2]\setminus \EPSS[1],x}, \mute{(\EPSS[1]\overlap\EPSS[2]\mapsto x), x \mapsto x})\mute{[x \leadsto \DELTA\vert_\qsat{p}]} & \text{by (i)} \\
                                                                             & = & \DELTA\vert_{\EPSS[2]\setminus \EPSS[1],x}\mute{[x \leadsto \DELTA\vert_\qsat{p}]} & \text{by (v) and } x\notin\DOM(\mute{\DELTA\vert_\qsat{p}})\\
                                                                             & = & \DELTA\vert_{\EPSS[2]\setminus \EPSS[1],x} & \text{by } x \notin\CODOM(\DELTA)
                                  \end{array}\]
                                  \item From that, and (1), we conclude
                                   \[\begin{array}{rcll}
                                    \mute{\DEP[1],\DEP[2][x \leadsto \DELTA\vert_\qsat{p}]}  & = & \DELTA\vert_{\EPSS[1]},\DELTA\vert_{\EPSS[2]\setminus \EPSS[1],x} & \\
                                                                                      & = & \DELTA\vert_{\EPSS[1],\EPSS[2]\setminus x,\mute{\qsat{p}}} & \text{by (v)}\\
                                                                                      & = & \DELTA\vert_{\EPSS[1]\EFFSEQ\EPSS[2][\qsat{p}/x]}  & \text{by def. of substitution and effect composition} \\
                                                                                      & = & \DELTA\vert_{\FX{(\qsat{\EPS[1]\EFFSEQ\EPS[2]\theta)}}}\quad \checkmark & \text{by prop. of saturation} \\
                                     \end{array}\]
                                \end{enumerate}
                         \end{enumerate}
                \end{enumerate}
        \end{enumerate}
\end{itemize}
\end{proof}
The above proof of \Cref{lem:deps-invariant}, justifies that we could alternatively pick the dependency \(\DELTA\vert_{\FX{\qsat{(\EPS[1]\EFFSEQ\EPS[2]\theta)}}}\)
as the synthesis result in the conclusion of rule \rulename{\(\leadsto\)-let}, and certifies that sequential dependency map composition and
rewiring are consistent with effect composition and substitution, provided that the effects are saturated in the context.

\begin{lemma}[Soundness of Synthesis]\label{lem:synth-soundness}
  Synthesis produces well-typed annotated programs:
  \begin{enumerate}
    \item If\ \ \(\GSD[\flt]\ts n : \ty[q]{T}\ \EPS\yields{\bm{n} \has \DEP},\) then \(\GSD[\varphi]\vdash \bm{n} : \ty[q]{T}\ \EPS\).
    \item If\ \ \(\GSD[\flt]\ts g : \ty[q]{T}\ \EPS\yields{\bm{g} \has \DEP},\) then \(\GSD[\varphi]\vdash \bm{g} : \ty[q]{T}\ \EPS\).
  \end{enumerate}
\end{lemma}
\begin{proof}
By straightforward mutual induction over the respective derivation, making use of \Cref{lem:deps-invariant} where appropriate.
\end{proof}

\begin{lemma}[Synthesis is Total]\label{lem:synth-totality}\hfill
  \begin{enumerate}
    \item If\ \ \(\GS[\varphi]\tsM n : \ty[q]{T}\ \EPS\) then \(\forall \DELTA.\ \exists \mute{\bm{n}}.\ \exists\DEP.\) \(\GSD[\varphi]\vdash n : \ty[q]{T}\ \EPS\yields{\bm{n} \has \DEP}\) and \(\mute{\bm{n}}\) erases back to \(n\).
    \item If\ \ \(\GS[\varphi]\tsM g : \ty[q]{T}\ \EPS\) then \(\forall \DELTA.\ \exists \mute{\bm{g}}.\ \exists\DEP.\) \(\GSD[\varphi]\vdash g : \ty[q]{T}\ \EPS\yields{\bm{g} \has \DEP}\) and \(\mute{\bm{g}}\) erases back to \(g\).
  \end{enumerate}
\end{lemma}
\begin{proof}
By straightforward mutual induction over the respective derivations, exploiting that typing rules of each system are in one-to-one correspondence.
\end{proof}
\begin{corollary}[End-to-end Type Preservation]
For any well-formed context and compatible \(\DELTA\), there is a type/effect/qualifier-preserving translation from the direct-style \directlang-calculus into
the \irlang graph IR.
\end{corollary}
\begin{proof}
By lemmas \ref{lem:type_preservation:translation}, \ref{lem:synth-totality}, and \ref{lem:synth-soundness}.
\end{proof}

\subsubsection{Substitution and Rewiring}

\begin{lemma}[Top-Level Rewiring and Substitution]\label{lem:graphir:subst_term}\hfill\\
If \(\ \WF{\cdsx{\G,x:\ty[p\overlap r]{S}}{\Sigma}{\DELTA}.z}\),\ \ \(\cdsx[p]{\varnothing}{\Sigma}{\NODEP}\ts \ell : \ty[p]{S}\ \PURE\), and  \(\theta = [p/x]\)
where \(p\subq\dom(\Sigma)\) and \( p \qglb \flt \subq p\overlap r\), and \(\DOM(\DEP)\subq\DOM(\Sigma)\), and \(\CODOM(\DEP)\subq \{z\}\) then
  \infrule[1]{%
        \cdsx[\flt]{\G,x:\ty[p\overlap r]{S}}{\Sigma}{\DELTA} \ts n : \ty[q]{T}\ \EPS
    }{
        \cdsx[\flt\theta]{\G\theta}{\Sigma}{\mute{\DELTA[x\leadsto \DEP\uparrow^z]\theta}} \ts n\mute{[x\leadsto \DEP]}\theta[\ell/x]_{\mathsf{t}} : (\ty[q]{T}\ \FX{\EPS})\theta
    }
  \infrule[2]{%
        \cdsx[\flt]{\G,x:\ty[p\overlap r]{S}}{\Sigma}{\DELTA} \ts g : \ty[q]{T}\ \EPS
    }{
        \cdsx[\flt\theta]{\G\theta}{\Sigma}{\mute{\DELTA[x\leadsto \DEP\uparrow^z]\theta}} \ts g\mute{[x\leadsto \DEP]}\theta[\ell/x]_{\mathsf{t}} : (\ty[q]{T}\ \FX{\EPS})\theta
    }
\end{lemma}
\begin{proof} %
  The proof proceeds by mutual induction over the respective derivation. Ignoring dependencies, each case uses similar
reasoning steps as the previous substitution lemma proof for the system with store-allocated values (\Cref{lem:subst_term_letv}).
We focus here only on the interesting cases involving dependencies, which are the typing rules for \(\lambda\)-abstraction and let bindings.
\begin{itemize}
  \item Case \rulename{n-abs}: That is, \(n = \lambda y.g\has\DEP'\).
         \begin{enumerate}
          \item We have \(\csx[q,y]{\G\,,\, x: \ty[p\overlap r]{S}\,,\, y : \ty[p']{T}}{\Sigma}\has\mute{\pointsto{y}} \ts g : \ty[r']{U}\ \EPSPR\).
          \item We have \(\DEP' \sqsubseteq \mute{\EPSSPR\mapsto y}\).
          \item By IH:  \(\csx[q\theta,y]{\G\theta,\, y : \ty[p'\theta]{T\theta}}{\Sigma}\has\mute{(\pointsto{y}[x\leadsto\DEP\uparrow^z]\theta)} \ts g\mute{[x\leadsto\DEP]}\theta[\ell/x]_{\mathsf{t}} : (\ty[r']{U}\ \EPSPR)\theta\).
          \item Since \(x \neq y\), it holds that \(\mute{(\pointsto{y}[x\leadsto\DEP\uparrow^z]\theta) = \pointsto{y}}\).
          \item By (2), no entry in \(\DEP'\) points to \(x\), because it is a submap of \(\mute{\EPSSPR\mapsto y}\) and \(y \neq x\).
          \item Thus \(\DEP'\mute{[x\leadsto \DEP]}\theta = \DEP'\theta\), and by (2) and monotonicity of substitution, \(\DEP'\theta \sqsubseteq \mute{\EPSPR\theta\mapsto y}\).
          \item Hence \(n\mute{[x\leadsto \DEP]}\theta[\ell/x]_{\mathsf{t}} = \lambda y. (g\mute{[x\leadsto \DEP]}\theta[\ell/x]_{\mathsf{t}})\has\delta'\theta\).
          \item By (3),(6),(7), and \rulename{n-abs} the proof goal follows.
         \end{enumerate}
  \item Case \rulename{g-let}: That is, \(g = \tlet~{y = b\has\DEP'}~\tin~g'\).
        \begin{enumerate}
          \item We have  \(\cdsx[\flt]{\G\, ,\, x: \ty[p\overlap r]{S}}{\Sigma}{\DELTA}\ts b : \ty[q]{T} \ \EPS[1]\).
          \item We have \(\cdsx[\flt,y]{\G\, ,\, x: \ty[p\overlap r]{S}\,,\, y : \ty[q]{T}}{\Sigma}{\DELTA,\mute{(\EPSS[1],y)\mapsto y}} \ts g' : \ty[r]{U}\ \EPS[2]\).
          \item We have \(\mute{\DEP'\sqsubseteq\DELTA\vert_{\EPSS[1]}}\).
          \item By IH: \(\cdsx[\flt\theta]{\G\theta}{\Sigma}{\DELTA\mute{[x\leadsto\DEP\uparrow^z]\theta}}\ts b\mute{[x\leadsto\DEP]}\theta[\ell/x]_{\mathsf{t}} : (\ty[q]{T} \ \EPS[1])\theta\).
          \item By IH: \(\cdsx[\flt\theta,y]{\G\theta\,, y : \ty[q\theta]{T\theta}}{\Sigma}{(\DELTA,\mute{(\EPSS[1],y)\mapsto y})\mute{[x\leadsto\DEP\uparrow^z]\theta}} \ts g'\mute{[x\leadsto\DEP]}\theta[\ell/x]_{\mathsf{t}} : (\ty[r]{U}\ \EPS[2])\theta\).
          \item From \(x \neq y\) and the properties of rewiring and qualifier substitution on dependencies, it follows that
          \[(\DELTA,\mute{(\EPSS[1],y)\mapsto y})\mute{[x\leadsto\DEP\uparrow^z]\theta} = \DELTA\mute{[x\leadsto\DEP\uparrow^z]\theta}, \mute{(\EPSS[1]\theta,y)\mapsto y}. \]
          \item From (3) and monotonicity of substitution, we have \(\DEP'\theta \sqsubseteq \DELTA\theta\vert_{\EPSS[1]\theta}\).
          \item By (4), (5), (6), (7), and \rulename{g-let} the proof goal follows.
        \end{enumerate}
\end{itemize}
\end{proof}

\subsubsection{Context Typing and Effect Introductions}

We have generative effect introductions that modify the runtime context, and thus need lemmas for
plugging/decomposition and growing the stack of dependencies by fresh effect dependencies from
generative effects (i.e., reference allocations in this system).

\begin{definition}[Context Typings]\label{def:graphir:context_typing} We define the typings of graph and binding contexts (\Cref{fig:graphir:cbv}) relative to an ambient block start variable \(z\), as follows:
  \begin{enumerate}
    \item \(\GSD[\flt]\ts z.\CX{G}{\cdot}: \ty[p]{S}\ \EPS[1]\Rightarrow \ty[q]{T}\ \EPS[2]\) iff \(\cdsx[\flt]{\G, x : \ty[p]{S}}{\Sigma}{\mute{\DELTA,x\mapsto z}}\ts z.\CX{G}{x}: \ty[q]{T}\ \EPS[2]\) and \(\EPSS[1] = \DOM(\DEP)\) for
    the dependency \(\DEP\) at \(G\)'s hole.
    \item \(\GSD[\flt]\ts z.\CX{B}{\cdot}: \ty[p]{S}\ \EPS[1]\Rightarrow \ty[q]{T}\ \EPS[2]\) iff \(\cdsx[\flt]{\G, x : \ty[p]{S}}{\Sigma}{\mute{\DELTA,x\mapsto z}}\ts z.\CX{B}{x\has\DEP}: \ty[q]{T}\ \EPS\)
    for some \(\DEP\sqsubseteq\DELTA\vert_{\EPSS[1]}\).
  \end{enumerate}
\end{definition}

\begin{lemma}[Contextual Effect Propagation]\label{lem:graphir:context_propagation}\hfill
   \infrule{
     \cdsx[\flt]{\varnothing}{\Sigma}{\mute{\pointsto{z}}}\ts z.\CX{G}{\cdot}: \ty[p]{S}\ \EPS[1]\Rightarrow \ty[q]{T}\ \EPS[2] \quad \ell\notin\DOM(\Sigma) %
   }{
     \cdsx[\flt]{\varnothing}{\Sigma, \ell : \ty[r]{U}}{\mute{\pointsto{z}}}\ts z.\CX{G\langle \iota: z.\ell\rangle}{\cdot}: \ty[p]{S}\ \EPS[1],\FX{\ell}\Rightarrow \ty[q]{T}\ \EPS[2],\FX{\ell}
   }
\end{lemma}
\begin{proof}
Let \(\cdsx[\flt]{\varnothing}{\Sigma}{\mute{\pointsto{z}}}\ts z.\CX{G}{\cdot}: \ty[p]{S}\ \EPS[1]\Rightarrow \ty[q]{T}\ \EPS[2]\),
and \(\ell\notin\DOM(\Sigma)\). We proceed by induction over \(G\):
\begin{itemize}
  \item Case \(G = \square\has\DEP\): Hence \(G\langle \iota: z.\ell\rangle = \square\has\mute{\DEP,\ell\mapsto z}\), and \(\GS\ts\ty[p]{S}\ \EPS[1] <: \ty[q]{T}\ \EPS[2]\),
  and \(\EPS[1] = \DOM(\DEP)\). By the properties of subtyping, it holds that \(\csx{\varnothing}{\Sigma,\ell : \ty[r]{U}}\ts\ty[p]{S}\ \EPS[1],\FX{\ell} <: \ty[q]{T}\ \EPS[2],\FX{\ell}\).
  Using \rulename{g-ret} and \rulename{b-sub} proves the goal.
  \item Case \(G = (\tlet~{x = G'}~\tin~g)\has\DEP\): Hence \(G\langle \iota: z.\ell\rangle = (\tlet~{x = G'\langle \iota: z.\ell\rangle}~\tin~g)\has\mute{\DEP,\ell\mapsto z}\), and
  \begin{enumerate}
    \item By Def.~\ref{def:graphir:context_typing} and typing inversion:
    \begin{enumerate}
      \item \(\cdsx[\flt]{y : \ty[p]{S}}{\Sigma}{\mute{\pointsto{z}}}\ts z.\CX{G'}{y}: \ty[p']{S'}\ \EPSPR[2]\) for some \(\ty[p']{S'}\ \EPSPR[2]\).
      \item \(\cdsx[\flt]{y : \ty[p]{S}, x : \ty[p']{S'}}{\Sigma}{\mute{(\pointsto{z},\EPSSPR[2]\mapsto x)}}\ts g : \ty[q']{T}\ \EPS[3]\).
      \item \(q = q'[p'/x]\), \(\EPS[2] = \EPSPR[2]\EFFSEQ\EPS[3]\FX{[p'/x]}\).
      \item \(\EPSS[1]=\DOM(\DEP')\) for the dependency \(\DEP'\) at the hole of \(G'\).
    \end{enumerate}
    \item By IH: \(\cdsx[\flt]{\varnothing}{\Sigma, \ell : \ty[r]{U}}{\mute{\pointsto{z}}}\ts z.\CX{G'\langle \iota: z.\ell\rangle}{\cdot}: \ty[p]{S}\ \EPS[1],\FX{\ell}\Rightarrow \ty[p']{S'}\ \EPSPR[2],\FX{\ell}\).
    \item By weakening on (1b): \(\cdsx[\flt]{y : \ty[p]{S}, x : \ty[p']{S'}}{\Sigma, \ell:\ty[r]{U}}{\mute{(\pointsto{z},(\EPSSPR[2],\FX{\ell},\FX{\qsat{r}})\mapsto x)}}\ts g : \ty[q']{T}\ \EPS[3]\).
    \item By \rulename{g-let}: \(\cdsx[\flt]{y : \ty[p]{S}}{\Sigma, \ell : \ty[r]{U}}{\mute{\pointsto{z}}}\ts \tlet~{x = \CX{G'\langle \iota: z.\ell\rangle}{y}}~\tin~g: \ty[ {q'\theta} ]{T}\ \EPSPR[2],\FX{\ell}\EFFSEQ\EPS[3]\theta\)
    for \(\theta = [p'/x]\), and the goal follows from that.
  \end{enumerate}
\end{itemize}
\end{proof}

\begin{lemma}[Decomposition]\label{lem:graphir:decomp}\hfill
\begin{enumerate}
    \item If\(\ \ \cdsx[\flt]{\varnothing}{\Sigma}{\mute{\pointsto{z}}}\ts z.\CX{G}{g}: \ty[q]{T}\ \EPS\), then \(\cdsx[\flt]{\varnothing}{\Sigma}{\mute{\pointsto{z}}}\ts z.\CX{G}{\cdot}: \ty[p]{S}\ \EPSPR\Rightarrow \ty[q]{T}\ \EPS\)
    for some \(\ty[p]{S}\ \EPSPR\),
    and \(\cdsx[\flt]{\varnothing}{\Sigma}{\mute{\pointsto{z}}}\ts z.g\has\DEP: \ty[p]{S}\ \EPSPR\), where \(\DEP\) is the dependency at the hole of \(G\).
    \item If\(\ \ \cdsx[\flt]{\varnothing}{\Sigma}{\mute{\pointsto{z}}}\ts z.\CX{B}{b\has\DEP}: \ty[q]{T}\ \EPS\), then \(\cdsx[\flt]{\varnothing}{\Sigma}{\mute{\pointsto{z}}}\ts z.\CX{B}{\cdot}: \ty[p]{S}\ \EPSPR\Rightarrow \ty[q]{T}\ \EPS\)
    for some \(\ty[p]{S}\ \EPSPR\), where
    \(\cdsx[\flt]{\varnothing}{\Sigma}{\mute{\pointsto{z}}}\ts b: \ty[p]{S}\ \EPSPR\) and \(\DEP \sqsubseteq \mute{\pointsto{z}\vert_{\EPSSPR}}\).
\end{enumerate}
\end{lemma}
\begin{proof}
  Both cases are proved by induction over the respective context.
\end{proof}
\begin{lemma}[Plugging]\label{lem:graphir:plug}\hfill
\begin{enumerate}
  \item If\(\ \ \cdsx[\flt]{\varnothing}{\Sigma}{\mute{\pointsto{z}}}\ts z.\CX{G}{\cdot}: \ty[p]{S}\ \EPS[1]\Rightarrow \ty[q]{T}\ \EPS[2]\)
  and \(\cdsx[\flt]{\varnothing}{\Sigma}{\mute{\pointsto{z}}}\ts z.g\has\DEP: \ty[p]{S}\ \EPS[1]\) where \(\DEP\) is the dependency at the hole of \(G\),
  then \(\cdsx[\flt]{\varnothing}{\Sigma}{\mute{\pointsto{z}}}\ts z.\CX{G}{g}: \ty[q]{T}\ \EPS[2]\).
  \item If\(\ \ \cdsx[\flt]{\varnothing}{\Sigma}{\mute{\pointsto{z}}}\ts z.\CX{B}{\cdot}: \ty[p]{S}\ \EPS[1]\Rightarrow \ty[q]{T}\ \EPS[2]\)
  and \(\cdsx[\flt]{\varnothing}{\Sigma}{\mute{\pointsto{z}}}\ts b: \ty[p]{S}\ \EPS[1]\), then for all \(\mute{\DEP \sqsubseteq \EPSS[1]\mapsto z}\)
  it holds that \(\cdsx[\flt]{\varnothing}{\Sigma}{\mute{\pointsto{z}}}\ts z.\CX{B}{b\has\DEP}: \ty[q]{T}\ \EPS[2]\).
\end{enumerate}
\end{lemma}
\begin{proof}
  Both cases are proved by induction over the respective context.
\end{proof}
\begin{lemma}[Qualifier-Growing Replacement]\label{lem:graphir:replacement}\hfill
\begin{enumerate}
  \item If\(\ \ \cdsx[\flt]{\varnothing}{\Sigma}{\mute{\pointsto{z}}}\ts z.\CX{G}{g}: \ty[p]{S}\ \EPS[1]\Rightarrow \ty[q]{T}\ \EPS[2]\)
    and \(\cdsx[\flt]{\varnothing}{\Sigma'}{\mute{\pointsto{z}}}\ts g': \ty[p,r]{S}\ \EPS[1]\FX{,r}\) where \(\Sigma'\supseteq\Sigma\) and \(r\subq\DOM(\Sigma'\setminus\Sigma)\), then \(\cdsx[\flt]{\varnothing}{\Sigma'}{\mute{\pointsto{z}}}\ts z.\CX{G'}{g'}: \ty[q,r]{T}\ \EPS[2]\FX{,r}\) where
    \(G'\) is the result of contextual effect propagation for the fresh introduction form \(r\).
  \item  If\(\ \ \cdsx[\flt]{\varnothing}{\Sigma}{\mute{\pointsto{z}}}\ts z.\CX{B}{b\has\DEP}: \ty[p]{S}\ \EPS[1]\Rightarrow \ty[q]{T}\ \EPS[2]\)
    and \(\cdsx[\flt]{\varnothing}{\Sigma'}{\mute{\pointsto{z}}}\ts b': \ty[p,r]{S}\ \EPS[1]\FX{,r}\) where \(\Sigma'\supseteq\Sigma\) and \(r\subq\DOM(\Sigma'\setminus\Sigma)\), then \(\cdsx[\flt]{\varnothing}{\Sigma'}{\mute{\pointsto{z}}}\ts z.\CX{B'}{b'\has\DEP}: \ty[q,r]{T}\ \EPS[2]\FX{,r}\)
    where \(B'\) is the result of contextual effect propagation for the fresh introduction form \(r\).
\end{enumerate}
\end{lemma}
\begin{proof}
  By decomposition (\Cref{lem:graphir:decomp}), contextual effect propagation (\Cref{lem:graphir:context_propagation}),
  and plugging (\Cref{lem:graphir:plug}).
\end{proof}
\subsubsection{Soundness}
\begin{theorem}[Progress]\label{thm:graphir:progress}
  If \(\ \csx[\DOM(\Sigma)]{\varnothing}{\Sigma} \ts z.g\has\DEP : \ty[q]{T}\ \EPS\), then either \(g\) is a location \(\ell\in\DOM(\Sigma)\), or
  for any store \(\sigma\) where \(\csx[\DOM(\Sigma)]{\varnothing}{\Sigma} \ts \sigma\), there exists
  a graph term \(g'\), store \(\sigma'\), and dependency \(\mute{\DEP'}\) such that \(\sigma \mid z.g\has\DEP  \redg \sigma'\mid z.g'\has\mute{\DEP'}\).
\end{theorem}
\begin{proof}
  By induction over the typing derivation.
\end{proof}
\begin{theorem}[Preservation]\label{thm:graphir:preservation}\hfill\\
\begin{adjustbox}{scale=.9}
\begin{minipage}{\linewidth}
  \infrule{\csx[\DOM(\Sigma)]{\varnothing}{\Sigma}  \ts z.g\has\DEP : \ty[q]{T}\ \EPS\quad \WF{\cdsx{\varnothing}{\Sigma}{\mute{\pointsto{z}}.z}}\quad \cdsx[\DOM(\Sigma)]{\varnothing}{\Sigma}{\mute{\pointsto{z}}} \ts \sigma\quad\sigma\mid z.g\has\DEP \redg \sigma' \mid z.g'\has\mute{\DEP'}
  }{
    \exists \Sigma' \supseteq \Sigma.\;\exists p \subq\DOM(\Sigma'\setminus\Sigma).\quad\csx[\DOM(\Sigma')]{\varnothing}{\Sigma'} \ts z.g'\has\mute{\DEP'} : \ty[q,p]{T}\ \FX{\EPS,p}\qquad\qquad\\[1ex]
    \phantom{\exists \Sigma' \supseteq \Sigma.\;\exists p \subq\DOM(\Sigma'\setminus\Sigma).\ \ \,\,}\WF{\cdsx{\varnothing}{\Sigma'}{\mute{\pointsto{z}}.z}}\qquad  \cdsx[\DOM(\Sigma')]{\varnothing}{\Sigma'}{\mute{\pointsto{z}}} \ts \sigma'
  }
\end{minipage}
\end{adjustbox}
\end{theorem}
\begin{proof}
  By inspecting the rule applied for the step \(\sigma\mid z.g\has\DEP \redg \sigma' \mid z.g'\has\mute{\DEP'}\) and
  the qualifier-growing replacement \Cref{lem:graphir:replacement}, it is sufficient to prove that
  each rule is type/effect/qualifier preserving up to fresh store introductions in a minimal context:
  \begin{itemize}
    \item Case \rulename{\(\beta\)}: In which case we have a well-typed application \(\ell_1~\ell_2\has\DEP[1]\) in the hole. We proceed by induction
    over its typing derivation, which ends either in \rulename{n-app} or \rulename{b-sub}. No store introduction occurs, hence the context type is preserved.
    \begin{itemize}
      \item Case \rulename{n-app}:
            \begin{enumerate}
              \item We have \(\CODOM(\DEP[1])\subq\{z\}\).
              \item We have \(\ell_1 : \ty[q]{\left(x: \ty[\qsat{p}\overlap \qsat{q}]{T} \to^{\EPS} \ty[r]{U}\right)}\in\csx[\DOM(\Sigma)]{\varnothing}{\Sigma}\).
              \item We have \(\ell_2 : \ty[p]{T}\in\csx[\DOM(\Sigma)]{\varnothing}{\Sigma}\).
              \item We have \(x \notin\FV(U)\),  \(\EPS\subq q,x\),  \(r\subq\flt,x\), and \(\theta = [p/x]\).
              \item By (1) and environment relation, we have \(\sigma(\ell_1) = \lambda x.g\has\DEP[2]\).
              \item By inversion: \(\csx[q',x]{x: \ty[p']{T'}}{\Sigma}\has\mute{\pointsto{x}} \ts g : \ty[r']{U'}\ \EPS'\),
                    \(\mute{\DEP[2]\sqsubseteq\EPSSPR\mapsto x}\), \(\csx{\varnothing}{\Sigma}\ts \ty[\qsat{p}\overlap \qsat{q}]{T}\ \PURE <: \ty[p']{T'}\ \PURE\), and \(\csx{x: \ty[\qsat{p}\overlap \qsat{q}]{T}}{\Sigma}\ts\ty[r']{U'}\ \EPS' <: \ty[r]{U}\ \EPS\), and \(q'\subq q\).
              \item By narrowing, weakening, subsumption: \(\csx[\DOM(\Sigma),x]{x: \ty[\qsat{p}\overlap \qsat{q}]{T}}{\Sigma}\has\mute{\pointsto{x}} \ts g : \ty[r]{U}\ \EPS\).
              \item By (1) and the substitution and rewiring \Cref{lem:graphir:subst_term}: \[\csx[\DOM(\Sigma)]{\varnothing}{\Sigma}\has\mute{\pointsto{z}} \ts g\mute{[x\leadsto\DEP[1]]}[\ell_2/x][\ell_2/x]_{\mathsf{t}} : \ty[{r[p/x]}]{U}\ \EPS{[p/x]}.\]
              \item By (6) and transitivity, we have \(\mute{\DEP[2]\sqsubseteq\EPSS\mapsto x}\), and hence \[\DEP[1]\mute{[x\leadsto\DEP[1]]}[\ell_2/x][\ell_2/x]_{\mathsf{t}}\mute{\,\sqsubseteq\FX{\EPSS{[p/x]}} \mapsto z}\]
              \item By (8), (9), and plugging \Cref{lem:graphir:plug} we can now prove this case.
            \end{enumerate}
      \item Case \rulename{b-sub}: By IH and subsumption.
    \end{itemize}
    \item Case \rulename{let}: Follows from the substitution and rewiring \Cref{lem:graphir:subst_term}.
    \item Case \rulename{intro}: In which case we have a well-typed let binding \(\tlet~{x = \iota\has\DEP}~\tin~g\) in the hole.
    By induction over the derivation, only \rulename{b-sub} or \rulename{g-let} applies. The first is trivial, and we consider the latter case.
    By the contextual effect propagation \Cref{lem:graphir:context_propagation}, we have increased the qualifiers and effects of the right-hand-side context typing with the fresh new store location,
    after which we can proceed as in the \rulename{let} case with the substitution and rewiring \Cref{lem:graphir:subst_term} and conclude.
    \item Case \rulename{deref}: In which case we have a well-typed dereference \(!~\ell\has\DEP\) with \(\CODOM(\DEP)\subq\{z\}\) in the hole.
    By induction over the derivation, only \rulename{b-sub} or \rulename{n-!} applies. The first is again trivial, and the latter case is straightforward,
    since by the environment predicate we have that thew hole is plugged with a store value of the same type.
    \item Case \rulename{assign}: Similar to the previous case.
  \end{itemize}
\end{proof}
\begin{corollary}[Preservation of Separation]\label{thm:graphir:preservation_of_separation}
  Interleaved executions preserve types and disjointness:\\
\begin{adjustbox}{scale=.88}
\begin{minipage}{\linewidth}
  \infrule{%
{\begin{array}{l@{\quad}l@{\quad}l@{\quad}l}
\csx[\DOM(\Sigma)]{\varnothing}{\Sigma} \vdash z.g_1\has\DEP[1]: \ty[q_1]{T_1}\ \EPS[1] & \sigma\phantom{'}\mid z.g_1\has\DEP[1]  \redg \sigma'\phantom{'} \mid z.g_1'\has\mute{\DEP[1]'} & \csx[\DOM(\Sigma)]{\varnothing}{\Sigma} \ts \sigma & \WF{\cdsx{\varnothing}{\Sigma}{\mute{\pointsto{z}}.z}} \\[1ex]
\csx[\DOM(\Sigma)]{\varnothing}{\Sigma} \vdash z.g_2\has\DEP[2]: \ty[q_2]{T_2}\ \EPS[2]  & \sigma'\mid z.g_2\has\DEP[2]  \redg \sigma'' \mid z.g_2'\has\mute{\DEP[2]'}          & q_1 \overlap q_2 \subq \qbot       &
\end{array}}
}{{\begin{array}{ll@{\quad}l@{\quad}l}
\exists p_1\;p_2\;\EPSPR[1]\;\EPSPR[2]\;\Sigma'\;\Sigma''. & \csx[\DOM(\Sigma')\phantom{'}]{\varnothing}{\Sigma'\phantom{'}} \ts z.g_1'\has\mute{\DEP[1]'} : \ty[p_1]{T_1}\ \EPSPR[1] & \Sigma'' \supseteq \Sigma' \supseteq \Sigma \\[1ex]
                                     & \csx[\DOM(\Sigma'')]{\varnothing}{\Sigma''} \ts z.g_2'\has\mute{\DEP[2]'} : \ty[p_2]{T_2}\ \EPSPR[2]                     & p_1 \overlap p_2 \subq \qbot
\end{array}}}
\end{minipage}%
\end{adjustbox}
\end{corollary}

\begin{corollary}[Dependency Safety]\label{coro:graphir:dep-safety}
Evaluation respects the order of effect dependencies for well-typed graph IR terms, \ie, an effectful graph node is executed only if all its dependencies
are resolved in the store.
\end{corollary} %
\section{Extension with Soft Dependencies}\label{sec:extension-with-soft}\label{sec:hardsoft}

\begin{figure}\small
\begin{mdframed}
\judgement{Graph IR}{\BOX{\irlang}}\vspace{-10pt}
\[\begin{array}{l@{\qquad}l@{\qquad}l@{\qquad}l}
     \HHD,\HDEP  &::= & \seq{\V{x}\mapsto \V{x}}                    & \text{Hard Dependencies}\\
     \SSD,\SDEP  &::= & \seq{\V{x}\mapsto \seq{\V{x}}}              & \text{Soft Dependencies}\\
     \DELTA,\DEP &::= & \HLBox[gray!20]{\HDEP; \SDEP  }             & \text{Dependencies}\\
     \EPS        &::= & \HLBox[gray!20]{\FX{\RD  {:} q; \WR {:} q}} & \text{Effects}\\
\end{array}\]\\
\judgement{Effects and Dependencies}{}\\
\[\begin{array}{l@{\ \,}c@{\ \,}l@{\ \ \ }l}
  \FX{(\RD:q_1\,;\, \WR : p_1)}\EFFSEQ\FX{(\RD:q_2; \WR : p_2)}\ \  & := & \FX{(\RD:q_1,q_2\,;\, \WR : p_1,p_2)} & \text{Sequential Composition}\\
  \FX{(\RD:q\,;\, \WR : p)}\subq r &:=& q,p \subq r & \text{Effect/Qualifier Inclusion}\\
  \GS\vdash{\FX{\qsat{(\RD:q\,;\, \WR : p)}}} &:=& \GS\vdash \FX{(\RD:\qsat{q}\,;\, \WR : \qsat{p})} & \text{Effect Saturation}\\
  \PURE &:=& \FX{(\RD:\qbot\,;\, \WR : \qbot)} & \text{Purity}\\
  \FX{\RD(q)} &:=& \FX{(\RD: q \,;\, \WR : \qbot)} & \text{Just a Read}  \\
  \FX{\WR(q)} &:=& \FX{(\RD: \qbot \,;\, \WR : q)} & \text{Just a Write}  \\
  \mute{(\HDEP[1];\SDEP[1]), (\HDEP[2];\SDEP[2])}&:= & \mute{(\HDEP[1],\HDEP[2]);(\SDEP[1],\SDEP[2])} & \text{Update}\\
  (\mute{\HDEP;\SDEP})\vert_{\FX{(\RD:q\,;\, \WR : p)}} &:= & (\mute{\HDEP}\vert_{\FX{q}}; \mute{\HDEP}\vert_{\FX{p}} \sqcup \mute{\SDEP}\vert_{\FX{p}}) & \text{Restriction}\\
  \mute{(\HDEP[1];\SDEP[1])[x \leadsto {\HDEP[2];\SDEP[2]} ]} & := & \mute{(\HDEP[1][x \leadsto {\HDEP[2]} ];\SDEP[1][x \leadsto {\SDEP[2]} ])} & \text{Rewiring}\\
  \mute{\SDEP\oplus_x\FX{q}} & :=&  \mute{\SDEP,\{ y\mapsto \SDEP(y),x\mid y \in \FX{q} \}} & \text{Insertion of}\ x \\
  \mute{(\mute{\HDEP;\SDEP}) \oplus_x \FX{(\RD:q\,;\, \WR : p)}} & := & (\mute{\HDEP, (\FX{p},x)\mapsto x\ ;\ (\SDEP,(\FX{p},x)\mapsto\qbot) \oplus_x \FX{q}  } & \text{Last Use at}\ x 
\end{array}\]
\typicallabel{\(\leadsto\)-let}%
\judgement{Dependency Synthesis}{\BOX{\GSD[\varphi]\ts (n\mid g) : \ty[q]{T}\ \EPS\yields{(\bm{n}\mid\bm{g})\has\DEP}}}
\begin{minipage}[t]{.47\linewidth}\vspace{0pt}
\infrule[\(\leadsto\)-ref]{
  \V{x} : \ty[q]{\Typ{Alloc}}\in\GS[\flt]\\\V{y} : \ty[\qbot]{B}\in \GS[\flt]
    }{
      \GSD[\flt]\ts \tref_\V{x}~\V{y} : \ty[\qbot]{(\TRef~\ty{B})}\ \HLBox[gray!20]{\FX{\RD(\V{x})}}\\\yields{\tref_\V{x}~\V{y}\has\DELTA\vert_{\HLBox[gray!20]{\FX{\qsat{\RD(\V{x})}}}}}
    }
\vgap
    \infrule[\(\leadsto\)-\(!\)]{
      \V{x} : \ty[q]{(\TRef~\ty{B})}\in\GS[\flt]
    }{
      \GSD[\flt]\ts {!~\V{x}} : \ty[\qbot]{B}\ \HLBox[gray!20]{\FX{\RD(\V{x})}}\\\yields{{!~\V{x}}\has\DELTA\vert_{\HLBox[gray!20]{\FX{\qsat{\RD(\V{x})}}}}}
    }
\end{minipage}%
\begin{minipage}[t]{.03\linewidth}
\hspace{1pt}%
\end{minipage}%
\begin{minipage}[t]{.5\linewidth}\vspace{0pt}
\infrule[\(\leadsto\)-\(:=\)]{
      \V{x} : \ty[q]{(\TRef~\ty{B})} \in\GS[\flt]\\
      \V{y} : \ty[\qbot]{B}\in\GS[\flt]
    }{
      \GSD[\flt]\ts \V{x} \coloneqq \V{y} : \ty[\qbot]{\TUnit}\ \HLBox[gray!20]{\FX{\WR(\V{x})}}\\\yields{\V{x} \coloneqq \V{y}\has\DELTA\vert_{\HLBox[gray!20]{\FX{\qsat{\WR(\V{x})}}}}}
    }
\vgap
\infrule[\(\leadsto\)-let]{\ \\
      \GSD[\flt]\ts b : \ty[p]{S} \ \EPS[1]\yields{\bm{b}\has\DEP[1]}\\
      \cdsx[\flt,x]{\G\, ,\, x: \ty[p]{S}}{\Sigma}{\HLBox[gray!20]{\mute{\DELTA\oplus_x\EPSS[1]}}} \ts g : \ty[q]{T}\ \EPS[2] \\
      \yields{\bm{g}\has\DEP[2]}\\
      \theta = [p/x]\quad x \notin\FV{(T)}
    }
    {
      \GSD[\varphi] \vdash \Let{x}{b}{g} : (\ty[q]{T}\ \EPS[1] \EFFSEQ  \EPS[2])\theta\\\yields{\left(\Let{x}{\bm{b}\has\DEP_1}{\bm{g}}\right)\has\DEP[1],\DEP_2[x\leadsto\DELTA\vert_\qsat{p}]}
    }
\end{minipage}
\judgement{Effect Subtyping}{\BOX{\GS\ts\EPS[1] <: \EPS[2]}}\\
\begin{mathpar}
\inferrule*[right=e-sub]{\GS\ts q_1 <: q_2\qquad \GS\ts p_1 <: p_2}{\GS\ts\FX{(\RD:q_1\,;\, \WR : p_1)} <: \FX{(\RD:q_2; \WR : p_2)}}
\end{mathpar}
\caption[The syntax and synthesis rules of the graph IR \irlang with hard and soft dependencies.]{The syntax and typing rules of the graph IR \irlang with hard and soft dependencies. We only show the changes relative to
\Cref{fig:graphir:synthesis}. All other rules remain exactly the same using overloaded operations on effects and qualifiers that act component-wise on hard and soft dependencies.}\label{fig:hardsoft:syntax}\label{fig:hardsoft:checking}
\end{mdframed}
\end{figure}

We extend the graph IR \irlang from the previous section with soft dependencies.
During code generation, a node that is only soft-depended by other nodes is
considered dead, and therefore is not scheduled (cf.~\Cref{sec:codemotion}).
If node $A$ hard-depends on node $B$, then $B$ must be executed (or scheduled)
before $A$. This is the default notion of dependency for the base \irlang{} system (\Cref{sec:monad-norm-with}), and
entails that no effect operation can be skipped.
This is evidently too rigid, and as motivated in the main paper, soft dependencies gives us more slack
to outright omit effects that are not observable, \eg, write-after-write (WAR) on a mutable reference cell.
If  $A$ \emph{soft-depends} on $B$, then $B$ should never be scheduled after
$A$, but $B$ might not be scheduled even if $A$ is scheduled. %
Being able to tell that some effectful part of a higher-order program can be omitted is immensely
useful.

The entire formal system and reasoning principles of \irlang carry over into a system with hard and
soft dependencies as presented in this section. The difference is the change in the effect and
dependency structure, \ie, effects are split into reads and writes, which induce hard
dependencies (the previous section's notion) and soft dependencies, respectively. That is to say, we
can regard these new structures as a product composition of the previous with new structures.

In future work, we would like to develop a generic theory of graph IRs that is parametric in such
effect and dependency structures. \citet{DBLP:journals/pacmpl/BaoWBJHR21}'s direct style system
already proposes one half of the solution by adopting \citet{DBLP:journals/toplas/Gordon21}'s effect
quantales. We anticipate that a general graph IR would require a ``dependency quantuale'' that
mirrors a given effect quantale.

In the following, we focus on the key differences to the previous section.

\subsection{Effects and Dependencies for Reads and Writes}\label{sec:hardsoft:effects}

The nature of effects changes from a simple ``effectful use of/on a variable'' to a more refined
distinction, classifying the effect on the variables as either a read or a write effect. Due to
aliasing/reachability, there is usually more than one variable involved, and compound expressions
accumulate their effects. Thus, we change the effect domain to labelled pairs \(\FX{\RD:q;\WR:p}\)
of qualifiers, grouping variables/locations by read and write (\Cref{fig:hardsoft:syntax}). We also
lift the preexisting operations and relations involving effects to such pairs in a straightforward
manner, with the intent that the \irlang{} typing and synthesis rules can be copied over almost
one-to-one with overloaded notations for the new effect structure. The only tweak needed is that
typing rules introducing effects should classify them as read or write effects, and the
last-use update at let bindings is a bit more involved.

\subsection{Hard-and-Soft Dependency Calculation}\label{sec:hardsoft:calculation}

We use the hard dependency in \(\DELTA\) to track the last write of any variable/location in context, whereas
its soft dependency tracks all the reads on a variable since it was last written.

Projections need to merge the hard dependencies of a write effect into its soft dependencies, and project the
hard dependencies for reads, \ie, \( (\mute{\HDEP;\SDEP})\vert_{\FX{(\RD:q\,;\, \WR : p)}} := (\mute{\HDEP}\vert_{\FX{q}}; \mute{\HDEP}\vert_{\FX{p}} \sqcup \mute{\SDEP}\vert_{\FX{p}}) \).

Last-use updates at let bindings need to reset the recorded last reads for any written variable, and add the bound variable to the last reads for
any written variable, \ie,
\[(\mute{\HDEP;\SDEP}) \oplus_x \FX{(\RD:q\,;\, \WR : p)} := (\mute{\HDEP, (\FX{p},x)\mapsto x\ ;\ (\SDEP,(\FX{p},x)\mapsto\qbot) \oplus_x \FX{q}  } ), \]
where \(x\) is the let-bound variable and \( \mute{\SDEP\oplus_x\FX{q}} := \SDEP,\{ y\mapsto \SDEP(y),x\mid y \in \FX{q} \} \)
adds \(x\) into each set pointed to by \(\FX{q}\).

\subsection{Statics}\label{sec:hardsoft:statics}

\Cref{fig:hardsoft:syntax} shows the required changes to the synthesis rules from \Cref{fig:graphir:synthesis} (and also indicates
the needed changes for the checking rules from \Cref{fig:graphir:checking}). With the overloaded operations on effects and dependencies,
the rules for mutable references need to classify their effects on the operands. That is, reference allocations cause a read
on the used allocation capability and its reachable aliases, dereferences a read on the target reference and its aliases, and
assignments a write, accordingly.

\subsection{Metatheory}\label{sec:hardsoft:metatheory}

The theorems and proofs for the graph IR with hard and soft dependency are for the most part identical
to the previous system, and we just repeat the most relevant theorems without proof.

\begin{lemma}[Synthesis Invariant]\label{lem:hardsoft:deps-invariant}
  Dependencies are completely determined by the context and effect, as follows:
  \begin{enumerate}
    \item If\ \ \(\GSD[\varphi]\vdash n : \ty[q]{T}\ \EPS\yields{\bm{n} \has \DEP},\) then \(\DEP = \DELTA\vert_{\FX{\qsat{\EPS}}}\).
    \item If\ \ \(\GSD[\varphi]\vdash g : \ty[q]{T}\ \EPS\yields{\bm{g}\has \DEP},\) then \(\DEP = \DELTA\vert_{\FX{\qsat{\EPS}}}\).
  \end{enumerate}
\end{lemma}
\begin{proof}
Analogous to the proof of \Cref{lem:deps-invariant}.
\end{proof}

\begin{lemma}[Soundness of Synthesis]\label{lem:hardsoft:synth-soundness}
  Synthesis produces well-typed annotated programs:
  \begin{enumerate}
    \item If\ \ \(\GSD[\flt]\ts n : \ty[q]{T}\ \EPS\yields{\bm{n} \has \DEP},\) then \(\GSD[\varphi]\vdash \bm{n} : \ty[q]{T}\ \EPS\).
    \item If\ \ \(\GSD[\flt]\ts g : \ty[q]{T}\ \EPS\yields{\bm{g} \has \DEP},\) then \(\GSD[\varphi]\vdash \bm{g} : \ty[q]{T}\ \EPS\).
  \end{enumerate}
\end{lemma}
\begin{proof}
  Analogous to the proof of \Cref{lem:synth-soundness}.
\end{proof}

\begin{lemma}[Synthesis is Total]\label{lem:hardsoft:synth-totality}\hfill
  \begin{enumerate}
    \item If\ \ \(\GS[\varphi]\tsM n : \ty[q]{T}\ \EPS\) then \(\forall \DELTA.\ \exists \mute{\bm{n}}.\ \exists\DEP.\) \(\GSD[\varphi]\vdash n : \ty[q]{T}\ \EPS\yields{\bm{n} \has \DEP}\) and \(\mute{\bm{n}}\) erases back to \(n\).
    \item If\ \ \(\GS[\varphi]\tsM g : \ty[q]{T}\ \EPS\) then \(\forall \DELTA.\ \exists \mute{\bm{g}}.\ \exists\DEP.\) \(\GSD[\varphi]\vdash g : \ty[q]{T}\ \EPS\yields{\bm{g} \has \DEP}\) and \(\mute{\bm{g}}\) erases back to \(g\).
  \end{enumerate}
\end{lemma}
\begin{proof}
  Analogous to the proof of \Cref{lem:synth-totality}.
\end{proof}

\begin{corollary}[End-to-end Type Preservation]
  For any well-formed context and compatible \(\DELTA\), there is a type/effect/qualifier-preserving translation from the direct-style \directlang-calculus into
  the \irlang graph IR with hard and soft dependencies.
  \end{corollary}
  \begin{proof}
  By lemmas \ref{lem:type_preservation:translation}, \ref{lem:hardsoft:synth-totality}, and \ref{lem:hardsoft:synth-soundness}.
\end{proof}

\begin{theorem}[Progress]\label{thm:hardsoft:progress}
    If \(\ \csx[\DOM(\Sigma)]{\varnothing}{\Sigma} \ts z.g\has\DEP : \ty[q]{T}\ \EPS\), then either \(g\) is a location \(\ell\in\DOM(\Sigma)\), or
    for any store \(\sigma\) where \(\csx[\DOM(\Sigma)]{\varnothing}{\Sigma} \ts \sigma\), there exists
    a graph term \(g'\), store \(\sigma'\), and dependency \(\mute{\DEP'}\) such that \(\sigma \mid z.g\has\DEP  \redg \sigma'\mid z.g'\has\mute{\DEP'}\).
\end{theorem}

\begin{theorem}[Preservation]\label{thm:hardsoft:preservation}\hfill\\
    \begin{adjustbox}{scale=.9}
    \begin{minipage}{\linewidth}
      \infrule{\csx[\DOM(\Sigma)]{\varnothing}{\Sigma}  \ts z.g\has\DEP : \ty[q]{T}\ \EPS\quad \WF{\cdsx{\varnothing}{\Sigma}{\mute{\pointsto{z}}.z}}\quad \cdsx[\DOM(\Sigma)]{\varnothing}{\Sigma}{\mute{\pointsto{z}}} \ts \sigma\quad\sigma\mid z.g\has\DEP \redg \sigma' \mid z.g'\has\mute{\DEP'}
      }{
        \exists \Sigma' \supseteq \Sigma.\;\exists p \subq\DOM(\Sigma'\setminus\Sigma).\quad\csx[\DOM(\Sigma')]{\varnothing}{\Sigma'} \ts z.g'\has\mute{\DEP'} : \ty[q,p]{T}\ \FX{\EPS,p}\qquad\qquad\\[1ex]
        \phantom{\exists \Sigma' \supseteq \Sigma.\;\exists p \subq\DOM(\Sigma'\setminus\Sigma).\ \ \,\,}\WF{\cdsx{\varnothing}{\Sigma'}{\mute{\pointsto{z}}.z}}\qquad  \cdsx[\DOM(\Sigma')]{\varnothing}{\Sigma'}{\mute{\pointsto{z}}} \ts \sigma'
      }
    \end{minipage}
    \end{adjustbox}
\end{theorem}
\begin{corollary}[Preservation of Separation]\label{thm:hardsoft:preservation_of_separation}
    Interleaved executions preserve types and disjointness:\\
  \begin{adjustbox}{scale=.88}
  \begin{minipage}{\linewidth}
    \infrule{%
  {\begin{array}{l@{\quad}l@{\quad}l@{\quad}l}
  \csx[\DOM(\Sigma)]{\varnothing}{\Sigma} \vdash g_1\has\DEP[1]: \ty[q_1]{T_1}\ \EPS[1] & \sigma\phantom{'}\mid z.g_1\has\DEP[1]  \redg \sigma'\phantom{'} \mid z.g_1'\has\mute{\DEP[1]'} & \csx[\DOM(\Sigma)]{\varnothing}{\Sigma} \ts \sigma & \WF{\cdsx{\varnothing}{\Sigma}{\mute{\pointsto{z}}.z}} \\[1ex]
  \csx[\DOM(\Sigma)]{\varnothing}{\Sigma} \vdash g_2\has\DEP[2]: \ty[q_2]{T_2}\ \EPS[2]  & \sigma'\mid z.g_2\has\DEP[2]  \redg \sigma'' \mid z.g_2'\has\mute{\DEP[2]'}          & q_1 \overlap q_2 \subq \qbot       &
  \end{array}}
  }{{\begin{array}{ll@{\quad}l@{\quad}l}
  \exists p_1\;p_2\;\EPSPR[1]\;\EPSPR[2]\;\Sigma'\;\Sigma''. & \csx[\DOM(\Sigma')\phantom{'}]{\varnothing}{\Sigma'\phantom{'}} \ts z.g_1'\has\mute{\DEP[1]'} : \ty[p_1]{T_1}\ \EPSPR[1] & \Sigma'' \supseteq \Sigma' \supseteq \Sigma \\[1ex]
                                       & \csx[\DOM(\Sigma'')]{\varnothing}{\Sigma''} \ts z.g_2'\has\mute{\DEP[2]'} : \ty[p_2]{T_2}\ \EPSPR[2]                     & p_1 \overlap p_2 \subq \qbot
  \end{array}}}
  \end{minipage}%
  \end{adjustbox}
\end{corollary}
\begin{corollary}[Dependency Safety]\label{coro:hardsoft:dep-safety}
  Evaluation respects the order of effect dependencies for well-typed graph IR terms, \ie, an effectful graph node is executed only if all its dependencies
  are resolved in the store.
\end{corollary} %
\section{Contextual Equivalence - The Direct-Style \maybelang{}-Calculus}
\label{sec:direct-lr}
We apply a logical relations approach following \cite{logical-approach,DBLP:conf/popl/AhmedDR09,DBLP:conf/ppdp/BentonKBH07} to support relational reasoning with respect to the \emph{observational equivalence} of two programs. We define binary logical relations over reachability types (the \maybelang{}-calculus in \secref{sec:directstyle}), and prove the soundness of the equational rules.
Our development is based on a framework for modeling reachability types with logical relations developed in parallel with this work \cite{bao2023logrel}. To make the present report self-contained, pieces of \citet{bao2023logrel} are repeated in this section without further reference.
To avoid technical complications, we choose a model that allows mutable references to contain only  first-order values, consistent with the previous sections. The definition of the logical relation can be extended to support higher-order references using well-established techniques such as step-indexing~\cite{ahmed2004semantics,DBLP:conf/popl/AhmedDR09,DBLP:journals/toplas/AppelM01}, which we leave as future work.

\subsection{High-level Overview of the Proofs}
A program $t_1$ is said to be \emph{contextually equivalent} to another program $t_2$, written as $\G[\flt] \models t_1 \equiva t_2 : \ty[p]{T}\ \EPS$, if for any program context $C$ with a hole of type $\ty[p]{T}\ \EPS$,
if $C[t_1]$ has some (observable) behavior, then so does $C[t_2]$.
The definition of context $C$ can be found in ~\secref{sec:direct_context}.

Following the approach of \citet{logical-approach} and related prior works
\cite{DBLP:conf/popl/AhmedDR09}, we define
a judgement for logical equivalence using binary logical relations,
written as $\G[\flt] \models t_1 \equivlog t_2 : \ty[q]{T} \ \EPS$.

The high-level structure of the proof is the following:
\begin{itemize}
    \item Soundness~(\thmref{thm:direct_lr_soundness}, \secref{sec:direct_soundness}). We show that the logical relation is sound with respect to contextual equivalence:
          \[
              \G[\flt] \models t_1 \equivlog t_2 : \ty[q]{T} \ \EPS \text{ implies } \G[\flt] \models t_1 \equiva t_2 : \ty[q]{T} \ \EPS.
          \]

    \item Compatibility lemmas~(\secref{sec:direct_cls}). We show that the logical relation is compatible with syntactic typing.
\end{itemize}
These results can be used to prove the soundness of the equational rules (\secref{sec:direct_equiv}).

\subsection{Contextual Equivalence}
\label{sec:direct_context}
\begin{figure}[t]\small
	\begin{mdframed}
		\judgement{Context for Contextual Equivalence}{}
		\[\begin{array}{l@{\ \ }c@{\ \ }l@{\qquad\qquad\ }l@{\ \ }c@{\ \ }l}
				{C} & ::= & \square \mid C\ t \mid t\ C \mid  \lambda x. C \mid \tref_t ~C \mid \tref_C ~ t \mid\ !~{C} \mid {C} := {t} \mid {t} := {C} \mid \tlet~{x = C}~\tin~t \mid \tlet~{x = t}~\tin~C &  &  & \\
			\end{array}\]
		\judgement{Context Typing Rules}{\BOX{C : (\G[\flt]; \ty[q]{T}\ \EPS) \carrow (\GP[\flt]; \ty[q]{T}\ \EPS)}}
		\begin{minipage}[t]{1.0\linewidth}\vspace{0pt}
			\infrule[c-hole]{
				\G[\flt] \ts \ty[q]{{T}}\ \EPS <: \ty[q']{{T'}}\ \EPSPR
			}{
				\square: (\G[\flt]; \ty[q]{{T}}\ \EPS) \carrow (\GP[\fltp]; \ty[q']{{T'}}\ \EPSPR)
			}
			\vgap
			\infrule[c-app-1]
			{
			C: (\G[\flt]; \ty[r]{{U}} \ \EPS[1]) \carrow
			(\GP[\fltp]; ((x: \ty[\qsat{p}\overlap \qsat{q}]{{T}}) \to^{\EPS[3]}  \ty[r']{{U'}})^{q'} \ \EPS[4])
			\quad
			\GP[\fltp] \ts t_2: \ty[p]{{T}} \ \EPS[2] \\
			x\notin\FV(U') \quad
			r' \subq \fltp, x \quad \EPS[3] \subq \fltp,x \quad \theta = [p/x]
			}
			{
			C \ t_2: (\G[\flt]; \ty[r]{{U}} \ \EPS[1]) \carrow (\GP[\fltp]; (\ty[r']{U'} \ \EPS[4] \EFFSEQ \EPS[2] \EFFSEQ \EPS[3])\theta)
			}
			\vgap
			\infrule[c-app-2]
			{
			\GP[\fltp] \ts t_1: ((x: \ty[\qsat{p} \overlap \qsat{q}]{{T}}) \to^{\EPS[4]}  \ty[r']{{U'}})^{q'} \ \EPS[2]
			\quad
			C: (\G[\flt]; \ty[r]{{U}} \ \EPS[1]) \carrow (\GP[\fltp]; \ty[p]{{T}}  \ \EPS[3]) \\
			x\notin\FV(U') \quad
			r' \subq \fltp,x \quad \EPS[3] \subq \fltp,x \quad \theta = [p/x]
			}
			{
			t_1 \ C: (\G[\flt]; \ty[r]{{U}} \ \EPS[1]) \carrow (\GP[\fltp]; (\ty[r']{U'} \ \EPS[2] \EFFSEQ \EPS[3] \EFFSEQ \EPS[4])\theta)
			}
			\vgap
			\infrule[c-$\lambda$]
			{  C : (\G[\flt]; \ty[r]{{S}}\ \EPS) \carrow ( (\G'\ ,\ x: \ty[p]{T})^{q, x}; \ty[r]{{U}} \ \EPSPR)
				\quad q \subq \flt
			}
			{
				\lambda x. C: (\G[\flt]; \ty[r]{{S}} \ \EPS) \carrow (\GP[\fltp]; ((x: \ty[p]{{T}}) \to^{\EPSPR} \ty[r]{{U}})^{q} \ \PURE)
			}
			\vgap
			\infrule[c-ref-1]
			{ \GP[\fltp] \ts t: \ty[q]{\Type{Alloc}} \ \EPS[1]  \quad
				C: (\G[\flt]; \ty[r]{{T}} \ \EPS[2]) \carrow (\GP[\fltp]; \ty[\qbot]{{B}} \ \EPSPR )
			}
			{ \tref_t \ C: (\G[\flt];  \ty[r]{{T}} \ \EPS[2]) \carrow (\GP[\fltp]; \ty[\qbot]{\TRef ~ {B}} \ \EPS[1]\EFFSEQ\EPSPR\EFFSEQ \FX{q} )
			}
			\vgap
			\infrule[c-ref-2]
			{
				\GP[\fltp] \ts t: \ty[\qbot]{{B}} \ \EPS[2]  \quad
				C: (\G[\flt]; \ty[r]{{T}} \ \EPS[1]) \carrow (\GP[\fltp]; \ty[q]{\Type{Alloc}} \ \EPSPR[1] )
			}
			{
				\tref_C \ t: (\G[\flt];  \ty[r]{{T}} \ \EPS[1]) \carrow (\GP[\fltp]; \ty[\qbot]{\TRef ~ {B}} \ \EPSPR[1]\EFFSEQ\EPS[2]\EFFSEQ \FX{q} )
			}
			\vgap
			\infrule[c-!]
			{ C: (\G[\flt]; \ty[r]{{T}} \ \EPS) \carrow (\GP[\fltp]; \ty[q']{(\TRef ~{B})} \ \EPSPR )
			}
			{\ts ! \ C: (\G[\flt]; \ty[r]{{T}} \ \EPS) \carrow (\GP[\fltp]; \ty[\qbot]{{B}} \ \EPSPR )
			}
			\vgap
			\infrule[c-:=-1]
			{ C : (\G[\flt] ; \ty[r]{{T}} \ \EPS) \carrow (\GP[\fltp] ; \ty[q']{(\TRef \ {B})} \ \EPSPR) \quad
				\GP[\fltp] \ts t_2: \ty[\qbot]{{B}} \ \EPSPR
			}
			{ C := t_2 : (\G[\flt]; \ty[r]{{T}} \ \EPS) \carrow (\GP[\fltp]; \ty[\qbot]{\TUnit} \ \EPSPR)
			}
			\vgap
			\infrule[c-:=-2]
			{ \GP[\fltp] \ts t_1 : \ty[q']{(\TRef\ {B})} \ \EPSPR \quad
				C : (\G[\flt];  \ty[r]{{T}} \  \EPS) \carrow (\GP[\fltp]; \ty[\qbot]{{B}} \ \EPSPR)
			}
			{ t_1 := C : (\G[\flt]; \ty[r]{{T}} \ \EPS) \carrow (\GP[\fltp];  \ty[\qbot]{\TUnit} \ \EPSPR)
			}
			\vgap
			\infrule[c-let-1]
			{
			C: (\G[\flt]; \ty[r]{{U}} \ \EPS[1]) \carrow
			((\GP, x: \ty[\qsat{p} \overlap \qsat{\fltp}]{{T}})^{\fltp, x};  \ty[r']{{U'}} \ \EPS[2])
			\quad
			\GP[\fltp] \ts t: \ty[p]{{T}} \ \EPS[3] \\
			x\notin\FV(U') \quad \theta = [p/x]
			}
			{
			\tlet ~ x = t~ \tin~ C: (\G[\flt];  \ty[r]{{U}} \ \EPS[1]) \carrow (\GP[\fltp];(\ty[r']{U'} \ \EPS[2] \EFFSEQ \EPS[3])\theta)
			}
			\vgap
			\infrule[c-let-2]
			{
			(\GP, (x: \ty[\qsat{p} \overlap \qsat{\fltp}]{{T}}))^{\fltp,x} \ts t: \ty[r']{{U'}} \ \EPS[2]
			\quad
			C: (\G[\flt]; \ty[r]{{U}} \ \EPS[1]) \carrow (\GP[\fltp]; \ty[p]{{T}}  \ \EPS[3]) \\
			x\notin\FV(U') \quad \theta = [p/x]
			}
			{
			\tlet ~ x = C~ \tin ~t : (\G[\flt]; \ty[r]{{U}} \ \EPS[1]) \carrow (\GP[\fltp]; (\ty[r']{U'} \ \EPS[2] \EFFSEQ \EPS[3])\theta)
			}
		\end{minipage}%
		\caption{Context typing rules for the \maybelang{}-Calculus.}
		\label{fig:direct_context}
	\end{mdframed}
\end{figure}

Unlike reduction contexts $E$ in \figref{fig:directstyle:semantics}, contexts $C$ for reasoning about the equivalence allow a ``hole'' to appear in any place.
We write $C: (\G[\flt]; \ty[q]{T} \  \EPS) \carrow (\GP[\fltp]; \ty[q']{T'}\ \EPSPR)$ to mean that the context $C$ is a program of type $\ty[q']{T'} \ \EPSPR$ (closed under $\GP[\fltp]$) with a hole that can be filled with any program of type $\ty[q]{T} \ \EPS$ (closed under $\G[\flt]$).
The typing rules for well-typed contexts imply that if
$\G[\flt] \ts t : \ty[q]{T} \ \EPS$ and $C:  (\G[\flt]; \ty[q]{T} \  \EPS)\carrow (\GP[\fltp]; \ty[q']{T'}\ \EPSPR)$ hold, then $\GP[\fltp] \ts C[t] : \ty[q']{T'} \ \EPSPR$.
\figref{fig:direct_context} shows the typing rules for well-typed contexts.

Two well-typed terms, $t_1$ and $t_2$, under type context $\G[\flt]$, are \emph{contextually equivalent} if any occurrences of the first term in a closed term can be replaced by the second term without affecting the \emph{observable results} of reducing the program, which is formally defined as follows:
\begin{definition}[Contextual Equivalence]\label{def:standard_equiv} We say $t_1$ is \emph{contextually equivalent} to $t_2$, written as $\G[\flt] \models t_1 \equiva t_2: \ty[p]{T}\ \EPS$, if $\G[\flt] \ts t_1: \ty[q]{T}\ \EPS$, and $\G[\flt] \ts t_2: \ty[q]{T}\ \EPS$, and:

    \[
        \forall \, C: (\G[\flt];  \ty[q]{T}\ \EPS) \carrow ( \emptyset; \ty[\qbot]{\TUnit} \ \PURE). \
        C[t_1]\downarrow \ \Longleftrightarrow \ C\hole{t_2} \downarrow.
    \]
\end{definition}
\noindent We write $t\downarrow$ to mean term $t$ terminates, if $\emptyset \mid t \mredv{\ast} \sigma \mid v $ for some value $v$ and final store $\sigma$.

The above definition is standard \cite{DBLP:conf/popl/AhmedDR09} and defines a partial program equivalence.
However, since we focus on a total fragment of the \directlang{}-calculus here, program termination can not be used as an observer for program equivalence.
We will thus rely on the following refined version of contextual equivalence using
Boolean contexts:
\begin{align*}
     & \forall \, C: (\G[\flt];  \ty[q]{T}\ \EPS) \carrow ( \emptyset; \ty[\qbot]{\Type{Bool}} \ \PURE). \ \exists \ \sigma, \sigma', v. \\
     & \qquad \emptyset \mid C[t_1] \mredv{\ast} \sigma \mid v\; \wedge\; \emptyset \mid C[t_2] \mredv{\ast} \sigma' \mid v.
\end{align*}
That is to say, we consider two terms contextually equivalent if they yield the same answer value in all Boolean contexts.

\subsection{The Model}
Following other prior works~\cite{DBLP:conf/icfp/ThamsborgB11,DBLP:conf/ppdp/BentonKBH07,ahmed2004semantics},
we apply Kripke logical relations to the \maybelang{}-calculus.
Our logical relations are indexed by types and store layouts via \emph{worlds}.
This allows us to interpret $\TRef \ \Type{B}$ as an allocated location that holds values of type $\Type{B}$.
The invariant that all allocated locations hold well-typed values with respect to the world must hold in the pre-state and be re-established in the post-state of a computation.
The world may grow as more locations may be allocated.
It is important that this invariant must hold in future worlds, which is commonly referred as \emph{monotonicity}.

Considering the restriction to first-order references here,  our store layouts are always ``flat'', \ie, free of cycles.
The notion of world for the \maybelang{}-calculus is defined in the following:
\begin{definition}[World]\label{def:world} A world $\W$ is a triple  $(L_1, L_2, f)$, where
    \begin{itemize}
        \item $L_1$ and $L_2$ are finite sets of locations.
        \item $f \subseteq (L_1 \times L_2)$ is a partial bijection.
    \end{itemize}
\end{definition}
A world is meant to define relational stores.
The partial bijection captures the fact that a relation holds under permutation of locations.

If $\W = (L_1, L_2, f)$ is a world, we refer to its components as follows:
\[
    \begin{array}{l l l l}
        \W(\ell_1, \ell_2) & = & \begin{cases}
                                     (\ell_1, \ell_2) \in f & \text{when defined} \\
                                     \emptyset              & \text{otherwise}
                                 \end{cases} \\
        \DOM_1(\W)         & = & L_1                               \\
        \DOM_2(\W)         & = & L_2                               \\
    \end{array}
\]
If $\W$ and $\W'$ are worlds, such that
$
    \DOM_1(\W) \overlap \DOM_1(\W') = \DOM_2(\W) \overlap \DOM_2(\W') = \emptyset
$,
then $\W$ and $\W'$ are called disjoint, and we write $\W \extends \W'$ to mean extending $\W$ with a disjoint world $\W'$.
Let $\sigma_1$ and $\sigma_2$ be two stores. We write $\WFRS{\sigma_1}{\sigma_2}{\W}$ to mean $\W = (\DOM(\sigma_1), \DOM(\sigma_2), f)$.

Our world definition allows us to specify that the domains of two relational stores may grow during a computation, but does not cover store operations, which is important when proving the soundness of equational rules.
Like prior works (\eg, ~\cite{DBLP:conf/icfp/ThamsborgB11,DBLP:conf/ppdp/BentonKBH07}), we use effects as a refinement for the definition of world.
The notation $\EPS$ denotes read/write effects, and $\omega$ means allocation occurs during a computation.
Local reasoning is enabled by reachability qualifiers and read/write effects,
meaning that what is preserved during an effectful computation are the locations that are \emph{not} mentioned in the read/write effects.
This is a common technique used in reasoning about frames in Hoare-style logics, \eg, separation logic~\cite{DBLP:conf/lics/Reynolds02}.
This treatment is also applicable to our refined effect system (\secref{sec:extension-with-soft}), where framing is achieved through write effects --
an established technique in Dafny~\cite{DBLP:conf/lpar/Leino10} and region logics~\cite{DBLP:journals/jacm/BanerjeeNR13,DBLP:conf/ecoop/BaoLE15}. In this case, a frame indirectly describes the locations that a computation may not change ~\cite{DBLP:journals/tse/BorgidaMR95}.
Framing allows the proof to carry properties of effectful terms, such as function applications, since properties that are true for unchanged locations will remain valid~\cite{DBLP:journals/fac/BaoLE18}.

\subsection{Interpretation of Reachability}
\label{sec:direct_interp_rq}
\begin{figure}[t]
  \begin{mdframed}\small
    \judgement{Interpretation of Value Reachability}
    \[
      \begin{array}{l}
        \dvalocss{\omega}  =  \emptyset \qquad   \dvalocss{\tunit} = \emptyset \qquad \dvalocss{c} = \emptyset \qquad \dvalocss{\ell}  = \{ \, \ell \, \} \\
        \dvalocss{\lambda x. t} = \{ \ell \mid \ell \in \FV(t) \land \ell \in \Loc \}                                                                   \\
      \end{array}
    \]
    \judgement{Interpretation of Reachability Qualifiers}
    \[
      \begin{array}{l@{\ \,}c@{\ \,}l@{\qquad\qquad\qquad\qquad\qquad\qquad \ \ }l}
        \dvarslocs{q} & \DEF & \{ \, \ell \mid \,  \ell \in q \, \land \, \ell \ \in \Loc \} \\
      \end{array}
    \]
    \judgement{Reachability Predicates}
    \[
      \begin{array}{l@{\ \,}c@{\ \,}l@{\qquad\qquad\qquad\qquad\qquad\qquad \ \ }l}
        \dvalq{\sigma}{v}{L} & \DEF & (\DOM(\sigma) \, \overlap \,  \dvalocss{v}) \subseteq L  \
      \end{array}
    \]
    \caption{Interpretation of reachability qualifiers.}
    \label{fig:direct_reachhability}
  \end{mdframed}
\end{figure}

In the \maybelang{}-calculus, reachability qualifiers are used to specify desired separation or permissible overlapping of reachable locations from a function's argument and its body.
\figref{fig:direct_reachhability} shows the interpretation of reachability qualifiers. As in the \maybelang{} calculus, values cannot be cyclic, we axiomatize the definition of reachability, without proving termination.
Here, we assume free variables are already substituted with values.

We use $\dvalocss{v}$ to define the set of locations that are reachable from a given value $v$.
Base type values, \ie, $\omega$ of type $\Type{Alloc}$, $\tunit$ of type $\TUnit$, and other constants $c$ of other base types $\Type{B}$, do not reach any store locations.
Thus, they reach the empty set of locations.
A location $\ell$ can only reach itself. Thus, its reachable set is the singleton set $\{\ell\}$.
The set of locations that are reachable from a function value $\lambda x.t$ are the set of the locations appearing in the function body.

We overload the function $\text{locs}$, and write $\dvarslocs{q}$ to mean the set of locations reachable from qualifier $q$, which are the set of the locations appeared in $q$.
A bound variable may appear in $q$, and serves as a placeholder to specify the set of locations that a function's return value may reach.
See \secref{sec:direct_binary} for details.
The notation $\dvalq{\sigma}{v}{L}$ is a predicate that asserts the set of locations that are reachable from $v$ in store $\sigma$ is a subset of $L$, where $L$ is a set of locations.

\subsection{Binary Logical Relations for \maybelang{}}
\label{sec:direct_binary}
This section presents the definition of binary logical relations for \maybelang{}.
Following the approach of \citet{logical-approach}, we define the binary logical relation for logical equivalence in two steps:
\begin{enumerate}
    \item We define binary interpretations on pairs of closed values, and pairs of closed terms.
    \item We define the logical equivalence relation on open terms, $\G[\flt] \models t_1 \equivlog t_2 : \ty[q]{T} \ \EPS$, by lifting the value and term relations to open terms using a closing substitution.
\end{enumerate}

The $\maybelang{}$-calculus has a dependent type system, where types may mention term variables.
We could either define logical relations indexed by two types with variants on qualifiers (\ie, the choice of locations after closing substitution); or indexed by a type where performing closing substitution to qualifiers in the definition.
Here, we choose the latter.
Thus, a relational value substitution is a parameter of the definition of logical relations.

The relational value substitution has to satisfy the context interpretation.
We define the interpretation of typing contexts:
\[
    \footnotesize
    \begin{array}{l l l}
        \UG{\emptyset^{\flt}}          & = & \emptyset                                                                                                                                                                     \\
        \UG{(\G, x: \ty[q]{T})^{\flt}} & = & \{(\W, \gamma\extends(x \mapsto (v_1, v_2))) \mid (\W, \gamma) \in \UG{\G[\flt]} \, \land \, \flt \subq \DOM(\G) \, \land \, q \subq \DOM(\Gamma) \, \land                    \\
                                       &   & \DVTC{\W}{v_1}{v_2} \in \DGMV{T} \, \land \,                                                                                                                                  \\
                                       &   & (\forall\, q, q'. \,  q \subseteq \flt \, \land \,  q' \subq \flt \,  \land \,  \Rightarrow                                                                                   \\
                                       &   & \qquad\qquad (\dvarslocs{\gamma_1(\qsat{q})} \, {\overlap} \, \dvarslocs{\gamma_1(\qsat{q'})} \subseteq \dvarslocs{\gamma_1(\qsat{q} \, {\overlap} \, \qsat{q'})} \, \land \, \\
                                       &   & \qquad\qquad \dvarslocs{\gamma_2(\qsat{q})} \, {\overlap} \, \dvarslocs{\gamma_2(\qsat{q'})} \subseteq \dvarslocs{\gamma_2(\qsat{q} \, {\overlap} \, \qsat{q'})}) )\}         \\
    \end{array}
\]

In the above definition, $\gamma$ ranges over relational value substitutions that are finite maps from variables $x$ to pairs of values $(v_1, v_2)$.
If $\gamma(x) = (v_1, v_2)$, then $\gamma_1(x)$ denotes $v_1$ and $\gamma_2(x)$ denotes $v_2$.
We write $\gamma_1(q)$ and $\gamma_2(q)$ to mean substituting the free variables in $q$ with respect to the relational value substitution $\gamma$.

\begin{figure*}[t]
    \begin{mdframed}
        \judgement{Value Interpretation of Types and Terms}{\BOX{\maybelang{}}}
        \begin{mathpar}
            \begin{footnotesize}
                \begin{array}{@{}r@{\hspace{1ex}}l@{\hspace{1ex}} l}
                    \DGMV{\Type{Alloc}}                     &      & \{ \DVTC{\W}{\omega}{\omega} \}                                                                                                                                                                                         \\
                    \DGMV{\TUnit}                           & =    & \{ \DVTC{\W}{\tunit}{\tunit} \}                                                                                                                                                                                         \\

                    \DGMV{\Type{Bool}}                      & =    & \{ \DVTC{\W}{v}{v} \mid v = \text{true}  \, \lor  \, v = \text{false} \}                                                                                                                                                \\

                    \DGMV{\TRef \ \Type{B}}                 & =    & \{ \DVTC{\W}{\ell_1}{\ell_2} \mid \forall \, \sigma_1, \sigma_2. \, \WFRS{\sigma_1}{\sigma_2}{\W} \, \land \, \ell_1 \in \DOM(\sigma_1) \, \land \, \ell_2 \in \DOM(\sigma_2) \, \land \, \W(\ell_1, \ell_2) \, \land   \\
                                                            &      & \DVTC{\W}{\sigma_1(\ell_1)}{\sigma_2(\ell_2)} \in \DGMV{\Type{B}} \}                                                                                                                                                    \\
                    \\
                    \DGMV{(x:\ty[p]{{T}}) \to^{\EPS} U^{r}} & =    & \{ \DVTC{\W}{\lambda x. t_1}{\lambda x. t_2} \mid \, \dvalocss{\lambda x. t_1} \subq \DOM_1 (\W) \, \land  \, \dvalocss{\lambda x. t_2}  \subq \DOM_2(\W) \, \land \,                                                   \\
                                                            &      & (\forall v_1, v_2, \W', \sigma_1, \sigma_2. \WFRS{\sigma_1}{\sigma_2}{\W \extends \W'}  \, \Rightarrow \, \DVTC{\W\extends \W'}{v_1}{v_2}\in \DGMV{T}\, \Rightarrow                                                     \\
                                                            &      & \qquad \dvalocss{\lambda x. t_1} \, {\overlap} \, \dvalocss{v_1} \, \subq \dvarslocs{\gamma_1(p)} \Rightarrow \dvalocss{\lambda x. t_2} \, {\overlap} \, \dvalocss{v_2} \, \subq \dvarslocs{\gamma_2(p)} \Rightarrow    \\
                                                            &      & \qquad \exists \, \W'', \sigma_1', \sigma_2', v_1', v_2'. \, \sigma_1 \mid t_1[x \mapsto v_1] \mredv{\ast} \sigma_1' \mid v_1' \, \land \, \sigma_2 \mid t_2[x \mapsto v_2]  \mredv{\ast} \sigma_2' \mid v_2'  \, \land \\
                                                            &      & \qquad \qquad \WFRS{\sigma_1'}{\sigma_2'}{\W \extends \W' \extends \W''} \, \land  \, \DVTC{\W \extends \W' \extends \W''}{v_1'}{v_2'} \in \DGMV{U} \, \land                                                            \\
                                                            &      & \qquad \qquad (x \in r \Rightarrow  \dvalq{\sigma_1}{v_1'}{(\dvalocss{\gamma_1(r)} \overlap \dvalocss{\lambda x. t_1}  \, \cup \, \dvalocss{v_1})} \, \land \,                                                          \\
                                                            &      & \qquad \qquad \qquad \qquad  \dvalq{\sigma_2}{v_2'}{(\dvalocss{\gamma_2(r)} \overlap \dvalocss{\lambda x. t_2} \, \cup \, \dvalocss{v_2})}) \, \land                                                                    \\
                                                            &      & \qquad \qquad (x \not\in r \Rightarrow \dvalq{\sigma_1}{v_1'}{(\dvalocss{\gamma_1(r)} \overlap \dvalocss{\lambda x. t_1})}  \, \land  \,                                                                                \\
                                                            &      & \qquad \qquad\qquad \qquad\dvalq{\sigma_2}{v_2'}{(\dvalocss{\gamma_2(r)} \overlap \dvalocss{\lambda x . t_2})}) \, \land                                                                                                \\
                                                            &      & \qquad \qquad (x \in \EPS \Rightarrow  \DEPS{\sigma_1}{\dvalocss{\gamma_1(\EPS \EFFSEQ \FX{p})}}{\sigma_1'} \, \land \, \DEPS{\sigma_2}{\dvalocss{\gamma_2(\EPS \EFFSEQ \FX{p})}}{\sigma_2'}) \, \land                  \\
                                                            &      & \qquad \qquad (x \not\in \EPS \Rightarrow \DEPS{\sigma_1}{\dvalocss{\gamma_1(\EPS)}}{\sigma_1'} \, \land \, \DEPS{\sigma_2}{\dvalocss{\gamma_2(\EPS)}}{\sigma_2'}))  \}                                                 \\
                    \DEPS{\sigma}{\EPS}{\sigma'}            & \DEF & \forall l \in \DOM(\sigma). \sigma(l) = \sigma'(l) \lor l \in \EPS                                                                                                                                                      \\

                    \\
                    \DGMt{\flt}{\ty[q]{T}\ \EPS}            & =    & \{ \DVTC{\W}{t_1}{t_2} \mid \forall \, \sigma_1, \sigma_2. \WFRS{\sigma_1}{\sigma_2}{\W} \, \land \, \exists \, \W', \sigma_1', \sigma_2', v_1, v_2.\,  t_1 \mid \sigma_1 \mredv{\ast} v_1 \mid \sigma_1' \, \land      \\
                                                            &      & t_2 \mid \sigma_2 \mredv{\ast} v_2 \mid \sigma_2' \, \land \, \WFRS{\sigma_1'}{\sigma_2'}{\W\extends \W'} \, \land  \, \DVTC{\W \extends \W'}{v_1}{v_2} \in \DGMV{T}  \, \land                                          \\
                                                            &      & \dvalq{\sigma_1}{v_1}{(\dvalocss{\gamma_1(\flt \overlap q)})} \, \land \, \dvalq{\sigma_2}{v_2}{(\dvalocss{\gamma_2(\flt \overlap q)})} \, \land                                                                        \\
                                                            &      & \DEPS{\sigma_1}{\dvalocss{\gamma_1(\EPS)}}{\sigma_1'} \, \land \, \DEPS{\sigma_2}{\dvalocss{\gamma_2(\EPS)}}{\sigma_2'} \}                                                                                              \\
                \end{array}
            \end{footnotesize}
        \end{mathpar}
        \caption{Binary value and term interpretation for the \maybelang{}-calculus.}
        \label{fig:direct_binary}
    \end{mdframed}
\end{figure*}

\paragraph{The Binary Value Interpretation.}
The definition of binary value interpretation of types is shown in \figref{fig:direct_binary}.
The relational interpretation of type ${T}$, written as $\DGMV{T}$, is a set of tuples of form $\DVTC{\W}{v_1}{v_2}$, where $v_1$ and $v_2$ are values, and $\W$ is a world.
We say $v_1$ and $v_2$ are related at type ${T}$ with respect to $\W$.

\paragraph{Ground Types.}
We restrict the base types $B$ to $\Type{Alloc}$, $\TUnit$ and $\Type{Bool}$ to streamline the presentation.
A pair of allocation capabilities $(\omega,\omega)$ are related at type $\Type{Alloc}$.
A pair of unit values $(\tunit,\tunit)$ are related at type $\TUnit$.
A pair of boolean values are related if they are both \text{true} or \text{false}.
A pair of locations $(\ell_1, \ell_2)$ are related if they are in the domain of the relational store with respect to $\W$, ($\WFRS{\sigma_1}{\sigma_2}{\W}$), such that $\W(\ell_1, \ell_2)$.
It means that a pair of related locations store related values.

\paragraph{Function Types.}
Two  $\lambda$ terms, $\lambda x.t_1$ and $\lambda x.t_2$, are related at type $\ty[p]{{T}} \to^{\EPS} U^{r}$ with respect to world $\W$, meaning that it satisfies the following conditions:
\begin{itemize}
    \item The set of locations reachable from the two $\lambda$ terms are well-formed with respect to the world, \ie, $\dvalocss{\lambda x. t_1} \subq \DOM_1(\W)$ and $\dvalocss{\lambda x. t_2} \subq \DOM_2(\W)$.
    \item The arguments are allowed if
          \begin{itemize}
              \item the arguments $v_1$ and $v_2$ are related at type $T$ with respect $\W\extends\W'$, for all $\W'$; and
              \item the overlapping locations reachable from the functions and their arguments are permissible by the argument's qualifier $p$, \ie, $\dvalocss{\lambda x. t_1} \overlap \dvalocss{v_1} \subq \dvarslocs{\gamma_1(p)}$ and $\dvalocss{\lambda x. t_2} \overlap \dvalocss{v_2} \subq \dvarslocs{\gamma_2(p)}$.
          \end{itemize}
    \item After substitution, the two terms $t_1[x \mapsto v_1]$ and $t_2[x \mapsto v_2]$ are reduced to some values $v_1'$ and $v_2'$ with some final stores $\sigma_1'$ and $\sigma_2'$.
    \item  $\sigma_1'$ and  $\sigma_2'$ are related with respect to world  $\W\extends\W'\extends\W''$, for some $\W''$, \ie, $(\sigma_1', \sigma_2') : \W\extends\W'\extends\W''$.
    \item $v_1'$ and $v_2'$ are related at type $U$ with respect to world  $\W\extends\W'\extends\W''$.
    \item If the return value's qualifier $r$ depends on the argument (\ie, $x \in r$), then the locations reachable from $v_1'$ and $v_2'$ are subsets of those reachable both from the function and $r$, plus those reachable from the arguments;
          otherwise (\ie, $x \not\in r$), they are just subset of those reachable both from the function and $r$.
    \item If a bound variable $x$ appears in the effect $\EPS$, meaning the function body may modify the argument, then the effect will include the qualifier that may reach the value of function argument $p$; otherwise it is just $\EPS$.
\end{itemize}

\paragraph{The Binary Term Interpretation.}
Two related terms, $t_1$ and $t_2$, are defined based on the relation of their computational behaviors, \ie, returned values, reachability qualifiers and effects, which is defined by $\DGMt{\varphi}{{T}\ \EPS}$.
It means for all related stores with respect to world, $\WFRS{\sigma_1}{\sigma_2}{\W}$, if
\begin{itemize}
    \item $t_1$ is evaluated to some value $v_1$ with some final store $\sigma_1'$;
    \item $t_2$ is evaluated to some value $v_2$ with some final store $\sigma_2'$;
    \item $v_1$ and $v_2$ are  related at type $T$ with respect to world $\W \extends \W'$ for some $\W'$;
    \item $\sigma_1'$ and $\sigma_2'$ are related with respect to $\W \extends \W'$.
    \item The locations reachable from the values in the domain of pre-stores are subset of those reachable from $\gamma_1(\flt \overlap q)$ and $\gamma_2(\flt \overlap q)$ for each of the term.
    \item The effect captures what may be modified in the pre-state store.
\end{itemize}
Note that we interpret the function body (after substitution) and other terms separately, which allows us to provide more precise reasoning in the logical relations of function types.

Now, we define the binary logical relation for logical equivalence $\G[\flt] \models t_1 \equivlog t_2: \ty[q]{T} \ \EPS$ as follows:
\begin{definition}\label{def:direct:log_equiv}
    \[
        \begin{array}{l l}
            \G[\flt] \models t_1 \equivlog t_2: \ty[q]{T} \ \EPS \DEF  \forall (\gamma, \W) \in \UG{\G[\flt]}. \DVTC{\W}{\gamma_1(t_1)}{\gamma_2(t_2)} \in \DGMt{\varphi}{\ty[q]{T}\ \EPS}
        \end{array}
    \]
    where $\gamma_1(t)$ and $\gamma_2(t)$ means substitutions of the free variable in $t_1$ and $t_2$ with respect to the relational value substitution $\gamma$.
    Closing substitution over qualifiers and effects is performed in the definition of logical relations (\figref{fig:direct_binary}).
\end{definition}

\subsection{Metatheory}
This section discusses several key lemmas used in the proof of compatibility lemmas (\secref{sec:direct_cls}) and soundness of equational rules (\secref{sec:direct_equiv}).

\subsubsection{World Extension and Relational Stores}

\begin{lemma}[Relational Store Update]\label{lem:storet_update} If $\WFRS{\sigma_1}{\sigma_2}{\W}$, and $\DVTC{\W}{\ell_1}{\ell_2} \in \DGMV{\TRef ~ \Type{B}}$, and
    $\DVTC{\W}{v_1}{v_2} \in \DGMV{\Type{B}}$,
    then
    $\WFRS{\sigma_1[\ell_1 \mapsto v_1]}{\sigma_2[\ell_2 \mapsto v_2]}{\W}$.
\end{lemma}
\begin{proof} By definition of relational stores.
\end{proof}

\begin{lemma}[Relational Store Extension]\label{lem:storet_extend}
    If $\WFRS{\sigma_1}{\sigma_2}{\W}$, and $\DVTC{\W}{v_1}{v_2} \in \DGMV{\Type{B}}$,
    then \\
    $\WFRS{\sigma_1\extends (\ell_1:v_1)}{\sigma_2\extends{\ell_2:v_2}}{\W \extends (\ell_1, \ell_2, (\ell_1, \ell_2) \in f)}$, where $\ell_1 \not\in \DOM(\sigma_1)$ and $\ell_2 \not\in \DOM(\sigma_2)$.
\end{lemma}
\begin{proof} By definition of relational stores.
\end{proof}

\begin{lemma}[Logical Relation Closed Under Relational Value Substitution Extension]\label{lem:valt_extend}
    If $T$ is closed under $\G[\flt]$, and $(\W, \gamma) \in \UG{\G[\flt]}$, then
    $\DVTC{\W}{v_1}{v_2} \in \DGMV{T}$ if and only if $\DVTC{\W}{v_1}{v_2} \in \DGMVE{T}{\gamma\extends\gamma'}$, for all $\gamma'$.
\end{lemma}
\begin{proof} By induction on type $T$ and the constructs of values $v_1$  and $v_2$.
\end{proof}

\begin{lemma}[Logical Relation Closed Under World Extension]\label{lem:valt_store_extend} If $\DVTC{\W}{v_1}{v_2} \in \DGMV{T}$, then for all $\W'$, $\DVTC{\W\extends\W'}{v_1}{v_2} \in \DGMV{T}$.
\end{lemma}
\begin{proof} By induction on type $T$ and the constructs of values $v_1$  and $v_2$.
\end{proof}

\subsubsection{Well-formedness}
\begin{lemma}[Well-formed value interpretation]\label{lem:valt_wf} Let $(\W, \gamma) \in \UG{\G[\flt]}$.
    If $\UTC{\W}{v_1}{v_2} \in \DGMV{T}$, then $\dvalocss{v_1} \subq \DOM_1(\W)$ and $\dvalocss{v_2} \subq \DOM_2(\W)$.
\end{lemma}
\begin{proof} By induction on type $T$ and the constructs of value $v_1$ and $v_2$.
\end{proof}

\begin{lemma}[Well-formed Typing context interpretation] \label{lem:env_type_store_wf}    Let $(\W, \gamma) \in \UG{\G[\flt]}$, then for all $q \subq \flt$, $\dvarslocs{\gamma_1(q)} \subq \DOM_1(\W)$ and
    $\dvarslocs{\gamma_2(q)} \subq \DOM_2(\W)$.
\end{lemma}
\begin{proof} By definition of the typing context interpretation and \lemref{lem:valt_wf}.
\end{proof}

\begin{lemma}\label{lem:wf_env_type} Let $(\W, \gamma) \in \UG{\G[\flt]}$, then $\DOM(\gamma_1) = \DOM(\gamma_2) = \DOM(\G)$, and $\qsat{\DOM(\G)}$.
\end{lemma}
\begin{proof} Immediately by the definition of typing context interpretation and the definition of saturation in \figref{fig:saturation_overlap}.
\end{proof}

\subsubsection{Semantic Typing Context}

\begin{lemma}[Semantic Typing Context Tighten]\label{lem:envt_tighten} If $(\W, \gamma, ) \in \UG{\G[\flt]}$, then for all $p \subq \flt$, $(\W, \gamma) \in \UG{\G[p]}$.
\end{lemma}
\begin{proof} By the definition of typing context interpretation.
\end{proof}

\begin{lemma}[Semantic Typing Context Extension 1]\label{lem:envt_extend}
    If $(\W, \gamma) \in \UG{\G[\flt]}$,
    and $q \subq \DOM(\G)$,
    and $\DVTC{\W}{v_1}{v_2} \in \DGMV{T}$,
    and $\dvalocss{\gamma_1(\flt)} \overlap \dvalocss{v_1} \subq \dvalocss{\gamma_1(q)}$,
    and $\dvalocss{\gamma_2(\flt)} \overlap \dvalocss{v_2} \subq \dvalocss{\gamma_2(q)}$,
    then $(\W, \gamma\extends (x \mapsto (v_1, v_2))) \in \UG{(\G, x: \ty[q]{T})^{\flt,x}}$
\end{lemma}
\begin{proof}  By typing context interpretation  and \lemref{lem:valt_extend}.
\end{proof}

\begin{lemma}[Semantic Typing Context Extension 2]\label{lem:envt_extend_all}
    If $(\W, \gamma)\in \UG{\G[\flt]}$,
    and $\DVTC{\W\extends\W'}{v_1}{v_2} \in \DGMV{T}$,
    and $\dvalocss{\gamma_1(q)} \overlap \dvalocss{v_1} \subq \dvalocss{\gamma_1(p)}$,
    and $\dvalocss{\gamma_2(q)} \overlap \dvalocss{v_2} \subq \dvalocss{\gamma_2(p)}$,
    and $q \subq \flt$,
    then $(\W \extends \W', \gamma \extends (x \mapsto (v_1, v_2))) \in \UG{(\G, x: \ty[p]{T})^{q, x}}$.
\end{lemma}
\begin{proof} By typing context interpretation, \lemref{lem:valt_extend}, \lemref{lem:valt_store_extend} and \lemref{lem:envt_tighten}.
\end{proof}

\subsubsection{Reachability Qualifiers}

\begin{lemma}\label{lem:valq_omega} For all $\sigma$, $p$ and $q$, $\dvalq{\sigma}{\omega}{\dvalocss{p \overlap q}}$.
\end{lemma}
\begin{proof} Immediate by the definition in \figref{fig:direct_reachhability}.
\end{proof}

\begin{lemma}\label{lem:valq_unit}For all $\sigma$, $p$ and $q$, $\dvalq{\sigma}{\tunit}{\dvalocss{p \overlap q}}$.
\end{lemma}
\begin{proof} Immediate by the definition in \figref{fig:direct_reachhability}.
\end{proof}

\begin{lemma}\label{lem:valq_bool} For all $\sigma$, $b$, $p$ and $q$, $\dvalq{\sigma}{b}{\dvalocss{p \overlap q}}$, where $b$ is \text{true} or \text{false}.
\end{lemma}
\begin{proof} Immediate by the definition in \figref{fig:direct_reachhability}.
\end{proof}

\begin{lemma}\label{lem:valq_fresh} For all $\sigma$, $\ell$, $p$ and $q$, $\dvalq{\sigma}{\ell}{\dvalocss{p \overlap q}}$, where $\ell \not\in \DOM(\sigma)$.
\end{lemma}
\begin{proof} Immediate by the definition in \figref{fig:direct_reachhability}.
\end{proof}

\begin{lemma}\label{lem:valq_abs} If $(\G, x:\ty[p]{T})^{q,x} \ts t: \ty[r]{U}$, then
    for all $(\W, \gamma) \in \UG{(\G, x:\ty[p]{T})^{q,x}}$,
    $\dvalq{\DOM_1(\W)}{\gamma_1(\lambda x. t)}{\dvalocss{\gamma_1(q)}}$ and
    $\dvalq{\DOM_2(\W)}{\gamma_2(\lambda x. t)}{\dvalocss{\gamma_2(q)}}$.
\end{lemma}
\begin{proof} By the syntactic structure that ensures $q$ contains all the free variables in $t$,
    it is obvious that after substituting free variables with values, the conclusions hold.
\end{proof}

\subsubsection{Effects}
Here we introduce several notations to streamline the presentation.

Let $\sigma$ be a store. We write $\sigma = \sigma_1 \ast \sigma_2$ (for some $\sigma_1$ and $\sigma_2$) to denote that store $\sigma$  can be split into two disjoint parts $\sigma_1$ and $\sigma_2$.

Let $\sigma$ be a store and $L$ be a set of locations.
We write $\restrict{\sigma}{L}$ to mean localizing a partial store with respect to $L$,
meaning
$\DOM(\restrict{\sigma}{L}) = \DOM(\sigma) \overlap L \: \land \: \forall \: \ell \in \DOM(\restrict{\sigma}{L}). \restrict{\sigma}{L}(\ell) = \sigma(\ell)$.

\begin{lemma}[Read/Write Effects]\label{lem:write} If
    $\ell \in \DOM(\sigma)$,
    and
    $\dvalq{\sigma}{\ell}{\dvalocss{p \overlap q}}$,
    then
    $\DEPS{\sigma}{\dvalocss{q}}{\sigma[\ell \mapsto v]}$.
\end{lemma}
\begin{proof} By \lemref{lem:valq_fresh} and interpretation of effects.
\end{proof}

\begin{lemma}[No Effects]\label{lem:noeffects}
    $\DEPS{\sigma}{\PURE}{\sigma}$.
\end{lemma}
\begin{proof} Immediate by the definition of effects.
\end{proof}

\begin{lemma}[SubEffects]\label{lem:subeffects}
    If $\dvalocss{\EPS[1]} \subseteq \dvalocss{\EPS[2]}$,
    and $\DEPS{\sigma}{\dvalocss{\EPS[1]}}{\sigma'}$,
    then $\DEPS{\sigma}{\dvalocss{\EPS[2]}}{\sigma'}$.
\end{lemma}
\begin{proof} By the interpretation of effects.
\end{proof}

\begin{lemma}[Effects Composition]\label{lem:eff_composite}
    If $\DEPS{\sigma}{\dvalocss{\qsat{\EPS[1]}}}{\sigma'}$,
    and $\DEPS{\sigma'}{\dvalocss{\qsat{(\EPS[2] \EFFSEQ \EPS[3])}}}{\sigma''}$,
    and $\qsat{\EPS[2]}\overlap \qsat{\EPS[3]} =  \emptyset$,
    and $\dvalocss{\qsat{\EPS[2]}} \subq \DOM(\sigma)$,
    and $\dvalocss{\qsat{\EPS[3]}}  \overlap \DOM(\sigma) = \emptyset$.
    then $\DEPS{\sigma}{\dvalocss{\qsat{(\EPS[1] \EFFSEQ \EPS[2])}}}{\sigma''}$
\end{lemma}
\begin{proof} By the interpretation of effects.
\end{proof}

\begin{lemma}[Framing]\label{lem:framing} If $\DEPS{\sigma}{\dvalocss{\EPS}}{\sigma'}$,
    then $ \restricts{\sigma}{(\DOM(\sigma) - \dvalocss{\qsat{\EPS}})} = \restricts{\sigma'}{(\DOM(\sigma) - \dvalocss{\qsat{\EPS}})}$
\end{lemma}
\begin{proof} By the interpretation of observable effects: the set of locations that may be written in the reduction of $t$ must be in $\EPS$.
    Thus, the values stored in the locations $\sigma$, but are separate from $\qsat{\EPS}$ must be preserved.
\end{proof}

\begin{lemma}[Effect Separation]\label{lem:eff_sep}
    If $\sigma \mid t_1 \mredv{\ast} \sigma' \mid v_1$,
    and $\sigma \mid t_2 \mredv{\ast} \sigma'' \mid v_2$,
    and $\DEPS{\sigma}{\dvalocss{\EPS[1]}}{\sigma'}$,
    and $\DEPS{\sigma}{\dvalocss{\EPS[2]}}{\sigma''}$,
    and $\qsat{\EPS[1]} \overlap \qsat{\EPS[2]} = \qbot$,
    then
    $ \restricts{\sigma'}{(\DOM(\sigma) - \dvalocss{\qsat{\EPS[1]}} - \dvalocss{\qsat{\EPS[2]}})}  = \restricts{\sigma''}{(\DOM(\sigma) - \dvalocss{\qsat{\EPS[1]}} - \dvalocss{\qsat{\EPS[2]}})}$.
\end{lemma}
\begin{proof}
    Let
    $\sigma_1 =  \restricts{\sigma}{\dvalocss{\qsat{\EPS[1]}}}$,
    and $\sigma_2 =  \restricts{\sigma}{\dvalocss{\qsat{\EPS[2]}}}$,
    and $\sigma_3 = \restricts{\sigma}{(\DOM(\sigma) - \dvalocss{\qsat{\EPS[1]}} - \dvalocss{\qsat{\EPS[2]}})}$,
    We know that $\sigma = \sigma_1 \ast \sigma_2 \ast \sigma_3$, as $\qsat{\EPS[1]}  \overlap \qsat{\EPS[2]} = \qbot$.

    By $\DEPS{\sigma}{\EPS[1]}{\sigma'}$
    and \lemref{lem:framing}, we know $\sigma' = \sigma_1' \ast \sigma_2 \ast \sigma_3 \ast \sigma_{fr1}$, for some $\sigma_1'$, where $\sigma_{fr1} \ast \sigma$.

    By $\DEPS{\sigma}{\EPS[2]}{\sigma''}$
    and \lemref{lem:framing}, we know $\sigma'' = \sigma_1 \ast \sigma_2' \ast \sigma_3 \ast \sigma_{fr2}$, for some $\sigma_2'$,  where $\sigma_{fr2} \ast \sigma$.

    Then $ \restricts{\sigma'}{(\DOM(\sigma) - \dvalocss{\qsat{\EPS[1]}} - \dvalocss{\qsat{\EPS[2]}})}  = \sigma_3$,
    and $ \restricts{\sigma''}{(\DOM(\sigma) - \dvalocss{\qsat{\EPS[1]}} - \dvalocss{\qsat{\EPS[2]}})}  = \sigma_3$.
\end{proof}

\subsubsection{Other auxiliary lemmas}

\begin{lemma}[Qualifier intersection distributes over locations]\label{lem:overlapping} Let $(\W, \gamma) \in \UG{\G[\flt]}$, and $\WFRS{\sigma_1}{\sigma_2}{\W}$.
    For all $\sigma_1'$, $\sigma_2'$ and $\W'$, such that $\WFRS{\sigma_1'}{\sigma_2'}{\W \extends \W'}$ if
    $\dvalq{\sigma_1}{v_{f1}}{\dvalocss{\gamma_1(q_f)}}$,
    and $\dvalq{\sigma_2}{v_{f2}}{\dvalocss{\gamma_2(q_f)}}$,
    and $\dvalq{\sigma_1'}{v_1}{\dvalocss{\gamma_1(p)}}$,
    and $\dvalq{\sigma_2'}{v_2}{\dvalocss{\gamma_2(p)}}$,
    and $\dvalocss{v_{f1}} \subq \DOM(\sigma_1')$,
    and $\dvalocss{v_{f2}} \subq \DOM(\sigma_2')$,
    then
    $ (\dvalocss{v_{f1}} \, {\overlap} \, \dvalocss{v_1}) \subq \dvarslocs{\gamma_1(\qsat{p} \overlap \qsat{q_f})} $
    and $ (\dvalocss{v_{f2}} \, {\overlap} \, \dvalocss{v_2}) \subq \dvarslocs{\gamma_2(\qsat{p} \overlap \qsat{q_f})} $.
\end{lemma}
\begin{proof} By typing context interpretation, \lemref{lem:env_type_store_wf} and set theory.
\end{proof}

\begin{lemma}[Semantic Function Abstraction]\label{lem:sem_abs} Let $(\W, \gamma) \in \UG{\G[\flt]}$, $\WFRS{\sigma_1}{\sigma_2}{\W}$, and $\qsat{\DOM(\G)}$.
    For all $\W'$, if
    $\UTC{\W\extends\W'}{v_1}{v_2} \in \DGMV{T}$,
    and $\dvalocss{\gamma_1(\lambda x. t_1)} \overlap \dvalocss{v_1} \subq \dvalocss{\gamma_1(p)}$,
    and $\dvalocss{\gamma_2(\lambda x. t_2)} \overlap \dvalocss{v_2} \subq \dvalocss{\gamma_2(p)}$,
    and $(\W\extends \W', \gamma \extends x \mapsto (v_1, v_2)) \in \UG{(\G, x: \ty[p]{T})^{q,x}}$ implies
    that there exists  $\W''$, such that \\
    $\UTC{\W\extends \W'\extends \W''}{\gamma_1(t_1)[x \mapsto v_1]}{\gamma_2(t_2)[x \mapsto v_2]} \in \DGMt{q,x}{\ty[r]{U}\ \EPS}$.
    and $p \subq q$,
    then exists $\W''$, $v_1'$, $v_2'$, such that
    \begin{enumerate}
        \item $\sigma_1 \mid \gamma_1(t_1)[x \mapsto v_1] \mredv{\ast} \sigma_1' \mid v_2$
        \item $\sigma_2 \mid \gamma_2(t_2)[x \mapsto v_2] \mredv{\ast} \sigma_2' \mid v_2'$
        \item $\WFRS{\sigma_1'}{\sigma_2'}{\W \extends \W' \extends \W''}$
        \item $\UTC{\W\extends\W'\extends \W'' }{v_3}{v_4} \in \DGMV{\Type{U}}$
        \item $(x \in r \Rightarrow  \dvalq{\sigma_1'}{v_1'}{(\dvalocss{\gamma_1(r)} \overlap \dvalocss{\gamma_1(\lambda x. t_1)} \, \cup \, \dvalocss{v_1})} \, \land $ \\
              $\qquad \qquad \dvalq{\sigma_2'}{v_2'}{(\dvalocss{\gamma_2(r)} \overlap \dvalocss{\gamma_2(\lambda x. t_2)} \, \cup \, \dvalocss{v_2})})$
        \item $(x \not\in r \Rightarrow \dvalq{\sigma_1'}{v_1'}{(\dvalocss{\gamma_1(r)} \overlap \dvalocss{\gamma_1(\lambda x. t_1)})} \, \land \, \dvalq{\sigma_2'}{v_2'}{(\dvalocss{\gamma_2(r)} \overlap \dvalocss{\gamma_2(\lambda x . t_2)})})$
    \end{enumerate}

\end{lemma}
\begin{proof}
    By \lemref{lem:envt_extend_all}, $(\W\extends\W', \gamma\extends (x \mapsto (v_1, v_2)))  \in \UG{(\G, x: \ty[p]{T})^{q,x}}$. Thus, there exists  $\W''$, such that $\UTC{\W\extends \W'\extends \W''}{\gamma_1(t_1)[x \mapsto v_1]}{\gamma_2(t_2)[x \mapsto v_2]} \in \DGMt{q,x}{\ty[r]{U}\ \EPS}$, which can be used to prove (1) - (4).
    (5) and (6) can be proved by inspecting $x \in r$, \lemref{lem:valt_wf} and \lemref{lem:valq_abs}.
\end{proof}

\begin{lemma}[Semantic Application]\label{lem:sem_app} Let $(\W, \gamma) \in \UG{\G[\flt]}$.
    If \\
    $\UTC{\W\extends\W'}{\gamma_1(\lambda x. t_1)}{\gamma_2(\lambda x. t_2)} \in \DGMV{\ty[\qsat{p} \overlap \qsat{q}]{\Type{T}} \to^{\EPS} U^{r}}$,
    and $\dvalq{\sigma_1}{\gamma_1(\lambda x. t_1)}{\dvalocss{\gamma_1(q)}}$,
    and $\dvalq{\sigma_2}{\gamma_2(\lambda x. t_2)}{\dvalocss{\gamma_2(q)}}$,
    and $\UTC{\W \extends \W' \extends \W''} {v_1}{v_2} \in \DGMV{T}$,
    and $\dvalq{\DOM_1(\W \extends \W')}{v_1}{\dvalocss{\gamma_1(p)}}$,
    and $\dvalq{\DOM_2(\W \extends \W')}{v_2}{\dvalocss{\gamma_2(p)}}$,
    and $r \subq \flt, x$,
    and $\EPS \subq q,x$,
    and $\WFRS{\sigma_1}{\sigma_2}{\W\extends\W'\extends\W''}$
    then there exists $v_2$, $v_2'$, $\sigma_1'$, $\sigma_2'$, $\W'''$, such that
    \begin{enumerate}
        \item $\sigma_1 \mid \gamma_1(t_1)[x \mapsto v_1] \mredv{\ast}  \sigma_1' \mid v_2$;
        \item $\sigma_2 \mid \gamma_1(t_2)[x \mapsto v_2] \mredv{\ast}  \sigma_2' \mid v_2'$;
        \item $\WFRS{\sigma_1'}{\sigma_2'}{\W \extends \W' \extends \W'' \extends \W'''}$;
        \item $\UTC{\W \extends \W' \extends \W''\extends \W'''}{v_2}{v_2'} \in \DGMV{U}$;
        \item $x\in r \Rightarrow  \dvalq{\sigma_1}{v_1'}{(\dvalocss{\gamma_1(r)} \overlap \dvalocss{\gamma_1(\lambda x. t_1)}  \, \cup \, \dvalocss{v_1})}  \, \land$\\
              $\;\;\; \dvalq{\sigma_2}{v_2'}{(\dvalocss{\gamma_2(r)} \overlap \dvalocss{\gamma_2(\lambda x. t_2)} \, \cup \, \dvalocss{v_2})}$
        \item $x\not\in r \Rightarrow \dvalq{\sigma_1}{v_1'}{(\dvalocss{\gamma_1(r)} \overlap \dvalocss{\gamma_1(\lambda x. t_1)})} \, \land  \, \dvalq{\sigma_2}{v_2'}{(\dvalocss{\gamma_2(r)} \overlap  \dvalocss{\gamma_2(\lambda x . t_2)})} $
    \end{enumerate}
\end{lemma}
\begin{proof} By \lemref{lem:valt_wf}, we know the following:
    \begin{itemize}
        \item $\dvalocss{\gamma_1(\lambda x. t_1)} \subq \DOM_1(\W \extends \W')$;
        \item $\dvalocss{\gamma_2(\lambda x. t_2)} \subq \DOM_2(\W \extends \W')$;
        \item $\dvalocss{v_1} \subq \DOM_1(\W \extends \W' \extends \W'')$;
        \item $\dvalocss{v_2} \subq \DOM_2(\W \extends \W' \extends \W'')$;
    \end{itemize}

    Then (1) - (4) can be proved by the assumption  $\UTC{\W\extends\W'}{\gamma_1(\lambda x. t_1)}{\gamma_2(\lambda x. t_2)} \in \DGMV{\ty[p]{\Type{T}} \to^{\EPS} U^{r}}$, and $\UTC{\W \extends \W' \extends \W''} {v_1}{v_2} \in \DGMV{T}$, and \lemref{lem:overlapping}.
    (5) - (6) can be proved by inspecting $x \in r$.
\end{proof}

\subsection{Compatibility Lemmas}
\label{sec:direct_cls}
The following compatibility lemmas show that the logical relations is \emph{compatible} with all the constructs of the language~\cite{10.5555/1076265}.

\begin{lemma}[Compatibility: $\Type{Alloc}$]  $\G[\flt] \models  \omega \equivlog \omega: \ty[\qbot]{Alloc}\ \PURE$
\end{lemma}
\begin{proof} By the typing context interpretation, value interpretation in \figref{fig:direct_binary} and \lemref{lem:valq_omega}.
\end{proof}

\begin{lemma}[Compatibility: $\TUnit$]  $\G[\flt] \models  \tunit \equivlog \tunit: \ty[\qbot]{\TUnit}\ \PURE$
\end{lemma}
\begin{proof} By the typing context interpretation, value interpretation in \figref{fig:direct_binary} and \lemref{lem:valq_unit}.
\end{proof}

\begin{lemma}[Compatibility: $\Type{Bool}$]  $\G[\flt] \models  \text{true} \equivlog \text{true} : \ty[\qbot]{Bool}\ \PURE$
\end{lemma}
\begin{proof} By the typing context interpretation, value interpretation in \figref{fig:direct_binary}  and \lemref{lem:valq_bool}.
\end{proof}

\begin{lemma}[Compatibility: $\Type{Bool}$]  $\G[\flt] \models \text{false} \equivlog \text{false} : \ty[\emptyset]{Bool}\ \PURE$
\end{lemma}
\begin{proof} By the typing context interpretation, value interpretation in \figref{fig:direct_binary}  and \lemref{lem:valq_bool}.
\end{proof}

\begin{lemma}[Compatibility: Variables]
    If  $x: \ty[q]{T} \in \G$ and  $x \subq \flt$, then
    $\G[\flt] \models  x \equivlog x : \ty[x]{T}\  \PURE$
\end{lemma}
\begin{proof} Immediate by the typing context interpretation in \figref{fig:direct_binary}.
\end{proof}

\begin{lemma}[Compatibility: $\lambda$] If $ (\G\ ,\ x: \ty[p]{T})^{q,x}\ \models t_1 \equivlog t_2 : \ty[r]{U}\ \EPS$, $q \subq \flt$,
    then
    $\G[\flt] \models  \lambda x.t_1 \equivlog \lambda x.t_2 : \ty[q]{\left(x: \ty[p]{T} \to^{\EPS} \ty[r]{U}\right)}\; \PURE$.
\end{lemma}
\begin{proof}  Let $(\W, \gamma ) \in \UG{\G}$ and $\WFRS{\sigma_1}{\sigma_2}{\W}$.

    By definition of term interpretation, we need to show there exists $\W'$, $\sigma'$, $v_1$ and $v_2$ such that:
    \begin{enumerate}
        \item $\sigma_1 \mid \gamma_1(\lambda x. t_1) \mredv{\ast} \sigma_1' \mid v_1$
        \item $\sigma_2 \mid \gamma_2(\lambda x. t_2) \mredv{\ast} \sigma_2' \mid v_2$
        \item $\WFRS{\sigma'_1}{\sigma'_2}{\W\extends\W'}$
        \item $\UTC{\W \extends \W'}{v_1}{v_2} \in \DGMV{{\left(x: \ty[p]{T} \to^{\EPS} \ty[r]{U}\right)}}$
        \item $\dvalq{\sigma_1}{v_1}{\dvalocss{\gamma_1(\flt \overlap q)}}$
        \item $\dvalq{\sigma_2}{v_2}{\dvalocss{\gamma_2(\flt \overlap q)}}$
        \item $\DEPS{\sigma_1}{\PURE}{\sigma_1'}$
        \item $\DEPS{\sigma_2}{\PURE}{\sigma_2'}$
    \end{enumerate}

    By reduction semantics, we pick $\W = \emptyset$, $v_1 = \lambda x. t_1$, $v_2 = \lambda x. t_2$, $\sigma_1'= \sigma_1$ and $\sigma_2' = \sigma_2$.
    Thus, (1)- (3) are discharged.
    (4) can be proved by \lemref{lem:wf_env_type} and \lemref{lem:sem_abs}.
    (5) and (6) can be proved by \lemref{lem:valq_abs}.
    (7) and (8) can be proved by \lemref{lem:noeffects}.
\end{proof}

\begin{lemma}[Compatibility : Allocation]
    If $\G[\flt] \models t_1 \equivlog t_2 : \ty[q]{\Type{Alloc}}\ \EPS[1]$, and  $\strut\G[\flt] \ts t_3 \equivlog t_4 : \ty[\qbot]{B}\ \EPS[2]$, then
    $\G[\flt] \models  \tref_{t_1}~t_3 \equivlog \tref_{t_2}~t_4\  : \ty[\qbot]{(\TRef\ B)}\  \EPS[1]\EFFSEQ\EPS[2]\EFFSEQ\FX{q}$.
\end{lemma}
\begin{proof} Let $(\W, \gamma ) \in \UG{\G}$ and $\WFRS{\sigma_1}{\sigma_2}{\W}$.
    By the first assumption, we know that there exists $\sigma_1'$, $\sigma_2'$, $\W'$, $v_1$ and $v_2$, such that
    \begin{itemize}
        \item  $\gamma_1(t_1) \mid \sigma_1 \mredv{\ast} v_1 \mid \sigma_1'$
        \item  $\gamma_2(t_2) \mid \sigma_2 \mredv{\ast} v_2 \mid \sigma_2'$
        \item  $\WFRS{\sigma_1'}{\sigma_2'}{\W\extends \W'}$
        \item  $\DVTC{\W \extends \W'}{v_1}{v_2} \in \DGMV{\Type{Alloc}}$
        \item  $\dvalq{\sigma_1}{v_1}{\dvalocss{\gamma_1(\flt \overlap q)}}$
        \item  $\dvalq{\sigma_2}{v_2}{\dvalocss{\gamma_2(\flt \overlap q)}}$
        \item  $\DEPS{\sigma_1}{\dvalocss{\gamma_1(\EPS[1])}}{\sigma_1'} $
        \item  $\DEPS{\sigma_2}{\dvalocss{\gamma_2(\EPS[1])}}{\sigma_2'} $
    \end{itemize}

    By reduction semantics, we know $v_1 = v_2 = \omega$.

    By the second assumption, we know that there exists $\sigma_1''$, $\sigma_2''$, $\W''$, $v_3$ and $v_4$, such that
    \begin{itemize}
        \item  $\gamma_1(t_3) \mid \sigma_1' \mredv{\ast} v_3 \mid \sigma_1''$
        \item  $\gamma_2(t_4) \mid \sigma_2' \mredv{\ast} v_4 \mid \sigma_2''$
        \item  $\WFRS{\sigma_1''}{\sigma_2''}{\W\extends \W'\extends \W''}$
        \item  $\DVTC{\W \extends\W'\extends\W''}{v_3}{v_4} \in \DGMV{\Type{B}}$
        \item  $\dvalq{\sigma_1'}{v_3}{\dvalocss{\gamma_1(\flt \overlap \qbot)}}$
        \item  $\dvalq{\sigma_2'}{v_4}{\dvalocss{\gamma_2(\flt \overlap \qbot)}}$
        \item  $\DEPS{\sigma_1'}{\dvalocss{\gamma_1(\EPS[2])}}{\sigma_1''} $
        \item  $\DEPS{\sigma_2'}{\dvalocss{\gamma_2(\EPS[2])}}{\sigma_2''} $
    \end{itemize}
    By reduction semantics, we know
    \begin{itemize}
        \item $\tref_{\omega} ~ v_3 \mid \sigma_1'' \mredv{1}  \ell_1 \mid \sigma_1''\extends(\ell_1 \mapsto v_3)$, where $\ell_1 \not\in \DOM(\sigma_1'')$
        \item $\tref_{\omega} ~ v_4 \mid \sigma_2'' \mredv{1}  \ell_2 \mid \sigma_2''\extends(\ell_2 \mapsto v_4)$, where $\ell_2 \not\in \DOM(\sigma_2'')$
    \end{itemize}
    By \lemref{lem:storet_extend}, we know $\WFRS{\sigma_1''\extends(\ell_1 \mapsto v_3)}{\sigma_2''\extends(\ell_2 \mapsto v_4)}{\W\extends\W'\extends\W'' \extends ((\ell_1 \mapsto v_3), (\ell_2 \mapsto v_4), \{(\ell_1, \ell_2)\})}$.
    The rest of the proof can be done by the definition of value interpretation, \lemref{lem:eff_composite} and \lemref{lem:valq_fresh}.
\end{proof}

\begin{lemma}[Compatibility: Dereference ($!$)] If $\G[\flt] \models t_1 \equivlog t_2 : \ty[q]{(\TRef ~B)} \ \EPS$, then $\G[\flt] \models !t_1 \equivlog !t_2 : \ty[\qbot]{B} \ \EPS\EFFSEQ\FX{\bm{q}}$.
\end{lemma}
\begin{proof} Let $(\W, \gamma) \in \UG{\G[\flt]}$ and $\WFRS{\sigma_1}{\sigma_2}{\W}$.
    By the assumption, $\DVTC{\W}{\gamma_1(t_1)}{\gamma_2(t_2)} \in \DGMt{\varphi}{\ty[q]{\TRef ~ B}\ \EPS}$,
    and reduction semantics, we know there exists $\sigma_1'$, $\sigma_2'$, $\ell_1$ and $\ell_2$ such that
    \begin{itemize}
        \item $\sigma_1 \mid \gamma_1(t_1) \mredv{\ast} \sigma_1' \mid \ell_1$
        \item $\sigma_2 \mid \gamma_2(t_2) \mredv{\ast} \sigma_2' \mid \ell_2$
        \item $\WFRS{\sigma_1'}{\sigma_2'}{\W\extends\W'}$
        \item $\DVTC{\W\extends\W'}{\ell_1}{\ell_2} \in \DGMV{\TRef \ B}$
        \item $\dvalq{\sigma_1}{\ell_1}{\dvalocss{\gamma_1(\flt \overlap q)}}$
        \item $\dvalq{\sigma_2}{\ell_2}{\dvalocss{\gamma_2(\flt \overlap q)}}$
        \item  $\DEPS{\sigma_1}{\dvalocss{\gamma_1(\EPS)}}{\sigma_1'} $
        \item  $\DEPS{\sigma_2}{\dvalocss{\gamma_2(\EPS)}}{\sigma_2'} $
    \end{itemize}
    We can finish the proof by reduction semantics, value interpretation, \lemref{lem:valq_bool}, \lemref{lem:subeffects},
    where we pick $\sigma_1''$ to be $\sigma_1'$, $\sigma_2''$ to be $\sigma'$, and $\W''$ to be $\emptyset$.

\end{proof}

\begin{lemma}[Compatibility: Assignments ($:=$)] \label{lem:assignment}
    If $\G[\flt] \models t_1 \equivlog t_2 : \ty[q]{(\TRef ~\Type{B})} \ \EPS[1]$,
    $\G[\flt] \models t_3 \equivlog t_4 : \ty[\qbot]{\Type{B}} \ \EPS[2]$,
    then $\G[\flt] \models t_1 := t_3 \equivlog t_2 := t_4 : \ty[\qbot]{\TUnit} \ \EPS[1]\EFFSEQ\EPS[2] \EFFSEQ\FX{\bm{q}}$.
\end{lemma}
\begin{proof} Let $(\W, \gamma) \in \UG{\G[\flt]}$ and $\WFRS{\sigma_1}{\sigma_2}{\W}$.
    By the first assumption, we know that there exists $\sigma_1'$, $\sigma_2'$, $\W'$, $v_1$ and $v_2$ such that
    \begin{itemize}
        \item  $\gamma_1(t_1) \mid \sigma_1 \mredv{\ast} v_1 \mid \sigma_1'$
        \item  $\gamma_2(t_2) \mid \sigma_2 \mredv{\ast} v_2 \mid \sigma_2'$
        \item  $\WFRS{\sigma_1'}{\sigma_2'}{\W\extends \W'}$
        \item  $\DVTC{\W \extends \W'}{v_1}{v_2} \in \DGMV{\TRef ~ \Type{B}}$
        \item  $\dvalq{\sigma_1}{v_1}{\dvalocss{\gamma_1(\flt \overlap q)}}$
        \item  $\dvalq{\sigma_2}{v_2}{\dvalocss{\gamma_2(\flt \overlap q)}}$
        \item  $\DEPS{\sigma_1}{\dvalocss{\gamma_1(\EPS[1])}}{\sigma_1'} $
        \item  $\DEPS{\sigma_2}{\dvalocss{\gamma_2(\EPS[1])}}{\sigma_2'} $
    \end{itemize}
    By the second assumption, we know that there exists $\sigma_1''$, $\sigma_2''$, $\W''$, $v_3$ and $v_4$, such that
    \begin{itemize}
        \item  $\gamma_1(t_3) \mid \sigma_1' \mredv{\ast} v_3 \mid \sigma_1''$
        \item  $\gamma_2(t_4) \mid \sigma_2' \mredv{\ast} v_4 \mid \sigma_2''$
        \item  $\WFRS{\sigma_1''}{\sigma_2''}{\W\extends \W'\extends \W''}$
        \item  $\DVTC{\W \extends \W'\extends \W''}{v_3}{v_4} \in \DGMV{\Type{B}}$
        \item  $\dvalq{\sigma_1'}{v_3}{\dvalocss{\gamma_1(\flt \overlap \qbot)}}$
        \item  $\dvalq{\sigma_2'}{v_4}{\dvalocss{\gamma_2(\flt \overlap \qbot)}}$
        \item  $\DEPS{\sigma_1'}{\dvalocss{\gamma_1(\EPS[2])}}{\sigma_1''} $
        \item  $\DEPS{\sigma_2'}{\dvalocss{\gamma_2(\EPS[2])}}{\sigma_2''} $
    \end{itemize}
    Then the proof can be done by the reduction semantics, \lemref{lem:storet_update}, value interpretation, \lemref{lem:valq_unit}, \lemref{lem:write} and \lemref{lem:eff_composite}.

\end{proof}

\begin{lemma}[Compatibility: Applications ($\beta$)]. If $\G[\flt] \models t_1 \equivlog t_2: \ty[q]{\left(x{\,:\,}\ty[\qsat{p} \overlap \qsat{q}]{T} \to^{\EPS[3]} \ty[r]{U}\right)}\ \EPS[2]$, and
    $\G[\flt]\models t_3 \equivlog t_4 : \ty[p]{T}\ \EPS[1]$, and $x\notin\FV(U)$,
    and $r \subq \flt,x$, and
    and $\EPS[3] \subq \flt,x$, and  $\theta = [p/x]$,
    then
    $\G[\flt] \models t_1~t_3 \equivlog t_2~t_4 : (\ty[r]{U}\ \EPS[1]\EFFSEQ\EPS[2]\EFFSEQ\FX{\EPS[3]})\theta$.
\end{lemma}
\begin{proof} The proof is done by the definition of term interpretation,  \lemref{lem:sem_app}  and \lemref{lem:eff_composite}.
\end{proof}

\begin{lemma}[Compatibility: Let]
    If $\G[\flt] \models t_1  \equivlog t_2: \ty[p]{S} \ \EPS[1]$,
    and  $(\G\, ,\, x: \ty[\qsat{p}\cap\qsat{\flt}]{S})^{\flt,x}\models t_3 \equivlog t_4 : \ty[q]{T}\ \EPS[2]$,
    and  $\theta = [p/x]$
    and  $\quad x \notin\FV{(T)}$,
    then $\G[\flt] \models \tlet~{x = t_1}~\tin~t_3 \equivlog \tlet~{x = t_2}~\tin~t_4 : (\ty[q]{T}\ \EPS[1]\EFFSEQ\FX{\EPS[2]})\theta$
\end{lemma}
\begin{proof} Since the \textsc{t-let}  is a combination of rules \textsc{t-abs}, \textsc{t-app} and weakening, the proof is analogous.
\end{proof}

\begin{lemma}[Compatibility: Subtyping]
    If $\G[\flt]  \models t_1 \equivlog t_2 : \ty[p]{S}\ \EPS[1]$
    and $\G\ts\ty[p]{S}\ \EPS[1] <: \ty[q]{T}\ \EPS[2]$
    and $q,\EPS[2]\subq\flt$,
    then $\G[\flt] \models t_1 \equivlog t_2 : \ty[q]{T}\ \EPS[2]$.
\end{lemma}
\begin{proof} By induction on the subtyping derivation.
\end{proof}

\subsection{The Fundamental Theorem and Soundness}
\label{sec:direct_soundness}

\begin{theorem}[Fundamental Property]\label{them:direct_fp} If $\G[\flt] \ts t: \ty[q]{T} \ \EPS $, then $\G[\flt] \models t \equivlog t : \ty[q]{T}\ \EPS$.
\end{theorem}
\begin{proof} By induction on the derivation of $\strut\G[\flt] \ts t: \ty[q]{T} \ \EPS$. Each case follows from the corresponding compatibility lemma.
\end{proof}

\begin{lemma}[Congruency of Binary Logical Relations]\label{lem:direct_congruence} The binary logical relation is closed under well-typed program contexts,
    \ie, if $\G[\flt] \models t_1 \equivlog t_2: \ty[p]{T} \ \EPS$,
    and $C:(\G[\flt]; \ty[p]{T} \ \EPS) \carrow (\GP[\flt']; \ty[p']{T'} \ \EPS')$, then $\GP[\flt'] \models C[t_1] \equivlog C[t_2]: \ty[p']{T'} \ \EPS'$.
\end{lemma}
\begin{proof} By induction on the derivation of context $C$. Each case follows from the corresponding compatibility lemma and may use the fundamental theorem (\thmref{them:direct_fp}) if necessary.
\end{proof}

\begin{lemma}[Adequacy of the binary logical relations]\label{lem:direct_adequacy}
    The binary logical relation preserves termination, \ie, if $\emptyset \models t_1 \equivlog t_2: \ty[\qbot]{T} \ \PURE$,
    then $\exists \ \sigma, \sigma', v. \ \emptyset \mid t_1 \mredv{\ast} \sigma \mid v \wedge \emptyset \mid t_2 \mredv{\ast} \sigma' \mid v$.
\end{lemma}
\begin{proof} We know $(\emptyset, \emptyset) \in \UG{\emptyset}$ by the interpretation of typing context.
    Then we can prove the result by the binary term interpretation (\figref{fig:direct_binary}).
\end{proof}

\begin{theorem}[Soundness of Binary Logical Relations]\label{thm:direct_lr_soundness} The binary logical relation is sound w.r.t. contextually equivalence, \ie,
    if $\G[\flt] \ts t_1: \ty[p]{T}\ \EPS$ and $\G[\flt] \ts t_2: \ty[p]{T}\ \EPS$, then
    $\G[\flt] \models t_1 \equivlog t_2: \ty[p]{T}\ \EPS$ implies $\G[\flt] \models t_1 \equiva t_2: \ty[p]{T}\ \EPS$.
\end{theorem}
\begin{proof} By the refined definition of contextual equivalence, to prove the result, we are given a well-typed context $C: (\G[\flt]; \ty[p]{T} \ \EPS) \carrow (\emptyset; \ty[\qbot]{B} \ \PURE)$,
    and we need to show $\exists \ \sigma, \sigma', v. \ \emptyset \mid C[t_1] \mredv{\ast}  \sigma \mid v  \wedge \emptyset \mid C[t_2] \mredv{\ast} \sigma' \mid v$. By the assumption, and the congruency lemma (\lemref{lem:direct_congruence}), we have $\emptyset \models C[t_1] \equivlog C[t_2]: \ty[\qbot]{B} \ \PURE$, which
    leads to $\exists \ \sigma, \sigma', v. \ \emptyset \mid C[t_1] \mredv{\ast}  \sigma \mid v  \wedge \emptyset \mid C[t_2] \mredv{\ast} \sigma' \mid v$ by the adequacy lemma (\lemref{lem:direct_adequacy}).
\end{proof}

\subsection{Equational Rules}
\label{sec:direct_equiv}
\begin{figure*}[t]
    \begin{mdframed}
        \begin{mathpar}
            \inferrule*[left=dce]
            {
            \G[\flt]\ts t_1: \ty[q_1]{T_1}\ \EPS[1] \\
            \G[\flt]\ts t_2: \ty[q_2]{T_2}\ \EPS[2] \\
            t_1 \text{ terminates}  \\
            \EPS[1] = \PURE \ \text{or} \ \omega
            }
            {
            \G[\flt] \models \tlet~{x = t_1}~\tin~t_2 \equivlog t_2 : \ty[q_2]{T_2} \EPS[2]
            }

            \inferrule*[left=comm]
            {\G[\flt]\ts  t_1: \ty[q_1]{T_1}\ \EPS[1] \\
            \G[\flt]\ts t_2: \ty[q_2]{T_2}\ \EPS[2] \\
            \csxs[\flt, x, y]{\G\ ,\ x: \ty[\qsat{q_1} \overlap \qsat{\flt}]{T_1}, \ y: \ty[\qsat{q_2}\overlap \qsat{(\flt, x)}]{T_2}} \ts t: \ty[q]{T} \ \EPS \\
            \qsat{\EPS[1]} \overlap \qsat{\EPS[2]} = \emptyset \\
            x \not\in \FV(T) \\  y \not\in \FV(T) \\
            \theta = [q_2/y][q_1/x]
            }
            {\G[\flt] \models \tlet~x = t_1~\tin~ \tlet~ y = t_2 ~ \tin ~ t \equivlog \tlet~ y = t_2\ \tin \ \tlet \ x = t_1 ~ \tin ~ t : (\ty[q]{T} \ \EPS[1] \EFFSEQ \EPS[2] \EFFSEQ \EPS)\theta}

            \inferrule*[left=$\lambda$-hoist]
            { \G[\flt]\ts  t_1: \ty[q_1]{T_1}\ \PURE \\
            \csxs[q, x, y]{\G\ ,\ x: \ty[p]{T}, y: \ty[\qsat{q_1}\overlap \qsat{(\flt, x)}]{T_1}}  \ts \ t: \ty[r]{U} \ \EPS \\ \theta =[q_1/y]  \\ x \not\in \FV(U) \\ y \not\in \FV(U)
            }
            {
            \G[\flt] \models (\lambda x: \ty[p]{T}. \ \tlet \ y = t_1 \ \tin \ t) \equivlog (\tlet \ y = t_1 \ \tin \ \lambda x: \ty[p]{T}. t) : (x: \ty[p]{T} \to^{\EPS} \ty[r]{U}\theta)^{q} \ \PURE
            }

            \inferrule*[left=$\beta$-inlining]
            {
            \csxs[q,x]{\G, x: \ty[p]{T}}   \ts t_1 : \ty[r]{U} \EPS \\
            \G[\flt]\ts t_2 : \ty[p]{T}\ \PURE \\ x\notin\FV(U) \\
            \theta = [p/x]
            }
            {
            \G[\flt] \models (\lambda x: \ty[p]{T}.  t_1)(t_2) \equivlog t_1[t_2/x] :  (\ty[r]{U}\ \EPS)\theta
            }

            \inferrule*[left = e-cse]
            {
            \G[\flt] \ts t_1: \ty[q_1]{T_1} \ \EPS[1] \\
            \csxs[\flt, x, y]{\G\ ,\ x: \ty[\qsat{q_1} \overlap \qsat{\flt}]{T_1}, y: \ty[\qsat{q_1}\overlap \qsat{(\flt, x)}]{T_1}} \ts \ t: \ty[q]{T} \ \EPS \\
            \omega \notin \EPS[1] \\ \theta = [x/y]
            }
            {
            \G[\flt] \models (\tlet\ x = t_1 \ \tin \ \tlet \ y = t_1 \ \tin \ t) \equivlog (\tlet \ x = t_1 \ \tin \ t\ \theta) : (\ty[q]{T} \ \EPS)\theta \EFFSEQ\FX{\EPS[1]} \\
            }

        \end{mathpar}
        \caption{Equational rules for the \maybelang{}-calculus.}
        \label{fig:direct_equiv}
    \end{mdframed}
\end{figure*} 
\figref{fig:direct_equiv} shows equational rules for \maybelang{} with effects specifying logically equivalent terms.
Rule \rulename{dce} permits the removal of a terminating term, $t_1$, whose computation result is an unused value, provided that the effect of that computation only allows allocation.
Removing a term with effects may not be sound, as the effects could be observed by the following computations.
Rule \rulename{comm} permits re-ordering of two terms if their effects are separate, which entails disjoint sets of store locations (\colref{coro:preservation_separation}).
Rule \rulename{$\lambda$-hoist} permits a pure computation to be moved out of the abstraction boundary.
Rule \rulename{$\beta$-inlining} permits replacing a function call site $t_2$ with the body of the called function, provided that $t_2$ is pure.
Rule \rulename{e-cse} permits removing a  duplicated computation, provided that no fresh allocations occur during the reduction of the term.
The rest of this section shows the proofs of those equational rules by using our logical relations.

\begin{lemma}[Dead Code Elimination]\label{lem:dce} If
    $\G[\flt]\ts t_1: \ty[q_1]{T_1}\ \EPS[1]$,
    and $\G[\flt]\ts t_2: \ty[q_2]{T_2}\ \EPS[2]$,
    and $\EPS[1] = \PURE$ or $\omega$,
    and $t_1$ terminates,
    then
    $\G[\flt]\ts \tlet~{x = t_1}~\tin~t_2 \equivlog t_2 : \ty[q_2]{T_2} \EPS[2]$.
\end{lemma}
\begin{proof}
    By the fundamental property (\thmref{them:direct_fp}), we know on the first assumption, we know
    $\G[\flt] \models t_1 \equivlog t_1: \ty[q_1]{T_1}\ \EPS[1]$.

    Let $(\W, \gamma) \in \UG{\G[\flt]}$ and $\WFRS{\sigma_1}{\sigma_2}{\W}$.
    By the definition of binary logical relations and binary term interpretation, we know there exists $\sigma_{11}$, $\sigma_{12}$, $\W_1$, $v_{11}$ and $v_{12}$, such that
    \begin{itemize}
        \item $\sigma_1 \mid \gamma_1(t_1) \mredv{\ast} \sigma_{11} \mid v_{11}$
        \item $\sigma_2 \mid \gamma_2(t_1) \mredv{\ast} \sigma_{12} \mid v_{12}$
        \item $\DVTC{\W\extends\W_1}{v_{11}}{v_{12}} \in \DGMV{\Type{T_1}}$
        \item $\WFRS{\sigma_{11}}{\sigma_{12}}{\W \extends \W_1}$
        \item $\dvalq{\sigma_1}{v_{11}}{(\dvalocss{\gamma_1(\flt \overlap q_1)})}$
        \item $\dvalq{\sigma_2}{v_{12}}{(\dvalocss{\gamma_2(\flt \overlap q_1)})}$
        \item $\DEPS{\sigma_1}{\dvalocss{\gamma_1(\EPS[1])}}{\sigma_{11}}$
        \item $\DEPS{\sigma_2}{\dvalocss{\gamma_2(\EPS[1])}}{\sigma_{12}}$
    \end{itemize}

    By $\EPS[1] = \PURE$ or $\omega$, we know
    $\sigma_{11} = \sigma_1 \ast \sigma_{fr1}$, and $\sigma_{12} = \sigma_2 \ast \sigma_{fr2}$.

    By the fundamental property (\thmref{them:direct_fp}) again, we know
    $\G[\flt] \models t_2 \equivlog t_2: \ty[q_2]{T_2}\ \EPS[2]$.

    By the binary term interpretation, we know
    there exists $\sigma_{21}$, $\sigma_{22}$, $\W_2$, $v_{21}$ and $v_{22}$, such that

    \begin{itemize}
        \item $\sigma_1 \mid \gamma_1(t_2) \mredv{\ast} \sigma_{21} \mid v_{21}$
        \item $\sigma_2 \mid \gamma_2(t_2) \mredv{\ast} \sigma_{22} \mid v_{22}$
        \item $\DVTC{\W\extends\W_2}{v_{21}}{v_{22}} \in \DGMV{\Type{T_2}}$
        \item $\WFRS{\sigma_{21}}{\sigma_{22}}{\W \extends \W_2}$
        \item $\dvalq{\sigma_1}{v_{21}}{(\dvalocss{\gamma_1(\flt \overlap q_2)})}$
        \item $\dvalq{\sigma_2}{v_{22}}{(\dvalocss{\gamma_2(\flt \overlap q_2)})}$
        \item $\DEPS{\sigma_1}{\dvalocss{\gamma_1(\EPS[2])}}{\sigma_{11}}$
        \item $\DEPS{\sigma_2}{\dvalocss{\gamma_2(\EPS[2])}}{\sigma_{12}}$
    \end{itemize}

    From the left, by the reduction semantics, we have:
    \begin{itemize}
        \item $\sigma_1 \mid \gamma_1(t_1) \mredv{\ast} \sigma_1 \ast \sigma_{fr1} \mid v_{11}$
        \item $\sigma_1 \ast \sigma_{fr1} \mid \gamma_1(t_2) \mredv{\ast} \sigma_{21} \mid v_{21}'$
    \end{itemize}

    By the deterministic of reduction semantics, we know $v_{21} = v_{21}'$.

    Form the right, by the reduction semantics, we know $\sigma_2 \mid \gamma_2(t_2) \mredv{\ast} \sigma_{22} \mid v_{22}$.

    By the fact we have above, the proof is done.

\end{proof}

\begin{lemma}[comm]\label{lem:comm} If
    $\G[\flt]\ts  t_1: \ty[q_1]{T_1}\ \EPS[1]$,
    and $\G[\flt]\ts t_2: \ty[q_2]{T_2}\ \EPS[2]$,
    and $\csxs[\flt, x, y]{\G\ ,\ x: \ty[\qsat{q_1} \overlap \qsat{\flt}]{T_1}, \ y: \ty[\qsat{q_2}\overlap \qsat{(\flt, x)}]{T_2}}  \ts t: \ty[q]{T} \ \EPS$,
    and $\qsat{\EPS[1]} \overlap \qsat{\EPS[2]} = \qbot$,
    and $x \not\in \FV(T)$,
    and $y \not\in \FV(T)$,
    and $\theta = [q_2/y][q_1/x]$,
    then
    $\G[\flt] \models \tlet~x = t_1~\tin~ \tlet~ y = t_2 ~ \tin ~ t \equivlog \tlet~ y = t_2\ \tin \ \tlet \ x = t_1 ~ \tin ~ t : (\ty[q]{T} \ \EPS[1] \EFFSEQ \EPS[2] \EFFSEQ \EPS)\theta$.
\end{lemma}
\begin{proof}
    By the fundamental property (\thmref{them:direct_fp}), we know
    $\G[\flt] \models t_1 \equivlog t_1: \ty[q_1]{T_1}\ \EPS[1]$.

    Let $(\W, \gamma) \in \UG{\G[\flt]}$ and $\WFRS{\sigma_1}{\sigma_2}{\W}$.
    By the definition of binary logical relations and binary term interpretation, we know
    there exists $\sigma_{11}$, $\sigma_{12}$, $\W_1$, $v_{11}$ and $v_{12}$, such that
    \begin{itemize}
        \item $\sigma_1 \mid \gamma_1(t_1) \mredv{\ast} \sigma_{11} \mid v_{11}$
        \item $\sigma_2 \mid \gamma_2(t_1) \mredv{\ast} \sigma_{12} \mid v_{12}$
        \item $\DVTC{\W\extends\W_1}{v_{11}}{v_{12}} \in \DGMV{\Type{T_1}}$
        \item $\WFRS{\sigma_{11}}{\sigma_{12}}{\W \extends \W_1}$
        \item $\dvalq{\sigma_1}{v_{11}}{(\dvalocss{\gamma_1(\flt \overlap q_1)})}$
        \item $\dvalq{\sigma_2}{v_{12}}{(\dvalocss{\gamma_2(\flt \overlap q_1)})}$
        \item $\DEPS{\sigma_1}{\dvalocss{\gamma_1(\EPS[1])}}{\sigma_{11}}$
        \item $\DEPS{\sigma_2}{\dvalocss{\gamma_2(\EPS[1])}}{\sigma_{12}}$
    \end{itemize}

    By the fundamental property (\thmref{them:direct_fp}) again, we know
    $\G[\flt] \models t_2 \equivlog t_2: \ty[q_2]{T_2}\ \EPS[2]$.
    By the binary term interpretation, we know
    there exists $\sigma_{21}$, $\sigma_{22}$, $\W_2$, $v_{21}$ and $v_{22}$, such that

    \begin{itemize}
        \item $\sigma_1 \mid \gamma_1(t_2) \mredv{\ast} \sigma_{21} \mid v_{21}$
        \item $\sigma_2 \mid \gamma_2(t_2) \mredv{\ast} \sigma_{22} \mid v_{22}$
        \item $\DVTC{\W\extends\W_2}{v_{21}}{v_{22}} \in \DGMV{\Type{T_2}}$
        \item $\WFRS{\sigma_{21}}{\sigma_{22}}{\W \extends \W_2}$
        \item $\dvalq{\sigma_1}{v_{21}}{(\dvalocss{\gamma_1(\flt \overlap q_2)})}$
        \item $\dvalq{\sigma_2}{v_{22}}{(\dvalocss{\gamma_2(\flt \overlap q_2)})}$
        \item $\DEPS{\sigma_1}{\dvalocss{\gamma_1(\EPS[2])}}{\sigma_{21}}$
        \item $\DEPS{\sigma_2}{\dvalocss{\gamma_2(\EPS[2])}}{\sigma_{22}}$
    \end{itemize}

    Let $\sigma_{1a} =  \restricts{\sigma_1}{\dvalocss{\gamma_1(\qsat{\EPS[1]})}}$,
    and $\sigma_{1b} =  \restricts{\sigma_1}{\dvalocss{\gamma_1(\qsat{\EPS[2]})}}$, and

    $\sigma_{1c} = \restricts{\sigma_1}{(\DOM(\sigma_1) - \dvalocss{\gamma_1(\qsat{\EPS[1]})} - \dvalocss{\gamma_1(\qsat{\EPS[2]})})}$.

    By \lemref{lem:eff_sep}, we know,

    (a) $\sigma_{11} = \sigma_{1a}' \ast \sigma_{1b} \ast \sigma_{1c} \ast \sigma_{fr11}$, for some $\sigma_{1a}'$, where $\sigma_{fr11} \ast \sigma_1$; and

    (b) $\sigma_{21} = \sigma_{1a} \ast \sigma_{1b}' \ast \sigma_{1c} \ast \sigma_{fr21}$, for some $\sigma_{1b}'$, where $\sigma_{fr21} \ast \sigma_1$; and

    (c) $\sigma_{fr11} \ast \sigma_{fr21}$.

    Let $\sigma_{2a} =  \restricts{\sigma_2}{\dvalocss{\gamma_2(\qsat{\EPS[1]})}}$,
    and $\sigma_{2b} =  \restricts{\sigma_2}{\dvalocss{\gamma_2(\qsat{\EPS[2]})}}$,
    and

    $\sigma_{2c} = \restricts{\sigma_2}{(\DOM(\sigma_2) - \dvalocss{\gamma_2(\qsat{\EPS[1]})} - \dvalocss{\gamma_2(\qsat{\EPS[2]})})}$.

    By \lemref{lem:eff_sep}, we know,

    (d)  $\sigma_{12} = \sigma_{2a}' \ast \sigma_{2b} \ast \sigma_{2c} \ast \sigma_{fr12}$, for some $\sigma_{2a}'$, where $\sigma_{fr12} \ast \sigma_2$; and

    (e) $\sigma_{22} = \sigma_{2a} \ast \sigma_{2b}' \ast \sigma_{2c} \ast \sigma_{fr22}$, for some $\sigma_{2b}'$, where $\sigma_{fr22} \ast \sigma_2$; and

    (f) $\sigma_{fr12} \ast \sigma_{fr22}$.

    From the left, by the deterministic of reduction  semantics and (a), we know:
    \begin{itemize}
        \item $\sigma_1 \mid \gamma_1(t_1) \mredv{\ast} \sigma_{1a}' \ast \sigma_{1b} \ast \sigma_{1c} \ast \sigma_{fr11} \mid v_{11}$
        \item $\sigma_{1a}' \ast \sigma_{1b} \ast \sigma_{1c} \ast \sigma_{fr11} \mid \gamma_1(t_2) \mredv{\ast}  \sigma_{1a}' \ast \sigma_{1b}' \ast \sigma_{1c} \ast \sigma_{fr11} \ast \sigma_{fr21}' \mid v_{21}$
    \end{itemize}

    From the right, by the deterministic of reduction semantics and (e), we know

    \begin{itemize}
        \item $\sigma_2 \mid \gamma_2(t_2) \mredv{\ast} \sigma_{2a} \ast \sigma_{2b}' \ast \sigma_{2c} \ast \sigma_{fr22} \mid v_{22}$
        \item $\sigma_{2a} \ast \sigma_{2b}' \ast \sigma_{2c} \ast \sigma_{fr22} \mid \gamma_1(t_1) \mredv{\ast}  \sigma_{2a}' \ast \sigma_{2b}' \ast \sigma_{2c} \ast \sigma_{fr{22}} \ast \sigma_{fr12}' \mid v_{12}$
    \end{itemize}

    Let $\sigma_{c}  = \sigma_{1a}' \ast \sigma_{1b}' \ast \sigma_{1c} \ast \sigma_{fr11} \ast \sigma_{fr21}'$ and $\sigma_{d} = \sigma_{2a} \ast \sigma_{2b}' \ast \sigma_{2c} \ast \sigma_{fr{22}} \ast \sigma_{fr12}'$.

    We know that there exists $\W'$, such that $\WFRS{\sigma_c}{\sigma_d}{\W \extends \W'}$, and $\W \extends \W_1 \subseteq \W \extends \W'$. and $\W \extends \W_2 \subseteq \W \extends W'$.

    By \lemref{lem:valt_store_extend}, we know that $\DVTC{\W\extends \W'}{v_{11}}{v_{12}} \in \DGMV{\Type{T_1}}$, and
    $\DVTC{\W\extends \W'}{v_{21}}{v_{22}} \in \DGMV{\Type{T_2}}$

    By the reduction semantics, and  binary term interpretation, we know that

    As $\gamma_2(t_1)[y\mapsto v_{22}][x \mapsto v_{12}] = \gamma_2(t_1)[x \mapsto v_{12}][y\mapsto v_{22}]$,
    we have

    $\DVTC{\W\extends \W'}{\gamma_1(t_1)[x \mapsto v_{11}][y \mapsto v_{21}]}{\gamma_2(t_1)[y\mapsto v_{22}][x \mapsto v_{12}]} \in \DGMt{\varphi}{\ty[q]{T}\ \EPS}$ by the fundamental propety (\thmref{them:direct_fp}) on the second assumption.

\end{proof}

\begin{lemma}[$\lambda$-hoist]\label{lem:lambda_hoist} If
    $\G[\flt]\ts  t_1: \ty[q_1]{T_1}\ \PURE$,
    and $\csxs[q, x, y]{\G\ ,\ x: \ty[p]{T}, y: \ty[\qsat{q_1}\overlap \qsat{(\flt, x)}]{T_1}}  \ts \ t: \ty[r]{U} \ \EPS$,
    and $\theta =[q_1/y]$,
    and $x \not\in \FV(U)$,
    and $y \not\in \FV(U)$,
    then
    $\G[\flt] \models (\lambda x: \ty[p]{T}. \ \tlet \ y = t_1 \ \tin \ t) \equivlog (\tlet \ y = t_1 \ \tin \ \lambda x: \ty[p]{T}. t) : (x: \ty[p]{T} \to^{\EPS} \ty[r]{U}\theta)^{q} \ \PURE$.
\end{lemma}
\begin{proof}
    Let $t_{\lambda} \DEF (\lambda x. \ \tlet \ y = t_1 \ \tin \ t)$, and $t_{b}  \DEF \tlet \ y = t_1 \ \tin \ t $.

    By the fundamental property (\thmref{them:direct_fp}), we know $\G[\flt] \models t_{\lambda} \equivlog t_{\lambda} : (x: \ty[p]{T} \to^{\EPS} \ty[r]{U}\theta)^{q} \ \PURE$.

    Let $(\W, \gamma) \in \UG{\G[\flt]}$ and $\WFRS{\sigma_1}{\sigma_2}{\W}$.

    By the definition of binary logical relations and function type interpretation, we know
    \begin{itemize}
        \item $\dvalocss{\gamma_1(t_{\lambda})} \subq  \DOM_1(\W)$
        \item $\dvalocss{\gamma_2(t_{\lambda})} \subq  \DOM_2(\W)$
    \end{itemize}
    Let $v_1$, $v_2$, $\W'$, $\sigma_1'$ and $\sigma_2'$ be arbitrary, such that
    \begin{itemize}
        \item $\WFRS{\sigma_1'}{\sigma_2'}{\W \extends \W'}$
        \item $\DVTC{\W \extends \W'}{v_1}{v_2} \in \DGMV{T}$
    \end{itemize}
    and we know that there exists $\W''$, $\sigma_1''$, $\sigma_2''$, $v_3$, $v_4$, such that
    \begin{enumerate}[label=(\alph*)]
        \item $\sigma_1' \mid \gamma_1(t_b)[x \mapsto v_1] \mredv{\ast} \sigma_1'' \mid v_3$
        \item $\sigma_2' \mid \gamma_2(t_b)[x \mapsto v_2] \mredv{\ast} \sigma_2'' \mid v_4$
        \item $\WFRS{\sigma''}{\sigma''}{\W\extends\W'\extends\W''}$
        \item $\DVTC{\W \extends \W'\extends\W''}{v_3}{v_4} \in \DGMV{U}$
    \end{enumerate}

    By \thmref{them:direct_fp} again, we know $\G[\flt] \models t_1 \equivlog t_1 : \ty[q_1]{T_1} \ \PURE$.
    By the definition of binary logical relations and term interpretation, we know $\DVTC{\W}{t_1}{t_1} \in \DGMt{\flt}{\ty[q_1]{T_1}\ \PURE}$.
    Then we know there exists $\W_1$, $\sigma_a$, $\sigma_b$, $v_a$ and $v_b$ such that
    \begin{itemize}
        \item $\sigma_1 \mid \gamma_1(t_1) \mredv{\ast} \sigma_a \mid vx_a$
        \item $\sigma_2 \mid \gamma_2(t_1) \mredv{\ast} \sigma_b \mid vx_b$
        \item $\WFRS{\sigma_a}{\sigma_b}{\W\extends\W_1}$
        \item $\DVTC{\W \extends \W_1}{vx_a}{vx_b} \in \DGMV{T_1}$
        \item[$\ast$] $\dvalq{\sigma_1}{vx_a}{\dvalocss{\gamma_1(\flt)} \overlap \dvalocss{\gamma_1(q_1)}}$
        \item[$\ast\ast$] $\dvalq{\sigma_2}{vx_b}{\dvalocss{\gamma_2(\flt)} \overlap \dvalocss{\gamma_2(q_1)}}$
    \end{itemize}
    By \lemref{lem:noeffects}, we know $\sigma_a = \sigma_1$, $\sigma_b = \sigma_2$, and $\W_1 = \emptyset$.
    By \lemref{lem:valt_store_extend}, we know $\DVTC{\W \extends \W'}{vx_a}{vx_b} \in \DGMV{T_1}$.

    Now, we can further specify the reduction of $t_b$ as
    \begin{enumerate}[label=(\Alph*)]
        \item $\sigma_1' \mid \gamma_1(t)[y \mapsto vx_a] \mredv{\ast} \sigma_1'' \mid v_3$
        \item $\sigma_2' \mid \gamma_2(t)[y \mapsto vx_b] \mredv{\ast} \sigma_2'' \mid v_4$
    \end{enumerate}

    Combining (a)  and  (A), and we have
    \begin{enumerate}
        \item $\sigma_1' \mid \gamma_1(t)[x  \mapsto v_1][y \mapsto vx_a] \mredv{\ast} \sigma_1'' \mid v_3$
        \item $\sigma_2' \mid \gamma_2(t)[x  \mapsto v_2][y \mapsto vx_b] \mredv{\ast} \sigma_2'' \mid v_4$
    \end{enumerate}

    By the second assumption and \thmref{them:direct_fp}, we know $$\csxs[q, x, y]{\G\ ,\ x: \ty[p]{T}, y: \ty[\qsat{q_1}\overlap \qsat{(\flt, x)}]{T_1}}  \models \ t \equivlog t : \ty[r]{U} \ \EPS$$.

    By \lemref{lem:envt_tighten}, we know $(\W, \gamma) \in  \UG{\csxs[q, x, y]{\G}}$.
    By $\ast$ and $\ast\ast$, we apply \lemref{lem:envt_extend_all}, and have
    \[
        (\W\extends\W', \gamma\extends(x \mapsto(v_1, v_2))\extends(y \mapsto (vx_a, vx_b))) \in \UG{\csxs[q, x, y]{\G\ ,\ x: \ty[p]{T}, y: \ty[\qsat{q_1}\overlap \qsat{(\flt, x)}]{T_1}} }
    \]

    By definition of binary logic relations and  function type interpretation, we
    \begin{enumerate}[label=(\roman*)]
        \item $\sigma_1' \mid (\gamma_1\extends(x \mapsto v_1)\extends(y \mapsto vx_a))t_1 \mredv{\ast} \sigma_1''' \mid v_5$
        \item $\sigma_2' \mid (\gamma_2\extends(x \mapsto v_2)\extends(y \mapsto vx_a))t_1 \mredv{\ast} \sigma_2''' \mid v_6$
    \end{enumerate}

    We know $\gamma_1(t)[x  \mapsto v_1][y \mapsto vx_a] = (\gamma_1\extends(x \mapsto v_1)\extends(y \mapsto vx_a))t_1$, and

    $\gamma_2(t)[x  \mapsto v_2][y \mapsto vx_b] =  (\gamma_2\extends(x \mapsto v_2)\extends(y \mapsto vx_a))t_1$.

    By the deterministic of the reduction semantics, we know $v_3 = v_5$ and $v_4 = v_6$.
    By the fact we have, the proof is done.
\end{proof}

The following lemma is used in the proof of $\beta$-inlining.
\begin{lemma}\label{lem:subst} $(\G, x: \ty[p]{T})^{q, x}   \ts t_1 : \ty[r]{U} \EPS$,
    and $\G[\flt]\ts t_2 : \ty[p]{T}\ \PURE$,
    and $x\notin\FV(U)$,
    and $\theta = [p/x]$,
    then $\G[\flt] \models t_1[x/t_2] \equivlog t_1[x/t_2] :  (\ty[r]{U}\ \EPS)\theta$.
\end{lemma}
\begin{proof} By induction on the type derivation on the first. Each case follows from the corresponding compatibility lemma.
\end{proof}

\begin{lemma}[$\beta$-inlining]\label{lem:beta_inlining} If
    $(\G, x: \ty[p]{T})^{q, x}   \ts t_1 : \ty[r]{U} \EPS$,
    and $\G[\flt]\ts t_2 : \ty[p]{T}\ \PURE$,
    and $x\notin\FV(U)$,
    and $\theta = [p/x]$,
    then $\G[\flt] \ts (\lambda x. t_1)(t_2) \equivlog t_1[x/t_2] :  (\ty[r]{U}\ \EPS)\theta$.
\end{lemma}
\begin{proof}

    By the fundamental property (\thmref{them:direct_fp}), we know
    $\G[\flt] \models t_2 \equivlog t_2: \ty[p]{T}\ \PURE$.

    Let $(\W, \gamma) \in \UG{\G[\flt]}$ and $\WFRS{\sigma_1}{\sigma_2}{\W}$.
    By the definition of binary logical relations and binary term interpretation, we know there exists $\sigma_{21}$, $\sigma_{22}$, $\W_2$, $v_{x1}$ and $v_{x2}$, such that
    \begin{itemize}
        \item $\sigma_1 \mid \gamma_1(t_2) \mredv{\ast} \sigma_{21} \mid v_{x1}$
        \item $\sigma_2 \mid \gamma_2(t_2) \mredv{\ast} \sigma_{22} \mid v_{x2}$
        \item $\DVTC{\W\extends\W_2}{v_{x1}}{v_{x2}} \in \DGMV{\Type{T}}$
        \item $\WFRS{\sigma_{21}}{\sigma_{22}}{\W \extends \W_2}$
        \item $\dvalq{\sigma_1}{v_{x1}}{(\dvalocss{\gamma_1(\flt)} \overlap \dvalocss{\gamma_{1}(p)})}$
        \item $\dvalq{\sigma_2}{v_{x2}}{(\dvalocss{\gamma_2(\flt)} \overlap \dvalocss{\gamma_{2}(p)})}$
        \item $\DEPS{\sigma_1}{\emptyset}{\sigma_{21}}$
        \item $\DEPS{\sigma_2}{\emptyset}{\sigma_{22}}$
    \end{itemize}

    By \lemref{lem:noeffects}, we know $\sigma_{21} = \sigma_1$, $\sigma_{22} = \sigma_2$. Then by definition of world, $\W_2 = \emptyset$.

    Now we know $\DVTC{\W}{v_{x1}}{\gamma_2(t_2)} \in \DGMt{\flt}{\ty[q]{T}\ \PURE}$ by definition of term interpretation.

    We know $\G[\flt] \models (\lambda x. t_1)(t_2) \equivlog (\lambda x. t_1)(t_2): (\ty[r]{U}\ \EPS)\theta$.

    By the definition of binary logical relations and binary term interpretation, we know there exists $\sigma_1'$, $\sigma_2'$, $\W_1$, $v_{y1}$ and $v_{y2}$, such that
    \begin{itemize}
        \item $\sigma_1 \mid \gamma_1(t_1)[x \mapsto v_{x1}] \mredv{\ast} v_{y1} \mid \sigma_1'$
        \item $\sigma_2 \mid \gamma_2(t_1)[x \mapsto v_{x2}] \mredv{\ast} v_{y2} \mid \sigma_2'$
        \item $\WFRS{\sigma_1'}{\sigma_2'}{\W\extends\W_1}$
        \item $\DVTC{\W\extends\W_1}{v_{y1}}{v_{y2}} \in \DGMV{\Type{U}}$
    \end{itemize}

    By \lemref{lem:subst}, the definition of binary logical relations and binary term interpretation, we know there exists $\sigma_3$, $\sigma_4$, $\W'$, $v_3$ and $v_4$, such that
    \begin{itemize}
        \item $\sigma_1 \mid \gamma_1(t_1[x/t_2]) \mredv{\ast} \sigma_3 \mid v_3$
        \item $\sigma_2 \mid \gamma_2(t_1[x/t_2]) \mredv{\ast} \sigma_4 \mid v_4$
        \item $\WFRS{\sigma_3}{\sigma_4}{\W\extends \W'}$
        \item $\DVTC{\W\extends\W'}{v_3}{v_4} \in \DGMV{\Type{U}}$
    \end{itemize}

    By $\DVTC{\W}{v_{x1}}{\gamma_2(t_2)} \in \DGMt{\flt}{\ty[q]{T}\ \PURE}$,
    we know $v_4 = v_{y2}$. Then we are done.

\end{proof}

The following lemma is used in the proof of \lemref{lem:e_cse}.
\begin{lemma} \label{lem:subst_var} If
    $\G[\flt] \ts t_1: \ty[q_1]{T_1} \ \EPS[1]$
    and $\csxs[\flt, x, y]{\G\ ,\ x: \ty[\qsat{q_1} \overlap \qsat{\flt}]{T_1}, y: \ty[\qsat{q_1}\overlap \qsat{(\flt, x)}]{T_1}} \ts \ t: \ty[q]{T} \ \EPS$,
    and $\theta = [x/y]$,
    and $\omega \not\in \EPS[1]$,
    then
    $\G[\flt] \models \ (\tlet \ x = t_1 \ \tin \ t\ \theta) \equivlog (\tlet \ x = t_1 \ \tin \ t\ \theta) : (\ty[q]{T} \ \EPS)\theta \EFFSEQ\FX{\EPS[1]}$.
\end{lemma}
\begin{proof} By induction on the type derivation on the second. Each case follows from  the corresponding compatibility lemma.
\end{proof}

\begin{lemma}[e-cse]\label{lem:e_cse} If
    $\G[\flt] \ts t_1: \ty[q_1]{T_1} \ \EPS[1]$,
    and $\csxs[\flt, x, y]{\G\ ,\ x: \ty[\qsat{q_1} \overlap \qsat{\flt}]{T_1}, y: \ty[\qsat{q_2}\overlap \qsat{(\flt, x)}]{T_2}} \ts \ t: \ty[q]{T} \ \EPS$,
    and $\omega \notin \EPS[1]$,
    and $\theta = [x/y]$,
    then
    $\G[\flt] \models (\tlet\ x = t_1 \ \tin \ \tlet \ y = t_1 \ \tin \ t) \equivlog (\tlet \ x = t_1 \ \tin \ t\ \theta) : (\ty[q]{T} \ \EPS)\theta \EFFSEQ\FX{\EPS[1]}$.
\end{lemma}
\begin{proof}

    By the fundamental property (\thmref{them:direct_fp}), we know
    $\G[\flt] \models t_1 \equivlog t_1: \ty[q_1]{T_1}\ \EPS[1]$.

    Let $(\W, \gamma) \in \UG{\G[\flt]}$ and $\WFRS{\sigma_1}{\sigma_2}{\W}$.
    By the definition of binary logical relations and binary term interpretation, we know
    there exists $\sigma_{11}$, $\sigma_{12}$, $\W_1$, $v_{11}$ and $v_{12}$, such that
    \begin{itemize}
        \item $\sigma_1 \mid \gamma_1(t_1) \mredv{\ast} \sigma_{11} \mid v_{11}$
        \item $\sigma_2 \mid \gamma_2(t_1) \mredv{\ast} \sigma_{12} \mid v_{12}$
        \item $\DVTC{\W\extends\W'}{v_{11}}{v_{12}} \in \DGMV{\Type{T_1}}$
        \item $\WFRS{\sigma_{11}}{\sigma_{12}}{\W \extends \W_1}$
        \item $\dvalq{\sigma_1}{v_{11}}{(\dvalocss{\gamma_1(\flt)} \overlap \dvalocss{\gamma_{1}(q_1)})}$
        \item $\dvalq{\sigma_2}{v_{12}}{(\dvalocss{\gamma_2(\flt)} \overlap \dvalocss{\gamma_{2}(q_1)})}$
        \item $\DEPS{\sigma_1}{\dvalocss{\gamma_1(\EPS[1])}}{\sigma_{11}}$
        \item $\DEPS{\sigma_2}{\dvalocss{\gamma_2(\EPS[1])}}{\sigma_{12}}$
    \end{itemize}

    As $\omega \not\in \EPS[1]$, we know $\W_1 = \emptyset$.

    From the left, by the reduction semantics, we know there exists $\sigma_{11}$, $\sigma_{12}$, $\W_1$, $v_{11}$ and $v_{12}$, such that
    \begin{itemize}
        \item $\sigma_1 \mid \gamma_1(t_1) \mredv{\ast} \sigma_{11} \mid v_{11}$
        \item $\sigma_{11} \mid \gamma_1(t_1) \mredv{\ast} \sigma_{11}' \mid v_{11}'$
        \item $\sigma_{11}' \mid \gamma_1(t[x\mapsto v_{11}][y \mapsto v_{11}']) \mredv{\ast} \sigma_{1}' \mid v_1$
    \end{itemize}
    where $v_{11} = v_{11}'$ and $\sigma_{11} = \sigma_{11}'$,  by the deterministic of the reduction semantics.

    Thus, the last reduction step can be re-written as
    \[
        \sigma_{11} \mid \gamma_1(t[x\mapsto v_{11}][y \mapsto v_{11}]) \mredv{\ast} \sigma_{1}' \mid v_1,
    \]
    where $\DOM(\sigma_{11}) = \DOM_1{W}$.

    By the second assumption and \thmref{them:direct_fp}, we know $$\csxs[\flt, x, y]{\G\ ,\ x: \ty[\qsat{q_1} \overlap \qsat{\flt}]{T_1}, y: \ty[\qsat{q_2}\overlap \qsat{(\flt, x)}]{T_2}} \models \ t \equivlog t: \ty[q]{T} \ \EPS$$.

    By \lemref{lem:subst_var}, we  know
    \[
        \G[\flt] \models (\tlet \ x = t_1 \ \tin \ t\ \theta) \equivlog (\tlet \ x = t_1 \ \tin \ t\ \theta) : (\ty[q]{T} \ \EPS)\theta \EFFSEQ\FX{\EPS[1]}.
    \]

    By definition of binary logic relations and term interpretation, after reducing $t_1$, we know there exists $\sigma_{1a}$, $\sigma_{2a}$, $\W_2$, $v_3$ and $v_4$, such that
    \begin{itemize}
        \item  $\sigma_{11} \mid (\gamma_1)(t\theta)[x \mapsto v_{11}] \mredv{\ast} \sigma_{1a} \mid v_{3}$
        \item  $\sigma_{12} \mid (\gamma_2)(t\theta)[x \mapsto v_{12}] \mredv{\ast} \sigma_{2a} \mid v_{4}$
        \item $\WFRS{\sigma_{1a}}{\sigma_{2a}}{\W\extends\W_1\extends\W_2}$
        \item $\dvalq{\sigma_{1a}}{v_{11}}{(\dvalocss{\gamma_1(\flt)} \overlap \dvalocss{\gamma_{1}(q \ \theta)})}$
        \item $\dvalq{\sigma_{2a}}{v_{12}}{(\dvalocss{\gamma_2(\flt)} \overlap \dvalocss{\gamma_{2}(q \ \theta)})}$
        \item $\DEPS{\sigma_{1a}}{\dvalocss{\gamma_1(\EPS \FX{\theta} \EFFSEQ\EPS[1])}}{\sigma_{11}}$
        \item $\DEPS{\sigma_{2a}}{\dvalocss{\gamma_2(\EPS \FX{\theta} \EFFSEQ\EPS[2])}}{\sigma_{12}}$
    \end{itemize}

    We know
    $(\gamma_1(t[x/y]))[x \mapsto v_{11}] = (\gamma_1(t[v_{11}/y]))[x \mapsto v_{11}]$.
    Then we know $v_1 = v_3$.
    By the fact we have, the proof is done.

\end{proof}

\section{Optimization Rules and Equational Theory of \langg}\label{sec:optimizations}

In this section, we justify the soundness of the optimization rules shown in the main paper, \ie,
they equate contextually equivalent graphs. Our approach reuses the logical relation
development for the direct-style \directlang{} system from \Cref{sec:direct-lr}, through the use of
a ``round-trip'' translation exploiting the functional properties of dependency erasure and
synthesis. The key point is that this is already enough: dependencies, a type-level artifact derived
from the effects annotated in graphs, possess no operational meaning other than adhering to the order of
runtime effects, a consequence of type soundness.

\subsection{Logical and Contextual Equivalence for the Monadic Normal Form \mnflang}\label{sec:optimizations:monadic}

As the first step, we establish that we can restrict the logical relations development for the
direct-style \directlang (\Cref{sec:direct-lr}) to \mnflang in MNF due to the results proved in
\Cref{sec:monadic}. That is, MNF is a proper sublanguage of \directlang and preserved by reductions
(\Cref{lem:mnf:reduction_preserves_mnf,lem:type_preservation:translation_backwards}). Therefore, MNF
is also preserved by the logical equivalence (\Cref{def:direct:log_equiv}), and contextual
equivalence (\Cref{def:standard_equiv}). Note that by restricting to monadic syntax, the contexts
\(C\) in \Cref{fig:direct_context} can only be of the shape
\[C ::= \square \mid \tlet~{x = C}~\tin~g \mid \tlet~{x = \lambda y.C}~\tin~g \mid \tlet~{x =
b}~\tin~C  \] for which we can derive the specialized context typing rules (\Cref{fig:mnf_context}).
Finally, some equational rules on expressions (\eg, \rulename{$\beta$-inlining},
\Cref{fig:direct_equiv}) require a translation into MNF first to fit into the syntactic constraints
of \mnflang. This is always feasible due to the totality of the translation and its type-and-effect
preservation (\Cref{lem:type_preservation:translation}).

Thus, without further ado, we will treat the development of \Cref{sec:direct-lr} as being defined over \mnflang.

\begin{figure}[t]\small
	\begin{mdframed}
		\judgement{Context for Contextual Equivalence}{}
        \[C ::= \square \mid \tlet~{x = C}~\tin~g \mid \tlet~{x = \lambda y.C}~\tin~g \mid \tlet~{x = b}~\tin~C  \]
		\judgement{Context Typing Rules}{\BOX{C : (\G[\flt]; \ty[q]{T}\ \EPS) \carrow (\G[\flt]; \ty[q]{T}\ \EPS)}}
		\begin{minipage}[t]{1.0\linewidth}\vspace{0pt}
			\infrule[c-hole]{\ 
				\G\ts\ty[p]{S}\ \EPS[1] <: \ty[q]{T}\ \EPS[2]
			}{
				\square: (\G[\flt]; \ty[p]{{S}}\ \EPS[1]) \carrow (\G[\flt]; \ty[q]{{T}}\ \EPS[2])
			}
            \vgap
			\infrule[c-let-1]
			{
			C: (\GP[\fltp]; \ty[r']{{U'}} \ \EPS[1]) \carrow(\G[\flt];  \ty[r]{{U}} \ \EPS[2])
			\quad
			(\Gamma,\, x : \ty[\qsat{r}\cap\qsat{\flt}]{{U}})^{\flt,x} \tsM g: \ty[p]{{T}} \ \EPS[3] \\
			x\notin\FV({T}) \quad \theta = [r/x]
			}
			{
			\tlet ~ x = C~ \tin~ g: (\GP[\fltp];  \ty[r']{{U'}} \ \EPS[1]) \carrow (\G[\flt];(\ty[p]{{T}} \ \EPS[2] \EFFSEQ \EPS[3])\theta)
			}
			\vgap
			\infrule[c-let-$\lambda$]
			{
            C: (\GP[\fltp]; \ty[r']{{U'}} \ \EPS[1]) \carrow ({(\Gamma, y : \ty[p]{S})}^{\,q'',y}; \ty[q]{{T}}  \ \EPS[2]) \\
			(\G, x: \ty[\qsat{q''} \overlap \qsat{\flt}]{((y : \ty[p]{{S}})\to^{\EPS[2]}\ty[q]{{T}})})^{\flt,x} \tsM g: \ty[r]{{U}} \ \EPS[3] \\
			x\notin\FV({U}) \quad \theta = [q''/x]\quad q''\subq\flt
			}
			{
			\tlet ~ x = \lambda y.C~ \tin ~g : (\GP[\fltp]; \ty[r']{{U'}} \ \EPS[1]) \carrow (\G[\flt]; (\ty[r]{{U}} \ \EPS[3])\theta)
			}
            \vgap
            \infrule[c-let-2]
			{
			\G[\flt] \tsM b: \ty[r]{{S}} \ \EPS[1]\\
			C: (\GP[\fltp]; \ty[r']{{U'}} \ \EPS[2]) \carrow ({(\G, x : \ty[\qsat{r}\cap\qsat{\flt}]{\ty{{S}}})}^{\,\flt,x}; \ty[p]{{T}}  \ \EPS[3])\\
			x\notin\FV({T})\quad\theta = [r/x]
			}
			{
			\tlet ~ x = b~ \tin ~C : (\GP[\fltp]; \ty[r']{{U'}} \ \EPS[2]) \carrow (\G[\flt]; (\ty[p]{{T}}  \ \EPS[1]\EFFSEQ\EPS[3])\theta)
			}
		\end{minipage}%
		\caption{Context typing rules for the \mnflang{}-Calculus.}
		\label{fig:mnf_context}
	\end{mdframed}
\end{figure}

\subsection{Context Typing, Synthesis and Erasure}

In this section, we establish properties about contexts used for contextual equivalence in \langg{}. We essentially lift
the dependency synthesis (\Cref{fig:graphir:synthesis}) to a notion of dependency synthesis for \mnflang{}
contexts (\Cref{fig:mnf_context}).

\begin{definition}[Dependency Erasure]
    We write \(\erased{g}\) (\(\erased{b}\)), for the erasure of effect dependencies from \langg{} graph terms (bindings), yielding
    their unannotated version in \mnflang{} (\Cref{sec:monadic}), as well as \(\erased{C}\) for erasing dependencies in
    \langg{} contexts from \Cref{fig:graphir_context} into \mnflang{} contexts from \Cref{fig:mnf_context}.
\end{definition}

\begin{definition}[Context Dependency Synthesis]\label{def:ctx_synth}
    We write \[\DELTA \ts C: (\GP[\fltp];\ty[q']{T'}\ \EPSPR) \carrow (\G[\flt]; \ty[q]{T}\ \EPS) \yields{\mathbf{C}\has\DELTA'}\]
    for context synthesis, obtained from lifting the dependency synthesis (\Cref{fig:graphir:synthesis}) to context typing derivations
    (\Cref{fig:mnf_context}).
\end{definition}
Intuitively, the synthesis over contexts yields the annotated context with respect to the ambient last-uses coeffect \(\DELTA\),
with \(\DELTAP\) being the ambient last uses at the hole of the context. The functional properties of dependency synthesis
on \mnflang{} terms (\Cref{sec:hardsoft:metatheory}) carry over analogously to contexts:
\begin{lemma}[Soundness of Context Dependency Synthesis]\label{lem:ctx_synth_soundness}
If \[\DELTA \ts C: (\GP[\fltp];\ty[q']{T'}\ \EPSPR) \carrow (\G[\flt]; \ty[q]{T}\ \EPS) \yields{\mathbf{C}\has\DELTA'}\]
then
    \[\mathbf{C}: (\GP[\fltp]\has\DELTAP;\ty[q']{T'}\ \EPSPR) \carrow (\G[\flt]\has\DELTA; \ty[q]{T}\ \EPS) \]
\end{lemma}
\begin{proof}
    By induction over the derivation \(\DELTA \ts C: (\GP[\fltp];\ty[q']{T'}\ \EPSPR) \carrow (\G[\flt]; \ty[q]{T}\ \EPS) \yields{\mathbf{C}\has\DELTA'}\).
\end{proof}

\begin{lemma}[Context Dependency Synthesis is Total]\label{lem:ctx_synth_totality}
    If \[C: (\GP[\fltp];\ty[q']{T'}\ \EPSPR) \carrow (\G[\flt]; \ty[q]{T}\ \EPS)\]
    then for all \(\DELTA\) with \(\DOM(\DELTA) = \DOM(\G)\) there is \(\DELTAP\) with \(\DOM(\DELTAP) = \DOM(\GP)\) and
    an annotated context \(\mathbf{C}\) where
\[\DELTA \ts C: (\GP[\fltp];\ty[q']{T'}\ \EPSPR) \carrow (\G[\flt]; \ty[q]{T}\ \EPS) \yields{\mathbf{C}\has\DELTA'}\]
   and \(\erased{\mathbf{C}} = C\).
\end{lemma}
\begin{proof}
    By induction over the context typing derivation for \(C\).
\end{proof}

\begin{lemma}[Context Re-Synthesis]\label{lem:ctx_resynth}
    If \[C: (\GP[\fltp]\has\DELTAP;\ty[q']{T'}\ \EPSPR) \carrow (\G[\flt]\has\DELTA; \ty[q]{T}\ \EPS) \]
    then \[\DELTA\ts\erased{C}: (\GP[\fltp];\ty[q']{T'}\ \EPSPR) \carrow (\G[\flt]; \ty[q]{T}\ \EPS) \yields{{\color{black}C}\has\DELTAP}. \]
\end{lemma}
\begin{proof}
    By induction over the context typing derivation for \(C\).
\end{proof}

\begin{lemma}[Decomposition]\label{lem:ctx_graphir_decomposition}
    If\ \ \(\G[\flt]\has\DELTA\ts \CX{C}{g} : \ty[q]{T}\ \EPS\), then
    \(\GP[\fltp]\has\DELTAP\ts g : \ty[p]{S}\ \EPSPR\) and \(C : (\GP[\fltp]\has\DELTAP;\ty[p]{S}\ \EPSPR) \carrow (\G[\flt]\has\DELTA; \ty[q]{T}\ \EPS)\)
    for some \(\GP\), \(\fltp\), \(\DELTAP\), \(S\), \(p\), and \(\EPSPR\).
\end{lemma}
\begin{proof}
    By induction over the context \(C\).
\end{proof}

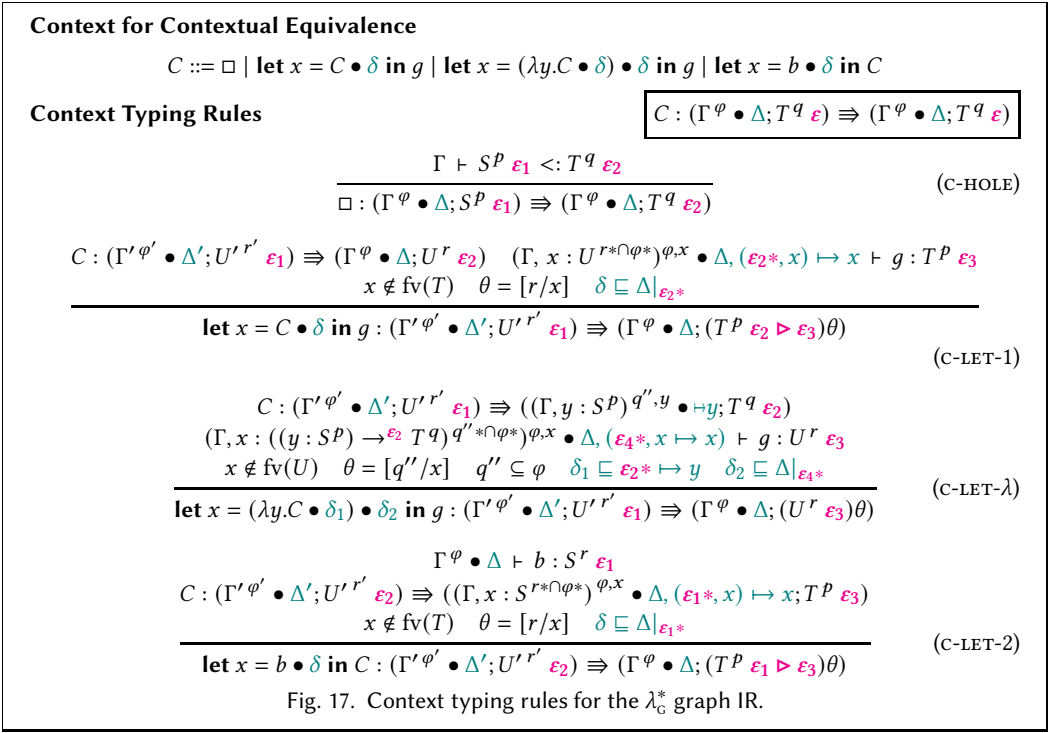
\begin{lemma}[Plugging]\label{lem:ctx_graphir_plugging}
    If\ \ \(\GP[\fltp]\has\DELTAP\ts g : \ty[p]{S}\ \EPSPR\) and \(C : (\GP[\fltp]\has\DELTAP;\ty[p]{S}\ \EPSPR) \carrow (\G[\flt]\has\DELTA; \ty[q]{T}\ \EPS)\),
    then \(\G[\flt]\has\DELTA\ts \CX{C}{g} : \ty[q]{T}\ \EPS\).
\end{lemma}
\begin{proof}
    By induction over the context typing for \(C\).
\end{proof}

\begin{lemma}[Synthesis Plugging]\label{lem:ctx_graphir_synth_plus}
    If \[\DELTA\ts C: (\GP[\fltp];\ty[q']{T'}\ \EPSPR) \carrow (\G[\flt]; \ty[q]{T}\ \EPS) \yields{\mathbf{C}\has\DELTAP} \]
    and
    \[\G[\flt]\has\DELTAP\vdash g:\ty[q']{T'}\ \EPSPR\yields{\mathbf{g}\has\DELTAP\vert_{\FX{\qsat{\EPSPR}}}}\]
    then \[\G[\flt]\has\DELTA\ts \CX{C}{g} : \ty[q]{T}\ \EPS \yields{\CX[mute]{\mathbf{C}}{\mute{\mathbf{g}}}\has\DELTA\vert_{\FX{\qsat{\EPS}}}}.\]
\end{lemma}
\begin{proof}
By induction over the context \(C\).
\end{proof}

\subsection{Logical and Contextual Equivalence for \langg with Hard Dependencies}\label{sec:optimizations:contextual_equiv}

\begin{figure}[t]\small
	\begin{mdframed}
		\judgement{Context for Contextual Equivalence}{}
        \[C ::= \square \mid \tlet~{x = C\has\DEP}~\tin~g \mid \tlet~{x = (\lambda y.C\has\DEP)\has\DEP}~\tin~g \mid \tlet~{x = b\has\DEP}~\tin~C  \]
		\judgement{Context Typing Rules}{\BOX{C : (\G[\flt]\has\DELTA; \ty[q]{T}\ \EPS) \carrow (\G[\flt]\has\DELTA; \ty[q]{T}\ \EPS)}}
		\begin{minipage}[t]{1.0\linewidth}\vspace{0pt}
			\infrule[c-hole]{\ 
				\G\ts\ty[p]{S}\ \EPS[1] <: \ty[q]{T}\ \EPS[2]
			}{
				\square: (\G[\flt]\has\DELTA; \ty[p]{{S}}\ \EPS[1]) \carrow (\G[\flt]\has\DELTA; \ty[q]{{T}}\ \EPS[2])
			}
            \vgap
			\infrule[c-let-1]
			{
			C: (\GP[\fltp]\has\DELTAP; \ty[r']{{U'}} \ \EPS[1]) \carrow
			(\G[\flt]\has\DELTA;  \ty[r]{{U}} \ \EPS[2])
			\quad
			(\Gamma,\, x : \ty[\qsat{r}\cap\qsat{\flt}]{{U}})^{\flt,x} \has\mute{\DELTA,(\FX{\qsat{\EPS[2]}},x)\mapsto x}\ts g: \ty[p]{{T}} \ \EPS[3] \\
			x\notin\FV({T}) \quad \theta = [r/x]\quad \mute{\DEP\sqsubseteq \DELTA\vert_{\FX{\qsat{\EPS[2]}}}}
			}
			{
			\tlet ~ x = C\has\DEP~ \tin~ g: (\GP[\fltp]\has\DELTAP;  \ty[r']{{U'}} \ \EPS[1]) \carrow (\G[\flt]\has\DELTA;(\ty[p]{{T}} \ \EPS[2] \EFFSEQ \EPS[3])\theta)
			}
			\vgap
			\infrule[c-let-$\lambda$]
			{
            C: (\GP[\fltp]\has\DELTAP; \ty[r']{{U'}} \ \EPS[1]) \carrow ({(\Gamma, y : \ty[p]{S})}^{\,q'',y}\has\mute{\pointsto{y}}; \ty[q]{{T}}  \ \EPS[2]) \\
			(\G, x: \ty[\qsat{q''} \overlap \qsat{\flt}]{((y : \ty[p]{{S}})\to^{\EPS[2]}\ty[q]{{T}})})^{\flt,x}\has\DELTA\mute{,(\FX{\qsat{\EPS[4]}},x\mapsto x)} \ts g: \ty[r]{{U}} \ \EPS[3] \\
			x\notin\FV({U}) \quad \theta = [q''/x]\quad q''\subq\flt\quad\mute{\DEP[1]\sqsubseteq\FX{\qsat{\EPS[2]}}\mapsto y}\quad\mute{\DEP[2]\sqsubseteq\DELTA\vert_{\FX{\qsat{\EPS[4]}}}}
			}
			{
			\tlet ~ x =(\lambda y.C\has\DEP[1])\has\DEP[2]~ \tin ~g : (\GP[\fltp]\has\DELTAP; \ty[r']{{U'}} \ \EPS[1]) \carrow (\G[\flt]\has\DELTA; (\ty[r]{{U}} \ \EPS[3])\theta)
			}
            \vgap
			\infrule[c-let-2]
			{
			\G[\flt]\has\DELTA \ts b: \ty[r]{{S}} \ \EPS[1]\\
			C: (\GP[\fltp]\has\DELTAP; \ty[r']{{U'}} \ \EPS[2]) \carrow ({(\G, x : \ty[\qsat{r}\cap\qsat{\flt}]{\ty{{S}}})}^{\,\flt,x}\has\mute{\DELTA,(\FX{\qsat{\EPS[1]}},x)\mapsto x}; \ty[p]{{T}}  \ \EPS[3])\\
			x\notin\FV({T})\quad\theta = [r/x]\quad\mute{\DEP\sqsubseteq\DELTA\vert_{\FX{\qsat{\EPS[1]}}}}
			}
			{
			\tlet ~ x = b\has\DEP~ \tin ~C : (\GP[\fltp]\has\DELTAP; \ty[r']{{U'}} \ \EPS[2]) \carrow ({{\G}}^{\,\flt}\has\DELTA; (\ty[p]{{T}}  \ \EPS[1]\EFFSEQ\EPS[3])\theta)
			}
		\end{minipage}%
		\caption{Context typing rules for the \langg{} graph IR.}
		\label{fig:graphir_context}
	\end{mdframed}
\end{figure}

The key point is that we can resort to the metatheory of the direct-style type-and-effect system
\directlang, because (1) MNF is a sublanguage of the direct-style language
(\Cref{lem:type_preservation:translation,lem:type_preservation:translation_backwards}), and (2)
dependencies are entirely determined by assigned effects (\Cref{lem:hardsoft:deps-invariant}). Furthermore, effect
dependencies have no operational meaning beyond asserting that they respect the observed
call-by-value evaluation order of effects (\Cref{coro:hardsoft:dep-safety}). Thus, from those results we can
appeal to the logical relation and contextual equivalence of \mnflang{}
(\Cref{sec:optimizations:monadic}) by the erasure and re-synthesis of dependencies to derive their counterparts for \langg:
\begin{definition}[Logical Equivalence for \langg]\label{def:graphir:log_equiv}\hfill\vspace{-8pt}
    \infrule{\G[\flt]\models \erased{g_1} \equivlog \erased{g_2} : \ty[q]{T}\ \EPS
    \quad \G[\flt]\has\DELTA\vdash \erased{g_1}:\ty[q]{T}\ \EPS\yields{{\color{black}g_1}\has\DELTA\vert_{\FX{\qsat{\EPS}}}}
    \quad \G[\flt]\has\DELTA\vdash \erased{g_2}:\ty[q]{T}\ \EPS\yields{{\color{black}g_2}\has\DELTA\vert_{\FX{\qsat{\EPS}}}}
    }{\G[\flt]\has\DELTA\models g_1 \equivlog g_2 : \ty[q]{T}\ \EPS}
\end{definition}
\begin{definition}[Contextual Equivalence for \langg]\label{def:graphir:ctx_equiv}\hfill\vspace{-8pt}
    \infrule{\G[\flt]\models \erased{g_1} \equiva \erased{g_2} : \ty[q]{T}\ \EPS
    \quad \G[\flt]\has\DELTA\vdash \erased{g_1}:\ty[q]{T}\ \EPS\yields{{\color{black}g_1}\has\DELTA\vert_{\FX{\qsat{\EPS}}}}
    \quad \G[\flt]\has\DELTA\vdash \erased{g_2}:\ty[q]{T}\ \EPS\yields{{\color{black}g_2}\has\DELTA\vert_{\FX{\qsat{\EPS}}}}
    }{\G[\flt]\has\DELTA\models g_1 \equiva g_2 : \ty[q]{T}\ \EPS}
\end{definition}
\noindent
Intuitively, graph terms are logically/contextually equivalent iff their dependency-erased versions in \mnflang{} are logically/contextually
equivalent, and we can recover the original terms by re-synthesizing their dependencies. More precisely, they
are in the image of the synthesis function with respect to the last-use coeffect \(\DELTA\) and the effect \(\EPS\)
(cf.\ the synthesis invariant~\Cref{lem:hardsoft:deps-invariant}).

\subsubsection{Properties of Logical Relations}

\begin{theorem}[Fundamental Property]\label{them:graphir_fp} If\ \ $\G[\flt]\has\DELTA \ts g: \ty[q]{T} \ \EPS $, then $\G[\flt]\has\DELTA \models g \equivlog g : \ty[q]{T}\ \EPS$.
\end{theorem}
\begin{proof}
    $\G[\flt]\has\DELTA \ts g: \ty[q]{T} \ \EPS $ implies $\G[\flt] \tsM \erased{g}: \ty[q]{T} \ \EPS $ and
    by the fundamental \Cref{them:direct_fp}, it follows that $\G[\flt] \models \erased{g} \equivlog \erased{g} : \ty[q]{T} \ \EPS $.
    Finally, by \Cref{lem:hardsoft:synth-totality}, we have that \(\G[\flt]\has\DELTA\vdash \erased{g}:\ty[q]{T}\ \EPS\yields{{\color{black}g}\has\DELTA\vert_{\FX{\qsat{\EPS}}}}\).
\end{proof}

\begin{lemma}[Congruency of Binary Logical Relations]\label{lem:graphir_congruence} The binary logical relation is closed under well-typed program contexts,
    \ie, if\ \  $\G[\flt]\has\DELTA \models g_1 \equivlog g_2: \ty[p]{T} \ \EPS$,
    and $C:(\G[\flt]\has\DELTA; \ty[p]{T} \ \EPS) \carrow (\GP[\flt']\has\DELTAP; \ty[p']{T'} \ \EPS')$, then $\GP[\flt']\has\DELTAP \models C[g_1] \equivlog C[g_2]: \ty[p']{T'} \ \EPS'$.
\end{lemma}
\begin{proof}
    \begin{enumerate}
        \item By definition of logical equivalence for \langg{}:
        \begin{enumerate}
            \item $\G[\flt] \models \erased{g_1} \equivlog \erased{g_2}: \ty[p]{T} \ \EPS$.
            \item $\G[\flt]\has\DELTA\vdash \erased{g_1}:\ty[q]{T}\ \EPS\yields{{\color{black}g_1}\has\DELTA\vert_{\FX{\qsat{\EPS}}}}$.
            \item $\G[\flt]\has\DELTA\vdash \erased{g_2}:\ty[q]{T}\ \EPS\yields{{\color{black}g_2}\has\DELTA\vert_{\FX{\qsat{\EPS}}}}$.
        \end{enumerate}
        \item By erasure: $\erased{C}:(\G[\flt]; \ty[p]{T} \ \EPS) \carrow (\GP[\flt']; \ty[p']{T'} \ \EPS')$.
        \item By the congruence \Cref{lem:graphir_congruence}, and (1), (2): $\G[\flt] \models \CX{\erased{C}}{\erased{g_1}} \equivlog \CX{\erased{C}}{\erased{g_2}}: \ty[p]{T} \ \EPS$.
        \item (3) is equivalent to $\G[\flt] \models \erased{\CX{C}{g_2}} \equivlog \erased{\CX{C}{g_2}}: \ty[p]{T} \ \EPS$.
        \item By assumption, (2), and re-synthesis \Cref{lem:ctx_resynth}:
        \[\DELTA\ts\erased{C}:(\G[\flt]; \ty[p]{T} \ \EPS) \carrow (\GP[\flt']; \ty[p']{T'} \ \EPS')\yields{{\color{black}C}\has\DELTAP}.\]
        \item By (1a), (1c), (5), and synthesis plugging \Cref{lem:ctx_graphir_synth_plus}:
        \begin{enumerate}
            \item $\G[\flt]\has\DELTA\vdash \erased{\CX{C}{g_1}}:\ty[q]{T}\ \EPS\yields{{\color{black}\CX{C}{g_1}}\has\DELTA\vert_{\FX{\qsat{\EPS}}}}$.
            \item $\G[\flt]\has\DELTA\vdash \erased{\CX{C}{g_2}}:\ty[q]{T}\ \EPS\yields{{\color{black}\CX{C}{g_2}}\has\DELTA\vert_{\FX{\qsat{\EPS}}}}$.
        \end{enumerate}
        \item (4) and (6) prove the goal.
    \end{enumerate}
\end{proof}

\begin{theorem}[Soundness of Binary Logical Relations]\label{thm:graphir_lr_soundness} The binary logical relation is sound w.r.t. contextually equivalence, \ie,
    if\ \  $\G[\flt]\has\DELTA \ts g_1: \ty[p]{T}\ \EPS$ and $\G[\flt]\has\DELTA \ts g_2: \ty[p]{T}\ \EPS$, then
    $\G[\flt]\has\DELTA \models g_1 \equivlog g_2: \ty[p]{T}\ \EPS$ implies
    $\G[\flt]\has\DELTA \models g_1 \equiva g_2: \ty[p]{T}\ \EPS$.
\end{theorem}
\begin{proof}
    \begin{enumerate}
        \item By definition of logical equivalence for \langg{}:
        \begin{enumerate}
            \item $\G[\flt] \models \erased{g_1} \equivlog \erased{g_2}: \ty[p]{T} \ \EPS$.
            \item $\G[\flt]\has\DELTA\vdash \erased{g_1}:\ty[q]{T}\ \EPS\yields{{\color{black}g_1}\has\DELTA\vert_{\FX{\qsat{\EPS}}}}$.
            \item $\G[\flt]\has\DELTA\vdash \erased{g_2}:\ty[q]{T}\ \EPS\yields{{\color{black}g_2}\has\DELTA\vert_{\FX{\qsat{\EPS}}}}$.
        \end{enumerate}
        \item By (1a), and soundness \Cref{thm:direct_lr_soundness}: $\G[\flt] \models \erased{g_1} \equiva \erased{g_2}: \ty[p]{T} \ \EPS$.
        \item The goal follows by (1b), (1c), (2), and the definition of contextual equivalence for \langg{}.
    \end{enumerate}
\end{proof}

\subsection{Soundness of the Optimization Rules}\label{sec:optimizations:soundness}

\begin{figure*}[t]
    \begin{mdframed}
        \begin{mathpar}
            \inferrule*[left=dce]
            {
            \G[\flt]\has\DELTA\ts b: \ty[q_1]{T_1}\ \EPS[1] \\
            \G[\flt]\has\DELTA\ts g: \ty[q_2]{T_2}\ \EPS[2] \\
            g \text{ terminates}  \\
            \EPS[1] = \PURE\ \text{or}\ \omega \\ \mute{\DEP\sqsubseteq\DELTA\vert_{\FX{\qsat{\EPS[1]}}}}
            }
            {
            \G[\flt]\has\DELTA \models \tlet~{x = b\has\DEP}~\tin~g \equivlog g: \ty[q_2]{T_2} \EPS[2]
            }

            \inferrule*[left=comm]
            {\G[\flt]\has\DELTA\ts  b_1: \ty[q_1]{T_1}\ \EPS[1] \qquad
            \G[\flt]\has\DELTA\ts b_2: \ty[q_2]{T_2}\ \EPS[2] \\
            \csxs[\flt, x, y]{\G\ ,\ x: \ty[\qsat{q_1} \overlap \qsat{\flt}]{T_1}, \ y: \ty[\qsat{q_2}\overlap \qsat{(\flt, x)}]{T_2}}\has\mute{\DELTA,(\FX{\qsat{\EPS[1]}},x)\mapsto x,(\FX{\qsat{\EPS[2]}},y)\mapsto y} \ts g: \ty[q]{T} \ \EPS[3] \\
            \FX{\qsat{\EPS[1]} \overlap \qsat{\EPS[2]} = \emptyset} \\
            x \not\in \FV(T) \\  y \not\in \FV(T) \\
            \theta = [q_2/y][q_1/x]\\\\ \mute{\DEP[1]\sqsubseteq\DELTA\vert_{\FX{\qsat{\EPS[1]}}}}\\\mute{\DEP[2]\sqsubseteq\DELTA\vert_{\FX{\qsat{\EPS[2]}}}}
            }
            {\G[\flt]\has\DELTA \models \phantom{\equivlog}\tlet~x = b_1\has\DEP[1]~\tin~ \tlet~ y = b_2\has\DEP[2] ~ \tin ~ g \qquad\qquad\qquad\\\phantom{\G[\flt]\models}\qquad\qquad\equivlog \tlet~ y = b_2\has\DEP[2]\ \tin \ \tlet \ x = b_1\has\DEP[1] ~ \tin ~ g \ : (\ty[q]{T} \ \EPS[1] \EFFSEQ \EPS[2] \EFFSEQ \EPS[3])\theta}

            \inferrule*[left=$\lambda$-hoist]
            { \G[\flt]\has\DELTA\ts  b: \ty[o]{S}\ \PURE \\
            \csxs[q, y, z]{\G\ ,\ y: \ty[p]{T}, z: \ty[\qsat{o}\overlap \qsat{(\flt, z)}]{S}} \has\mute{\pointsto{y},z\mapsto z} \ts g: \ty[r]{U} \ \EPS \\ \theta =[o/z]  \\ y \not\in \FV(U) \\ z \not\in \FV(U)\\
            }
            {\G[\flt]\has\DELTA \models \tlet~x = (\lambda y.(\tlet~z = b\has\NODEP~\tin~g)\has\DEP[1])\has\DEP[2]~\tin~ x \\\\
            \phantom{\G[\flt]\models}\qquad\qquad\equivlog \tlet~z = b\has\NODEP~\tin~\tlet~x = (\lambda y. g\has\DEP[1])\has\DEP[2]~ \tin ~ x \\ : ((y: \ty[p]{T}) \to^{\EPS} \ty[r]{U}\theta)^{q}
            }

            \inferrule*[left=$\beta$-inlining]
            {
            \csxs[q,x]{\G, x: \ty[p]{T}}\has\mute{\pointsto{x}} \ts g : \ty[r]{U} \EPS \\ q\subseteq \flt \\
            \G[\flt]\has\DELTA\ts b : \ty[p]{T}\ \PURE \\ x\notin\FV(U) \\
            \theta = [p/x]
            }
            {\G[\flt]\has\DELTA \models \tlet~x = (\lambda y.g\has\DEP[1])\has\DEP[2]~\tin~\tlet~z = b\has\DEP[2]~\tin~\tlet~w = y\;z\has\DEP[3]~\tin~w \\\\
            \equivlog \tlet~x = (\lambda y.g\has\DEP[1])\has\DEP[2]~\tin~\tlet~z = b\has\DEP[2]~\tin~ g\mute{[x\leadsto \DEP[3]]}[z/x] \\ : (\ty[r]{U}\ \EPS)\theta
            }

            \inferrule*[left = e-cse]
            {
            \G[\flt]\has\DELTA \ts b: \ty[p]{S} \ \EPS[1] \\
            \csxs[\flt, x, y]{\G\ ,\ x: \ty[\qsat{p} \overlap \qsat{\flt}]{S}, y: \ty[\qsat{p}\overlap \qsat{(\flt, x)}]{S}}\has\DELTA\mute{,(\FX{\qsat{\EPS[1]}},x)\mapsto x, (\FX{\qsat{\EPS[1]}},y)\mapsto y} \ts g: \ty[q]{T} \ \EPS[2] \\
            \omega \notin \EPS[1] \\ \theta = [x/y]\\\mute{\DEP\sqsubseteq\DELTA\vert_{\FX{\qsat{\EPS[1]}}}}
            }
            {
            \G[\flt]\has\DELTA \models (\tlet\ x = b\has\DEP \ \tin \ \tlet \ y = b\has\DEP \ \tin \ g) \equivlog (\tlet \ x = b\has\DEP\ \tin \ g\theta) : \ty[q]{T}\theta \ \EPS[1]\theta \EFFSEQ\EPS[2] \\
            }
        \end{mathpar}
        \caption[Equational and optimization rules for the \langg{} graph IR.]{Equational rules for the \langg{} graph IR. We obtain the optimization rules by congruence closure with the contexts from \Cref{fig:graphir_context}.}
        \label{fig:graphir_equiv}
    \end{mdframed}
\end{figure*} 
From the above results about logical equivalence for the \langg graph IR obtained by a ``round-trip
translation'' technique, we are now equipped to prove the soundness of the main paper's optimization
rules based on the results of \Cref{sec:direct_equiv}. The optimization rules are the congruence
closure (with respect to \(C\) contexts in \Cref{fig:graphir_context}) of the equations shown in
\Cref{fig:graphir_equiv}. Those are obtainable mechanically from their counterparts in \directlang{}
(\Cref{fig:direct_equiv}) by using the type-and-effect preserving translation into MNF
(\Cref{sec:mnf:translation}) followed by dependency synthesis for a given map \(\DELTA\) of last
uses (\Cref{fig:graphir:synthesis}). That is, the equational rules in \Cref{fig:graphir_equiv} are
the dependency-annotated versions of their counterparts in MNF.

\begin{theorem}[Compatibility of the Equational Rules]\label{thm:graphir_equational_admissible}
    Each rule in \Cref{fig:graphir_equiv} is compatible with the logical equivalence.
\end{theorem}
\begin{proof} Each individual rule \rulename{dce}, \rulename{comm}, \rulename{$\lambda$-hoist}, \rulename{$\beta$-inlining},
    and \rulename{e-cse} can be uniformly proved as follows:
    \begin{enumerate}
        \item By dependency erasure, and by \Cref{lem:type_preservation:translation_backwards}, we have
              that both graphs are equated by the direct-style version of the respective rule (\Cref{fig:direct_equiv}).
        \item By \Cref{lem:dce,lem:comm,lem:lambda_hoist,lem:beta_inlining,lem:e_cse}, the erased graphs are logically equivalent
        in \directlang.
        \item By totality and soundness of synthesis (\Cref{lem:hardsoft:synth-totality,lem:hardsoft:synth-soundness}),
              re-synthesizing the dependencies under the given \(\DELTA\) and context \(\G[\flt]\) of the erased graphs
              yields the initial dependency-annotated versions.
        \item By (2) and (3), both sides are logically equivalent in \langg{} (\Cref{def:graphir:log_equiv}).
    \end{enumerate}
\end{proof}

\begin{corollary}[Compatibility of the Optimization Rules]\label{coro:graphir_optimization_admissible}
    The optimization rules for \langg{}, \ie, the congruence closure of the rules in \Cref{fig:graphir_equiv} is compatible in logical equivalence.
\end{corollary}
\begin{proof}
    By \Cref{thm:graphir_equational_admissible} and the congruency \Cref{lem:graphir_congruence}.
\end{proof}

\begin{corollary}[Soundness of the Optimization Rules]\label{coro:graphir_optimization_soundness}
    The optimization rules for \langg{}, \ie, the congruence closure of the rules in \Cref{fig:graphir_equiv} describe contextually equivalent graphs.
\end{corollary}
\begin{proof}
    By \Cref{coro:graphir_optimization_admissible} and soundness \Cref{thm:graphir_lr_soundness}.
\end{proof}
\section{From Graphs Back to Trees} \label{sec:codemotion}

In this section, we discuss efficient algorithms and heuristics for code generation,
which transform the graph informed by \langa{} into a tree with nested structures.
This section \emph{follows the same structure as in the main paper~\cite{oopsla23} but explains
the algorithms with greater details}:
(1) we first present the basic scheduling algorithm (\Cref{sec:basic}) which
incorporates dead code elimination;
(2) based on that, a lightweight frequency estimation heuristics (\Cref{sec:freq})
introduces more flexible code motion;
(3) we finally present a compact scheduling algorithm (\Cref{sec:compact}) for
instruction selection and inlining expressions.
The optimizations are justified by the equational theory in \Cref{sec:optimizations:soundness}.

\begin{figure}[t]
\begin{mdframed}
\begin{minipage}[t]{1.0\textwidth}
\begin{lstlisting}[xleftmargin=4.0ex,numbers=left,basicstyle=\scriptsize\selectfont\ttfamily]
/* auxiliary functions to access different sorts of dependencies of a node
   boundDeps: dependencies that are bound variables
   dataDeps, effDeps: data- and effect-dependencies
   hardDeps ($\subseteq\,$effDeps): hard effect-dependencies */
val boundDeps, dataDeps, effDeps, |\HLCode[light-pink]{hardDeps}|: Node => Set[Node]
/* obtain estimated frequencies of data-/effect-dependencies of a node: */
|\HLCode{\textbf{val} depFreq: Node => Map[Node, Double]}|

/* traverse a single node to emit a tree node */
def traverseNode(inner: Set[Node], path: Set[Node], n: Node): TreeNode = n match
  case $\lambda$f(x).r =>                                              // schedule nodes into a $\lambda$ scope
    TreeNode.Scope($\lambda$f(x), scheduleBlock(inner, path $\cup$ {f, x}, r))
  ...
  case "|\$|sym := |\$|op(|\$|args)" =>                                // schedule common nodes as leaves
    TreeNode.Leaf(sym, Exp(op, args))

/* schedule a block given its final result, producing a scoping tree */
def scheduleBlock(scope: Set[Node], path: Set[Node], res: Node): List[TreeNode] =
  val reachable: PriorityQueue[Node] = {res}                  // reachable nodes, topologically ordered
  |\HLCode[light-pink]{\textbf{val} reachableHard: Set[Node] = \{res\}}|                       // reachable |\&| required by data/hard deps.
  |\HLCode{\textbf{val} reachableHot: Set[Node] = \{res\}}|                        // reachable |\&| frequently executed
  val current: List[Node] = $\varnothing$                                  // scheduled in current block
  val inner:   Set[Node]  = $\varnothing$                                  // scheduled in inner blocks
  def available(n: Node): Boolean = boundDeps(n) $\subseteq$ path        // available: bound vars. in deps. are ready

  for n $\leftarrow$ reachable do
    |\HLCode[light-pink]{\textbf{if} reachableHard(n) \textbf{then}}\,|                                 // reachable via data/hard dependencies
      if |\HLCode{reachableHot(n) $\land$}| available(n) then                 // reachable via hot paths
        current = n :: current
        |\HLCode{\textbf{for} m $\leftarrow$ (dataDeps(n) $\cup$ effDeps(n)) $\cap$ scope \textbf{do}}|      // consider deps. hot if freq > 0.5
          |\HLCode{\textbf{if} depFreq(n)[m] > 0.5 \textbf{then} reachableHot += \{m\}}|
      else                                                    // only via cold path, or hot but unavailable
        inner += {n}
        |\HLCode{\textbf{if} reachableHot(n) \textbf{then}}|                              // deps. of unavailable hot nodes are hot
          |\HLCode{reachableHot += (dataDeps(n) $\cup$ effDeps(n)) $\cap$ scope}|
      |\HLCode[light-pink]{reachableHard += (dataDeps(n) $\cup$ hardDeps(n)) $\cap$ scope}|  // reach by data and only hard dependencies
    reachable += (dataDeps(n) $\cup$ effDeps(n)) $\cap$ scope             // reach by data/effect dependencies

  for n $\leftarrow$ current yield traverseNode(inner, path, n)           // recursively build up the scoping tree
\end{lstlisting}
\end{minipage}
\caption[The pseudocode of the basic scheduling algorithm with two extensions.]{The pseudocode of the basic scheduling algorithm with two extensions.
  Function \lstinline|scheduleBlock| decides which nodes are scheduled
  into the \lstinline|current| scope and recursively schedules \lstinline|inner| scopes.
  To generate code for a graph \lstinline|g|,
  we make the call \lstinline|scheduleBlock(g.nodes, $\,\emptyset$, g.result)|.
  The extension to eliminate dead code by soft dependencies (cf. \Cref{sec:basic}) is marked in \HLCode[light-pink]{\scriptsize pink},
  and the extension for frequency estimation and code motion (cf. \Cref{sec:freq}) in \HLCode{\scriptsize teal}.
  }
\label{fig:basic-codemotion}
\end{mdframed}
\end{figure}

\subsection{Traversal without Redundant Code} \label{sec:basic}

In essence, the block scheduling algorithm traverses a hierarchy of
graph-represented blocks and selects unscheduled nodes to move into
tree-represented blocks (and emit code) based on the node's dependencies.
\Cref{fig:basic-codemotion} shows the vanilla block scheduling algorithm,
which operates over graph IR data structures.
We use @Node@ to represent a graph node, and
a few auxiliary functions such as @dataDeps@/@hardDeps@ to extract
different kinds of dependencies of a @Node@.

At the beginning of scheduling, we have a scope of @Node@s to schedule and a
symbol representing the final result of the top-level block.
Transitively following the dependencies of the final result, @scheduleBlock@
partitions the unscheduled nodes into two
groups: (1) nodes that are scheduled in the current block, and
(2) nodes that will be scheduled into other inner blocks.
This process is recursively applied when encountering lambda nodes
in @traverseNode@.
Schedule decisions are made relying on two properties over nodes,
\emph{available} and \emph{reachable}.

\bfparagraph{Nested Scopes}
A node is \emph{available} if all its dependent bound variables
have been introduced under the current path.
A node is scheduled in the @current@ block if it is both reachable and
available (Line 29); otherwise, it is moved to the @inner@ scope (Line 33).
To this end, we need to \emph{transtively} compute the bound variables depended
by nodes, of which the result is reflected by @boundDeps@.
In \Cref{fig:basic-codemotion},
@path@ represents the set of the accumulated bound variables (\eg,
introduced by lambdas) up to the current block.
Both @path@ and @scope@ need maintaining through the recursive calls
to properly handle scopes and nested blocks.

\bfparagraph{Dead Code Elimination}
A node is \emph{reachable} if it can be back-tracked from the current result node
through effect or data dependencies.
Only reachable nodes are considered for scheduling, which naturally eliminates
dead code (cf. rule \textsc{dce}, \Cref{fig:graphir_equiv}).
In \Cref{fig:basic-codemotion}, @reachable@ is a priority queue
which reflects the property and enforces the topological ordering.
It is populated with data and effect dependencies along the iteration (Line 37).
As an extension, we discern soft dependencies (\Cref{sec:hardsoft}) and identify data
and hard dependencies as \HLCode[light-pink]{\scriptsize\texttt{reachableHard}} (Line 36).
This ensures nodes that are only reachable via soft dependencies can be eliminated.

\bfparagraph{Complexity}
Given the total number of nodes $n$ and the maximal depth of nested blocks $k$, the
worst-case asymptotic time complexity is bound by $O(kn^2)$.
This is because the algorithm traverses over the reachable nodes in order
(bound by $n$, $O(n)$ each),
and repeats this process for nested scopes.
In practice, the complexity is bound by $O(kn\log n)$ given the decreasing
size of nested scopes and the limited degrees of graph nodes.
To exemplify, scheduling of symbolic execution (cf.\ \cite{oopsla23}, Section 7.2)
takes 19.3 sec for 548,976 graph nodes, which is rather efficient.

\subsection{Code Motion by Frequency Estimation}\label{sec:freq}

Our basic scheduling algorithm eagerly schedules nodes to their outermost
block, following the equational rules \inflabel{comm} and \inflabel{$\lambda$-hoist}
in \Cref{fig:graphir_equiv}.
This is a form of code motion for no extra effort and
generally desirable for functions and loops. For instance,
\begin{lstlisting}
  List(1, 2, 3, 4, 5).map(x => x * factorial(N))
\end{lstlisting}
Lifting the expensive @factorial@ out of the lambda is beneficial and feasible
since it does not depend on the bound variable @x@. However, this does not
always generate optimal code.  Consider a conditional expression that
transforms an array of complex numbers only in the then-branch,
\begin{lstlisting}
  if (cnd) compNums.map(f) else compNums
\end{lstlisting}
Since @compNums.map(f)@ has no dependency on @cnd@, this statement would be
lifted to the outer scope and always executed regardless of the condition,
thus imposing unnecessary runtime overhead when the else-branch is actually taken.

To avoid this situation, we can \emph{estimate} how frequently a node is used
and move less frequent nodes (\ie, cold) into inner scopes.
We assign a number to each node based on its dependents which
represents how relatively often the node is executed at runtime.
The results of functions and loops are assigned 100,
indicating that they and their dependencies can be executed multiple times (definitely hot).
The results of conditional branches are assigned 0.5 (cold),
assuming each branch is taken with equal probability.
All other nodes are assigned 1.0 (normal).
Numbers above are illustrative and context-insensitive.
Alternative metrics are possible, while what we present here is beneficial to
many code patterns. %

\Cref{fig:basic-codemotion} highlights in \HLCode{teal} the changes to the basic
scheduling algorithm to use frequency estimation. Given a node ready to schedule
in the current scope (Line 28-31), we use the function @depFreq@ to access the frequency
estimation of its dependencies. Only
those with frequencies greater than 0.5 are considered hot-reachable, and thus
can be included in @current@.
Others are classified to be cold, and are scheduled in inner blocks if all
reaching dependencies are cold.
Given a node scheduled in inner blocks (Line 32-35), its dependencies are considered to have
the same level of warmth as the node itself, ensuring consistent code motion
behavior for code with nested scopes.

The proposed heuristic works as expected in that it
(1) lifts computation out of hot constructs such as loops, and
(2) sinks computation into cold constructs such as conditionals.
Regarding nested scopes, the heuristic prioritizes (1) over (2). Suppose there
is a loop in the current scope and a conditional inside that loop. For a
node inside the conditional, the heuristic tends to lift it out to the current
scope, as the result is still used multiple times during the loop.
On the other hand, if there is a loop inside a conditional in the current scope,
a node used in the loop is not lifted to the current scope, but still subject to
lifting up to the conditional scope to optimize the loop.

Our approach is simpler to implement and more efficient compared to more
sophisticated analyses, such as lazy code motion \cite{Knoop1992}, partial
redundancy elimination \cite{Kennedy1999}, or even whole-program dataflow analyses.
The estimation heuristics associates a constant factor to each node traversed.
Therefore, it does not change the complexity of the basic scheduling algorithm.

\begin{figure}[t]
\begin{mdframed}
\begin{minipage}[t]{1.0\textwidth}
\begin{lstlisting}[xleftmargin=4.0ex,numbers=left,basicstyle=\scriptsize\selectfont\ttfamily]
/* shouldInline is a mutable set of nodes that are (1) locally defined,
   (2) locally used as value exactly once, and (3) have no inner use. */
val shouldInline: Set[Node] = { n$\,\in\,$localDef $|$ currentValUse[n] == 1 $\land$ innerValUse[n] == 0 }

/* seen is the set of processed nodes (cf. processNodeHere) */
val seen: Set[Node] = $\varnothing$

/* check if all successors of node n have been processed; exclude it if not. */
def checkInline(n: Node): Unit =
  if shouldInline(n) $\land$ $\forall\,$s$\,\in\,$succ[n].$\,$seen(s) then
    processNodeHere(n)
  else shouldInline -= {n}

/* considering node n, try to inline all its dependencies */
def processNodeHere(n: Node): Unit =
  seen += {n}
  for s $\leftarrow$ dataDeps(n).reverse do                         // respects order of argument evaluation
    checkInline(s)

/* traverse and emit a single node */
def traverseNode(inner: Set[Node], path: Set[Node], |\HLCode[light-yellow]{inlined: Set[Node],}| n: Node): TreeNode =
  def traverseInlined(n: Node): Exp = n match
    case "|\$|sym := |\$|op(|\$|args)" =>                        // recursively handling inlined expressions
      Exp(op, for m $\leftarrow$ args yield |\HLCode[light-yellow]{\textbf{if} inlined(m) \textbf{then} traverseInlined(m)}| else m)
    ...

  n match
    case $\lambda$f(x).r =>                                      // schedule nodes into a $\lambda$ scope
      TreeNode.Scope($\lambda$f(x), scheduleBlock(inner, path $\cup$ {f, x}, r))
    ...
    case "|\$|sym := |\$|op(|\$|args)" =>                        // schedule common nodes as leaves
      TreeNode.Leaf(sym, |\HLCode[light-yellow]{traverseInlined(n)}|)
\end{lstlisting}
\end{minipage}
\caption[Auxiliary definitions for compact code generation.]{
  The auxiliary functions for the block scheduling algorithm with
  compact traversal.
  \lstinline|shouldInline|, \lstinline|seen|, \lstinline|checkInline|, and \lstinline|processNodeHere|
  are defined within \lstinline|scheduleBlock| in \Cref{fig:compact-codemotion},
  while \lstinline|traverseNode| replaces its former version in \Cref{fig:basic-codemotion}.
  Usages of inlining information are marked in \HLCode[light-yellow]{\scriptsize yellow}.
  }
\label{fig:aux-compact-codemotion}
\end{mdframed}
\end{figure}

\begin{figure}[t]
\begin{mdframed}
\begin{minipage}[t]{1.0\textwidth}
\begin{lstlisting}[xleftmargin=4.0ex,numbers=left,basicstyle=\scriptsize\selectfont\ttfamily]
def scheduleBlock(scope: Set[Node], path: Set[Node], res: Node): List[TreeNode] =
  ...                                                      // see |\Cref{fig:basic-codemotion}| for initial definitions

  /* (1) a backward pass to track node usages, integrated with the basic scheduling algorithm */
  val localDef: Set[Node] = $\varnothing$                               // the collection of local definitions
  val innerValUse: Map[Node, Int] = {$\underline{\quad}\,\mapsto$ 0}                   // how often a symbol is used in inner blocks
  val currentValUse: Map[Node, Int] = {res $\mapsto$ 1, $\underline{\quad}\,\mapsto$ 0}         // how often a symbol is used in current block
  for n $\leftarrow$ reachable do
    |\HLCode[light-pink]{\textbf{if} reachableHard(n) \textbf{then}}|
      if |\HLCode{reachableHot(n) $\land$}| available(n) then
        current = n :: current
        |\HLCode[light-yellow]{localDef += \{n\}}|                                   // record local definition
        |\HLCode[light-yellow]{\textbf{for} m $\leftarrow$ dataDeps(n) $\cap$ scope \textbf{do}}|                  // tracking node usage
          |\HLCode[light-yellow]{\textbf{if} depFreq(n)[m] = 1.0 \textbf{then} currentValUse[m]++}|
          |\HLCode[light-yellow]{\textbf{else} innerValUse[m]++}|
        |\HLCode{\textbf{for} m $\leftarrow$ (dataDeps(n) $\cup$ effDeps(n)) $\cap$ scope \textbf{do}}|
          |\HLCode{\textbf{if} depFreq(n)[m] > 0.5 \textbf{then} reachableHot += \{m\}}|
      else
        inner += {n}
        |\HLCode[light-yellow]{\textbf{for} m $\leftarrow$ dataDeps(n) $\cap$ scope \textbf{do}}|                  // tracking node usage
          |\HLCode[light-yellow]{innerValUse[m]++}|
        |\HLCode{\textbf{if} reachableHot(n) \textbf{then}}|
          |\HLCode{reachableHot += (dataDeps(n) $\cup$ effDeps(n)) $\cap$ scope}|
      |\HLCode[light-pink]{reachableHard += (dataDeps(n) $\cup$ hardDeps(n)) $\cap$ scope}|
    reachable += (dataDeps(n) $\cup$ effDeps(n)) $\cap$ scope

  /* (2) a forward pass to compute local successors */
  val succ: Map[Node, Set[Node]] = {$\underline{\quad}\,\mapsto\,\varnothing$}
  for c $\leftarrow$ current do
    for m $\leftarrow$ (dataDeps(c) $\cup$ effDeps(c)) $\cap$ localDef do
      |\HLCode[light-yellow]{succ[m] += \{c\}}|

  /* (3) a backward pass to check if all successors are emitted after the point of inlining */
  val shouldInline: Set[Node] = ...                        // see |\Cref{fig:aux-compact-codemotion}| for definitions
  def checkInline(n: Node): Unit = ...
  def processNodeHere(n: Node): Unit = ...
  |\HLCode[light-yellow]{checkInline(res)}|                                        // inline the result node into "return" statement
  for n $\leftarrow$ current.reverse do                                // process all possible inline locations
    if !shouldInline(n) then |\HLCode[light-yellow]{processNodeHere(n)}|

  /* (finally) a forward pass to perform the acutal code emission */
  for n $\leftarrow$ current $\HLBox[light-yellow]{\texttt{\textbf{if} !shouldInline(n)}}$                                          // emit each non-inlinable node
    yield traverseNode(inner, path, $\HLBox[light-yellow]{\texttt{shouldInline}}$, n)                   // with inlining information (cf. |\Cref{fig:aux-compact-codemotion}|)
\end{lstlisting}
\end{minipage}
\caption[Compact code generation.]{
  The pseudocode of block scheduling algorithm with compact traversal.
  Changes for elimination by soft dependencies are highlighted in \HLCode[light-pink]{\scriptsize pink},
  frequency estimation in \HLCode{\scriptsize teal}, and compact traversal in \HLCode[light-yellow]{\scriptsize yellow}.
  Auxiliary function definitions can be found in \Cref{fig:aux-compact-codemotion}.
  }
\label{fig:compact-codemotion}
\end{mdframed}
\end{figure}

\subsection{Instruction Selection with Compact Traversal}\label{sec:compact}

The basic scheduling algorithm generates code that binds
every intermediate expression using \textsf{let}.
The result, nevertheless, is not only verbose but also suboptimal, without using
the target-specific primitives.
Consider the following tensor computation snippet,
\begin{lstlisting}
  val X = Matmul(A, B); val C = Add(C, X); C
\end{lstlisting}
The unique use of @X@ in @Add@ enables further transformation into
destination-passing style using generalized matrix multiplication @GEMM@,
which updates @C@ in-place by
$C \leftarrow \alpha AB + \beta C$.
Thus, we can match the tree structure @Add(C, Matmul(A, B))@ and generate
a single operation, eliminating the intermediate multiplication @X@:
\begin{lstlisting}
  GEMM(A, B, C, alpha=1.0, beta=1.0); C
\end{lstlisting}
This is basically a form of \emph{instruction selection} seen in optimizing
compilers. However, the procedure can be non-trivial on computation
graphs where all consumers of a value need accounting for. A proper solution
on graphs like LLVM's \texttt{SelectionDAG}~\cite{DBLP:conf/cgo/LattnerA04}
takes effort to compose and time to execute.

For the \langg{} graph IR, we perform a simple but highly useful alternative called \emph{compact
traversal}: we first turn the graph nodes into inlined trees whenever possible,
and then use tree matching algorithms (\eg, maximal munch) to select
the best primitive.
Compact traversal must respect the dependencies to preserve the semantics. Consider the
following code that reads the value from cell @x@ and then increments the value of @x@:
\begin{lstlisting}
  if (cond) { val y = !x; inc(x); println(y) }
\end{lstlisting}
Since @inc(x)@ depends on node @y@ in an effectful way, inlining @y@ to @println@ breaks the
semantics.

Compact traversal works on each @current@ scope determined by the basic scheduling
algorithm (cf. \Cref{fig:basic-codemotion}). Initially, all nodes in the scope are
viewed as individual trees.
\Cref{fig:compact-codemotion} shows the updated main function @scheduleBlock@
with compact traversal atop other extensions, and
\Cref{fig:aux-compact-codemotion} shows accompanying auxiliary definitions.
In the following, we summarize the compact traversal algorithm in
three steps.

\begin{enumerate}[wide=\parindent,itemsep=.5ex]
\item \textbf{Track Node Usage.}
The key idea of compact traversal is that we only consider inlining for nodes
that are locally defined and locally used exactly once, and are not used in
nested scopes.
To this end, we first identify all candidate nodes for inlining by tracking
node usage.
When scheduling code, we keep track of local definitions and their symbolic
names (Line 12).
We also track how many times a locally defined node is used in the current
scope and inner scopes as a proper value (Line 13-15, 20-21), respectively.

\item \textbf{Compute Local Successors.}
After recording the node usage information in the first step, we need to
calculate local successors of a node in the current block (Line 28-31).
A node @x@ is a local successor of node @y@ if @x@ and @y@ are scheduled in the
same block and @x@ depends on @y@.
The algorithm uses a map to store the local successors of a locally defined node.

\item \textbf{Check Inlining.}
Lastly, the algorithm runs a backward pass for all nodes that can be inlined.
This pass checks if all successors of a node are emitted after the point
considered for inlining (\Cref{fig:aux-compact-codemotion}, Line 9-12).
If not, we disable its inlining. Otherwise, we inline
this node and check if any other nodes used by this node can be further inlined
(\Cref{fig:aux-compact-codemotion}, Line 15-18).
\end{enumerate}

\bfparagraph{A Flexible Framework for Optimization Opportunities} In \Cref{sec:codemotion}
we have illustrated a series of optimizations possible by simple and composable
algorithms through scheduling graph IR back to nested tree representations.
Opportunities are not limited to the ones aforementioned. For instance, the basic
scheduling algorithm can introduce \emph{instruction scheduling} by assigning a proper
priority value to each node reflecting not only dependency but also timing, thus
tweaking the behavior of the traversal. To conclude, graph IR can be an efficient
and flexible system to generate performant code.
 
\begin{acks}
  A large number of people have contributed to the design
  of the LMS graph IR over the years
  \cite{DBLP:journals/corr/abs-1109-0778,
  DBLP:phd/ch/Rompf12,
  DBLP:journals/cacm/RompfO12,
  DBLP:conf/IEEEpact/BrownSLRCOO11,
  DBLP:conf/popl/RompfSABJLJOO13,
  DBLP:conf/ecoop/SujeethRBLCPWPJOO13,
  DBLP:conf/fpl/GeorgeLNRBSOOI14,
  DBLP:journals/micro/LeeBSCROO11,
  DBLP:journals/tecs/SujeethBLRCOO14,
  DBLP:conf/snapl/RompfBLSJAOSKDK15,
  DBLP:conf/scala/Rompf16,
  DBLP:conf/cgo/BrownLRSSAO16,
  DBLP:conf/popl/AminR17a,
  DBLP:conf/gpce/OfenbeckRP17,
  DBLP:conf/sigmod/TahboubER18,
  DBLP:conf/osdi/EssertelTDBOR18,
  DBLP:conf/gpce/StojanovRP19,
  DBLP:conf/mlsys/MoldovanDWJLNSR19,
  DBLP:conf/bigdataconf/WangCZR19,
  DBLP:journals/pacmpl/WangZDWER19,
  DBLP:conf/gpce/EssertelTR21,
  DBLP:conf/icse/WeiJGDTBR23}, 
  and in particular
  to the system of effect dependencies:
  Nada Amin, Thaïs Baudon,
  Kevin J. Brown, James M. Decker, Grégory Essertel,
  Georg Ofenbeck, Alen Stojanov,
  Arvind K. Sujeeth,
  Fei Wang, and Yushuo Xiao.
  This work was supported in part by NSF awards 1553471, 1564207, 1918483, 1910216, DOE award DE-SC0018050, as well as gifts from Meta, Google, Microsoft, and VMware.
\end{acks}

\bibliography{references}

\end{document}